\documentclass[12pt]{article}
\usepackage{authblk}
\usepackage{geometry}
\usepackage{amsthm}
\usepackage{mathrsfs}
\usepackage{customcommands}
\usepackage{bm}
\usepackage{Chrisstyle}

\usepackage{amsmath}
\usepackage{amssymb}
\usepackage{amsfonts}
\usepackage{braket}
\usepackage[normalem]{ulem}
\usepackage{color}
\usepackage{ulem}
\usepackage{mathtools}
\usepackage{tensor} 
\usepackage{overpic}
\usepackage{subcaption}
\usepackage{float}
\usepackage[export]{adjustbox}
\usepackage{longtable}
\newtheorem{theorem}{Theorem}

\DeclareMathOperator{\id}{id}

\graphicspath{{./fig/}}
\newtheorem{corollary}[theorem]{Corollary}
\newtheorem{lemma}[theorem]{Lemma}
\theoremstyle{definition}
\newtheorem{definition}{Definition}[section]

\theoremstyle{remark}
\newtheorem*{remark}{Remark}

\title{Reflected entropy in random tensor networks III: triway cuts}

\author[1]{Chris Akers,}
\author[2]{Thomas Faulkner,}
\author[3]{Simon Lin,}
\author[4]{Pratik Rath}

\affiliation[1]{Department of Physics and Center for Theory of Quantum Matter,\\
University of Colorado Boulder, Boulder, CO 80309, USA}
\affiliation[2]{Department of Physics, University of Illinois, \\
1110 W. Green St., Urbana, IL 61801-3080, USA}
\affiliation[3]{New York University Abu Dhabi, P.O. Box 129188, Abu Dhabi, United Arab Emirates}
\affiliation[4]{Center for Theoretical Physics and Department of Physics,
University of California, Berkeley, CA 94720, USA}

\emailAdd{chris.akers@colorado.edu}
\emailAdd{tomf@illinois.edu}
\emailAdd{simonlin@nyu.edu}
\emailAdd{pratik\_rath@berkeley.edu}

\abstract{
For general random tensor network states at large bond dimension, we prove that the integer R\'enyi reflected entropies (away from phase transitions) are determined by minimal triway cuts through the network. This generalizes the minimal cut description of bipartite entanglement for these states. A natural extrapolation away from integer R\'enyi parameters, suggested by the triway cut problem, implies the holographic conjecture $S_R=2EW$, where $S_R$ is the reflected entropy and $EW$ is the entanglement wedge cross-section. Minimal triway cuts can be formulated as integer programs which cannot be relaxed to find a dual maximal flow/bit-thread description.  This sheds light on the gap between the existence of tripartite entanglement in holographic states and the bipartite entanglement structure motivated by bit-threads. In particular, we prove that the Markov gap that measures tripartite entanglement is lower bounded by the integrality gap of the integer program that computes the triway cut.}

\begin{document}

\maketitle
\flushbottom

\section{Introduction}

There is compelling evidence that quantum gravity should be thought of as an emergent phenomenon with the underlying geometry being described by the entanglement structure of a dual/holographic wave-function. The Ryu-Takayanagi (RT) formula \cite{Ryu:2006bv,Ryu:2006ef,Hubeny:2007xt} expresses this idea, connecting areas of minimal surfaces and von Neumann entropies: a measure of bipartite entanglement in pure states. Minimal surfaces can be equivalently described by bit-threads \cite{Freedman:2016zud,Headrick:2022nbe}, maximal divergence free locally bounded flows between two boundary regions. Furthermore, these bit-threads give a vivid picture of (pure state) bipartite entanglement in the boundary wavefunction, with a thread corresponding to a distillable EPR pair.

Random tensor network states \cite{Collins:2010fsu,Hayden:2016cfa} model this behavior with a graph $G = (E,V)$ playing the role of the underlying geometry and boundary vertices $\partial \subset V$ containing the dual Hilbert space. Edges $e \in E$ contain maximally entangled states with large bond dimensions $\chi(e)$ and vertices are described by randomly chosen tensors that are contracted with the edge states. The RT formula arises as a minimal cut through the graph with edges weighted by $w(e) \propto \ln \chi(e)$. The cut divides the vertices $V$ into the two disjoint sets each containing the corresponding sets of boundary vertices whose von Neumann entropy we wish to compute. Bit-threads correspond to dual maximal flows between the two sets of boundary vertices with flow capacities set by $w(e)$.  This correspondence is a version of the max-flow min-cut theorem that can be proven using strong duality theorems from the theory of linear/convex programs. 

Bit-threads however tend to give a misleading picture of multipartite entanglement. For example, a bipartite dominance conjecture was formulated for three party holographic states based on the existence of such bit-thread configurations \cite{Cui:2018dyq}. However, the conjecture of Ref.~\cite{Cui:2018dyq} contradicts other measures of tripartite entanglement beyond von Neumann entropies, most notably for this work, that of reflected entropy. Ref.~\cite{Akers:2019gcv} argued using the reflected entropy that holographic states have large amounts of tripartite entanglement and thus, disproved a version of the bipartite dominance conjecture. 
Generally, minimal areas are only a limited probe of the underlying geometry, and one might expect other geometric objects -- such as surfaces of various co-dimensions -- to play an important role in a putative correspondence between geometry and quantum information. For example, computational complexity is believed to be associated to co-dimension $0$ or $1$ regions in spacetime \cite{Susskind:2014rva,Brown:2015bva}. There are now also hints that a class of tripartite entanglement should be associated to spatial co-dimension $2$ objects \cite{Hayden:2021gno}. In this paper, we find further evidence for the latter by proving a correspondence between reflected entropy and a minimal triway cut. Triway cuts generalize the bipartite cuts described above, and are likely the closest graph analog of a co-dimension $2$ object that is defined for any graph.\footnote{The cuts themselves are co-dimension $1$, however the three cuts meet at some locus that might be considered co-dimension $2$.}

Triway cuts are integer optimization programs that cannot be dualized to a bit-thread description.\footnote{What we mean here is that the (Lagrange) dual flow programs do not have the same optimal value, i.e., there is a \emph{duality gap}. However, it is possible to find a ``dual'' if one considers more exotic optimization problems. See \secref{sec:disc} for more discussion on this topic.}
Relaxing the integer constraint gives a linear program that underestimates the cut. The ratio between these values, the output of the integer program over that of the linear program, is generally a difficult quantity to compute, and is called the \emph{integrality gap}.
In fact, computing the integrality gap is an NP-complete problem \cite{doi:10.1137/S0097539792225297}.
 
We now introduce our main result in more detail. 
The reflected entropy $S_R$ \cite{Dutta:2019gen} of a state $\rho_{AB}$ is defined as the entropy of $AA^\star$ in the canonical purification $\big| \rho_{AB}^{1/2} \big> \in \mathcal{H}_{AA^\star BB^\star}$.
The R\'enyi generalization of $S_R$ is a simple one parameter family of quantum information measures: 
\begin{equation}
S^{(n)}_{R}(A:B) = -\frac{1}{n-1} \ln {\rm Tr} \rho_{AA^\star}^n \qquad \rho_{AA^\star} = {\rm Tr}_{BB^\star} \big| \rho_{AB}^{1/2} \big> \big< \rho_{AB}^{1/2} \big|
\end{equation}
We will prove:

\begin{theorem}
\label{thm:1}
For integer $n >  1$, the reflected entropy of a random tensor network state at large bond dimension, with a unique entanglement wedge for $AB:C$ and a unique triway cut for $A:B:C$ (with tensions specified below), satisfies:
\begin{equation}
\lim_{\chi \rightarrow \infty} \frac{\overline{ S_{R}^{(n)}(A:B) } }{\ln \chi} = \frac{1}{n-1}\mathcal{A}_{\mathbf{t}}(A:B:C) - \frac{n}{n-1} \mathcal{A}(AB:C)
\end{equation}
where $ \mathcal{A}_{\mathbf{t}}(A:B:C)$ is the area of a multiway cut with tensions $\mathbf{t}\equiv(t_{A:B},t_{B:C},t_{C:A}) = (2(n-1),n,n)$
\footnote{In this paper we mostly consider triway cuts defined with this tension. For this reason, we will often abbreviate the triway cuts simply as $\mathcal{A}(A:B:C)$ when it is unambiguous.}
and $\mathcal{A}(AB:C)$ is the area of the minimal cut (with tension $1$) for $AB:C$. 
\end{theorem}

Averaging is taken with respect to the Haar measure over unitary matrices that are applied to vertex states in the graph. See Definition~\ref{def:rtn} for the precise construction. The triway cut is defined in the same way as a cut: we split the vertices into three disjoint subsets  containing respectively boundary vertices $A,B,C$. The area is then the sum over the edges $e$ which intersect two of the three regions, weighted by the respective tensions and $w(e)$.  See Definition~\ref{def:tri-cut} and \figref{fig:network-multi-cut}.

\begin{figure}[h!]
\centering
\includegraphics[scale=1]{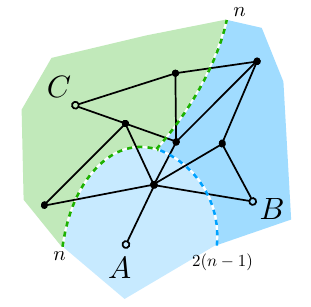} \hspace{2cm} \includegraphics[scale=1]{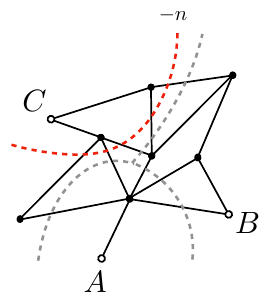}
\caption{ \label{fig:network-multi-cut} (\emph{left}): A network with three boundary vertices (open dots) $A,B,C$, showing the definition
of a triway cut, allowing for different domain wall tensions.  This figure is schematic: edges in the graphs are weighted by some unspecified bond dimensions. (\emph{right}): To compute the reflected entropy we must subtract the minimal cut for $AB:C$ which is always contained inside  the $C$ domain as shown. We should also divide the areas by $(n-1)$.
 }
\end{figure}

The Markov gap is defined as  \cite{Hayden:2021gno}:
\begin{equation}
h(A:B) = S_R(A:B) - I(A:B) \geq 0
\end{equation}
where $I(A:B)$ is the mutual information. The Markov gap vanishes iff the three party state has a particular structure: sum of triangle states with only bipartite entanglement between the three different parties \cite{Zou:2020bly}:
\begin{equation}
\ket{\psi}_{ABC} = \sum_{j} \sqrt{p_j}\ket{\psi_j}_{A_R^j B_L^j} \ket{\psi_j}_{B_R^j C_L^j}\ket{\psi_j}_{C_R^j A_L^j}.  
\end{equation}
The Markov gap thus detects a certain class of non-trivial tri-partite entanglement.\footnote{$h$ vanishes on GHZ states, so it does not detect all kinds of tripartite entanglement \cite{Akers:2019gcv}. A refined version based on the entanglement of purification does better \cite{Takayanagi:2017knl,Nguyen:2017yqw}. This is generally harder to compute but the results of this paper help compute it for a class of RTN states \cite{Akers:2023obn}.}

In particular, we have the following lower bound:
\begin{theorem} Under the uniqueness assumption for $n=2$ in Theorem~\ref{thm:1}, 
the (normalized) Markov gap ($MG$) of a random tensor network state at large bond dimension is lower bounded by:
\begin{equation}
\label{eq:MG-bound}
MG \equiv \lim_{\chi \rightarrow \infty} \frac{\overline{h(A,B)}}{\ln \chi} \geq 2 \mathcal{A}_{\mathbf{s}}(A:B:C) - \mathcal{A}(A:BC) -  \mathcal{A}(B:AC)- \mathcal{A}(C:AB)
\end{equation}
where $\mathcal{A}_{\mathbf{s}}$ is the standard minimal triway cut with equal tensions $\mathbf{s}=(1,1,1)$. 
\end{theorem}
\begin{proof}
    This follows from $\partial_n S_R^{(n)}\leq 0$ and Theorem~\ref{thm:1} applied at $n=2$.
\end{proof}

\begin{figure}[t]
\centering
\includegraphics[width=0.4 \textwidth]{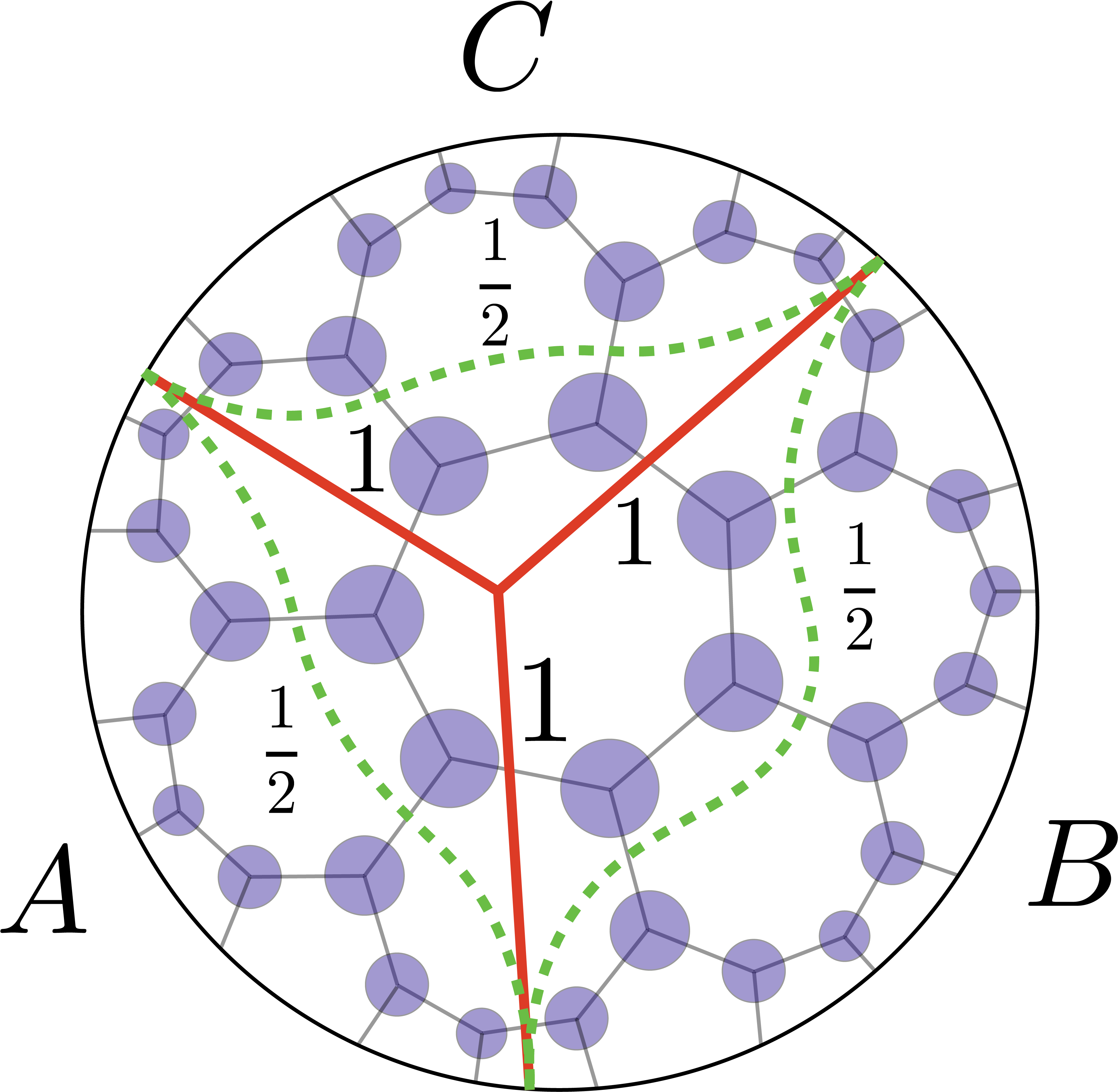}
\caption{The optimal configuration for the integer program computing the triway cut at $n=2$ involves domain walls (solid red) that indicate the location of $\rho(e)=1$. The optimal configuration for the same program after removing the integer constraint has domain walls (dashed green) with $\rho(e)=\frac{1}{2}$, that relax to the minimal surfaces in contrast with the triway cut.}
\label{fig:relax}
\end{figure}

We will show that the right-hand side of \Eqref{eq:MG-bound} is determined by the integrality gap of the integer program \cite{doi.org/10.1002/net.3230210106}:
\begin{align}
\begin{split}
\label{program-intro}
&\min_{\rho} \sum_{e} \rho(e) w(e) \\
&\forall L \in \mathcal{P}_{A,B} \cup \mathcal{P}_{A,C} \cup \mathcal{P}_{B,C} : \,\,  \sum_{e \in L} \rho(e) \geq 1
\end{split}
\end{align}
where $\rho(e) \in \mathbb{Z}_{\geq 0}$ and $\mathcal{P}_{x,y}$ refers to all paths through the edges of the network connecting vertices $x$ and $y$. The integrality gap $IG$ is the ratio between the two optimal values of the program before and after relaxing the integer constraint on $\rho$, a standard concept in the theory of integer programming \cite{Schrijver:LP}. The original program computes $\mathcal{A}(A:B:C)$ with equal tensions, while the relaxed program allows the domain walls to split into pairs and form three minimal cuts with tensions $1/2$ (see \figref{fig:relax}). The relaxed program is dual to the multicommodity flow problem \cite{Cui:2018dyq} on three parties for which efficient algorithms exist. In this sense, we can interpret the integrality gap as an obstruction to obtaining a bit threads picture. Explicitly our bound relates the two ``gaps'':
\begin{equation}
\label{gaps}
MG \geq (IG-1) \times ( \mathcal{A}(A:BC) +  \mathcal{A}(B:AC) + \mathcal{A}(C:AB))
\end{equation}
Roughly speaking $IG-1 \geq 0$ measures how computationally hard the integer program is, and so this gives an intriguing link between the Markov gap and complexity.\footnote{Note that this is logically different from the computational complexity of the state often discussed in AdS/CFT.}

\begin{figure}[t]
\centering
\includegraphics[width=0.3 \textwidth]{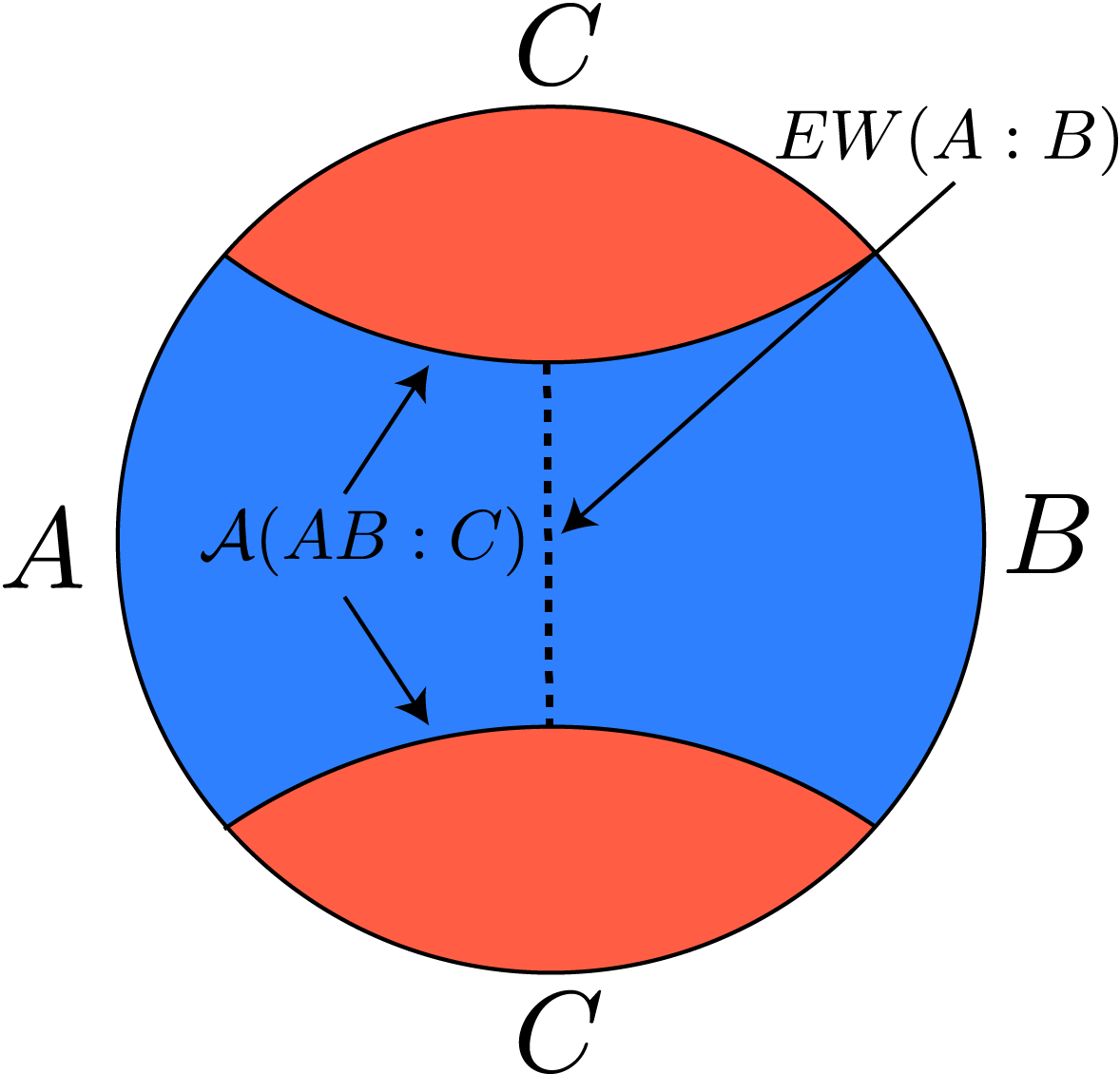}
\caption{The entanglement wedge of $AB$ (blue) is defined as the region between the boundary subregion $AB$ and the minimal cut $\mathcal{A}(AB:C)$. The entanglement wedge cross section $EW(A:B)$ is denoted with a dashed line.}
\label{fig:EW}
\end{figure}

While an obvious and natural continuation in $n$ away from the integers exists for the triway cut problem, we have not yet rigorously established that this still computes the $n$-R\'enyi reflected entropy. We have no reason to believe otherwise\footnote{Any lingering doubt might be used to challenge the conclusion \cite{Akers:2019gcv} that the bipartite dominance conjecture \cite{Cui:2018dyq} is false.  However, independent of the $EW$ duality, using \Eqref{gaps} we have now rigorously proven the bipartite dominance conjecture false for any random tensor network with $IG\neq 1$ for the program \Eqref{program-intro}.}, and so it is  worth noting that the limit $n \rightarrow 1$ reproduces the entanglement wedge cross-section:
\begin{equation}
\lim_{n \rightarrow 1 }  \left( \frac{1}{n-1}\mathcal{A}(A:B:C) - \frac{n}{n-1} \mathcal{A}(AB:C) \right) = 2 EW(A:B)
\end{equation}
where $EW(A:B)$ is simply the minimal cut dividing $A:B$ on the sub-graph defined by the ``entanglement wedge'' of the cut $AB:C$ (see \figref{fig:EW}). 

The calculation of reflected entropy in holography involves analytically continuing a two-parameter (usually called $m,n$) replica trick computation. It was shown in Ref.~\cite{Kusuki:2019evw} that the continuation of the naive saddles proposed in Ref.~\cite{Dutta:2019gen} suffers from an order of limits issue, an analog of which exists for RTNs as well \cite{Akers:2021pvd}. By incorporating new saddles, rigorously performing an analytic continuation in $m$ and proposing an analytic continuation in $n$ via the triway cut problem, we have resolved this issue in general RTNs. This motivates a similar prescription with the inclusion of new saddles even in AdS/CFT.

A summary of this paper is as follows. In \secref{sec:sum}, we state our main theorems
pertaining to the computation of
R\'enyi reflected entropies for the state $\ket{\rho_{AB}^{m/2}}$ with $m \geq 2$ an even integer. 
In \secref{sec:prelim} we give some required background, including results on standard network flows, minimal cuts, and the mathematics of permutations and set partitions that make appearances in various statistical mechanics models that we consider. \secref{sec:main} proves that the optimal solution to the reflected entropy statistical mechanics model can be found in a series of coarser models, finally ending in the multiway cut problem. In \secref{sec:cont} we continue the $m$ parameter away from even integers to $m \rightarrow 1$, where we make contact with the reflected entropy. In this step, we use the method of moments in conjunction with a weak form of measure concentration for random tensor network states. In \secref{sec:disc}, we discuss various aspects of our work such as bit threads, the relation to entanglement of purification as well as generalization to hypergraphs. Several lengthy proofs are relegated to Appendices. 

We end this introduction with a list of common notations used in this article: \\
\begin{center}
    \begin{longtable}{l|l|l}
    \caption{Glossary of symbols and notations.}
    \label{tab:glossary} \\
    Term & Description & First defined in \\
    \hline
    $G=\{E,V\}$ & graph & Def.~\ref{def:rtn}\\
    $w(e)$ & weight of an edge $e\in E$ & after Def.~\ref{def:rtn}\\ 
    $E_G[V']$ or $E[V']$ & set of edges that have some vertex in $V'$ & before Def.~\ref{def:R} \\
    $V_G[E']$ or $V[E']$ & set of vertices that lie in $E'$ & before Def.~\ref{def:R} \\
    $\mathcal{P}_{A,B}$ & set of paths from $A$ to $B$ & before Def.~\ref{def:R} \\
    $\widehat{\mathcal{P}}_{A,B}$ & set of edge-disjoint paths from $A$ to $B$ & before Def.~\ref{def:R} \\
    $r_{A}$ or $r_{A:B}$ & cut region containing $A$ or dividing $A:B$ & before \Eqref{eq:mu_def}\\
    $\mu(r)\subset E$ & cut surface of a region $r \subset V$ & \Eqref{eq:mu_def} \\
    $\mathcal{A}(A:B)$ & minimal cut for $A$ and $B$ & Def.~\ref{def:cut} \\
    $\mathcal{A}(A:B:C)$ & minimal triway cut among $A$, $B$ and $C$ & Def.~\ref{def:tri-cut} \\
    $\mathbf{1}_s(x)$ & indicator function of the set $s$ &  after \Eqref{cpfeas} \\
    $S_N$ & symmetric group of order $N$ & before \Eqref{eq:Cayley}\\
    $P_N$ & partitions of a set of order $N$ & before \Eqref{eq:q(g)}\\
    $B_N$ & set of string of $N$ Boolean algebras & before \Eqref{eq:bk(q)}\\
    $P(g)$ & coarse graining $P: S_{mn}\to P_{mn}$ & before \Eqref{eq:q(g)} \\
    $q_{g_0}(p)$ & coarse graining $q_{g_0}: P_{mn} \to P_{\#(g_0)}$ & \Eqref{eq:q_g0} \\
    $q_{X}(p) $ & coarse graining $q_X: P_{mn} \to P_{2n}$ & \Eqref{eq:q(g)} and \Eqref{eq:q_g0}  \\
    $s(q)$ & coarse graining $s: P_{2n} \to s(P_{2n}) \simeq B_{2n}$ & before \Eqref{eq:S_metric}\\
    $b^k(q)$ & coarse graining  $b^k: P_{2n} \to (B_{2n})^k \simeq \mathbb{Z}_2$ & before \Eqref{d1s} \\
    $u_k$ & largest partition with a singlet at location $k$ & after \Eqref{eq:bk(q)} \\
    $\#(x)$ & cycle ($x\in S_N$) or block ($x\in P_N$) counting function & \Eqref{eq:Cayley} and \Eqref{eq:P_metric}  \\
    $\#_1(q)$ & number of singlets in $q\in P_N$ & before \Eqref{eq:S_metric}\\
    $a \wedge b$ & meet of $a$ and $b$ & beginning of \secref{sec:ppb} \\
    $a \vee b$ & join of $a$ and $b$ & beginning of \secref{sec:ppb} \\
    $d(g_1,g_2)$ & ($g_1,g_2 \in S_N$) Cayley distance in $S_N$ & \Eqref{eq:Cayley} \\
    $d(q_1,q_2)$ & ($q_1,q_2 \in P_N$) distance on semimodular lattice $P_N$ & \Eqref{eq:P_metric} and \Eqref{eq:P_metric_2}\\
    $d_1(s_1,s_2)$ & ($s_i=s(p_i):p_i\in P_N$) singlet distance on $P_N$ & \Eqref{eq:S_metric} \\
    $d(b_1,b_2)$ & ($b_1,b_2 \in B_N$) Hamming distance on $B_N$ & \Eqref{eq:B_metric} and \Eqref{d1s}\\
    $d_\rho(x,y)$ & ($x,y\in E$) graph distance induced by $\rho:E\to \mathbb{R}$ & \Eqref{eq:d_graph}
    \end{longtable}
\end{center}
Throughout this article, we will define various distance functions on different sets. Most of these distances will be termed universally by $d(\cdot,\cdot)$, with an understanding that we use different definitions based on the set in context, as shown in Table~\ref{tab:glossary}.
For a function $\rho(e)$ on the set of edges $\rho:E\to\mathbb{R}$ we will sometimes write $\rho(E') \equiv \sum_{e\in E'}\rho(e)$ for a subset $E'\subset E$.

Note: An alternate proposal for the entanglement wedge cross section ($EW$) was made in Refs.~\cite{Takayanagi:2017knl,Nguyen:2017yqw} relating it to the entanglement of purification ($E_P$). In Ref.~\cite{Akers:2023obn}, we used the results obtained in this paper to prove the $E_P=EW$ conjecture for specific RTNs.
Also note that triway cuts in holography have also been discussed in Refs.~\cite{Gadde:2022cqi,Penington:2022dhr,Gadde:2023zzj,Gadde:2023zni} as a candidate holographic dual to a different quantity, the multi-entropy. We will comment more on the relation to this work later in the paper.

\section{Summary of main theorem}
\label{sec:sum}

We summarize our main mathematical results, setting up some of our notation at the same time. Our main results utilize the replica trick used to compute reflected entropy in random tensor networks \cite{Akers:2021pvd,Akers:2022zxr}, a computation we formulate using graph theory.
See also \cite{Cao:2024tog} for a related replica computation.

\begin{definition}[Random tensor network states \cite{Hayden:2016cfa}]
\label{def:rtn}
Consider an undirected graph $G = (V,E)$ where edges correspond to unordered pairs of vertices $e = \{ v_1,v_2\}$ for $e \in E$ and $v_{1,2} \in V$. We mark special vertices $\partial \subset V$ as boundary vertices
and the graph describes a pure state in the Hilbert space: 
\begin{equation}
\left|\psi\right> \in \bigotimes_{v \in \partial} \mathcal{H}^v , \qquad \mathcal{H}^v = \bigotimes_{e \in E(v)} \mathcal{H}_{\chi(e)}^{(v)},
\end{equation}
where $E(v):=\{\{x,y\}\in E: x=v \text{ or } y=v\}$ are the subset of edges that contain $v$ and $\chi(e)$ is the bond dimension of the edge.
For a given vertex $v$ not in the boundary, i.e. $v \in V \backslash \partial$, we pick random tensors $T(v) \in \mathcal{H}^v$ 
according to the Haar measure. Then the states in question are defined via:
\begin{equation}
\left| \psi \right> \propto \left(\mathop{\otimes}_{v \in V \backslash \partial} \left< T(v) \right| \right)  \left(\mathop{\otimes}_{e \in E} \left| \Psi_e \right> \right)
\end{equation}
where $\ket{\Psi_e} \in \mathcal{H}_{\chi(e)}^{(v_1)} \otimes \mathcal{H}_{\chi(e)}^{(v_2)}$
is a maximally entangled state between the vertex Hilbert spaces of $\{v_1,v_2 \} = e$.
\end{definition}

For convenience, on the graph we will use the rescaled weighting
\begin{equation}
w(e) = \frac{\ln \chi(e)}{\ln\chi},
\end{equation}
which we hold fixed as we send $\chi \rightarrow \infty$. 

We will care about graphs with boundary $\partial = A \sqcup B \sqcup C$, and we are interested in using the replica trick to compute the averaged $(m,n)$-R\'enyi reflected entropy
\begin{equation}
\label{BBn}
    S_R^{(m,n)}(A:B) = - \frac{1}{n-1}\ln \overline{{\rm Tr} (\rho^{(m)}_{AA^\star})^n}~,~~~~ \rho_{AA^*}^{(m)}={\rm Tr}_{BB^\star} \left| \rho_{AB}^{m/2} \right> \left< \rho_{AB}^{m/2} \right|,
\end{equation}
where $\rho_{AB} = {\rm Tr}_C \left| \psi \right> \left< \psi \right|$ and where $m \in 2 \mathbb{Z}_{\geq 1}$.
The $n$-th moment ${\rm Tr}(\rho^{(m)}_{AA^\star})^n$ contains $nm$ copies of $\ket{\psi}\bra{\psi}$, and the overline denotes that we compute the quantity in \Eqref{BBn} averaged over the choice of $T(v)$ from the Haar ensemble.
This calculation reduces to a statistical mechanics model with vertex-valued group elements $g(v) \in S_{mn}$ \cite{Hayden:2016cfa, Akers:2021pvd}:
\begin{align}
\label{statmech}
    \overline{{\rm Tr} (\rho^{(m)}_{AA^\star})^n} = \frac{Z_{m,n}}{(Z_{m,1})^n}, \quad
    Z_{m,n} =\sum_{\{g(v)\}} \exp\left(- \sum_{e = \{x,y\} \in E} d(g(x),g(y)) \ln \chi(e) \right)~,
\end{align}
where $\sum_{\{g(v)\}}$ denotes a sum over all configurations of $g(v)$.
We have denoted by $d(g_1,g_2)$ the Cayley distance in $S_{mn}$:
\begin{equation}
\label{eq:Cayley}
d(g_1,g_2) = mn - \#(g_1 g_2^{-1}),
\end{equation}
where $\#(\cdot)$ is the cycle counting function.
Note that we do not sum over permutations at boundary vertices; instead their permutations are fixed by the patterns of contractions that compute the moments in \Eqref{BBn}.
We fix $g(v) = g_A$ for $v \in A$, $g(v) = g_B$ for $v \in B$, and $g(v) = \id$ for $v \in C$ where $\id$ is the identity group element.\footnote{The identity element in a group $G$ is colloquially termed $e\in G$. This unfortunately clashes with our notation for edges. Hence in this article we will use $id$ to denote the identity element in $S_N$ as well as the finest element of set partitions $P_N$. This will hopefully not cause confusion to our readers.}
The group elements $g_A$ and $g_B$ can be read off from \Eqref{BBn}. They are related by a conjugation and contain $n$ cycles of length $m$. The Cayley distances between these elements are $d(g_A,g_B) = 2(n-1)$, $d(\id,g_A)=d(\id,g_B)=n(m-1)$.

Exactly evaluating \Eqref{statmech} for general graphs is difficult, but in the large $\chi$ limit we can often obtain a very good approximation by evaluating \Eqref{statmech} at a saddle point, i.e., finding the configuration $\tilde{g}(v)$ that maximizes the exponential and dropping all other terms. Namely, we have
\begin{equation}
    Z_{m,n} \approx \exp\left(- \sum_{e = \{x,y\} \in E} d(\tilde{g}(x),\tilde{g}(y)) \ln \chi(e) \right)~,
\end{equation}
for $\tilde{g}(v)$ now the configuration minimizing the ``free energy''
\begin{equation}
\label{eq:free_energy}
    \sum_{e = \{x,y\} \in E} d(\tilde{g}(x),\tilde{g}(y)) \ln \chi(e)~.
\end{equation}
The saddle point approximation is valid away from phase transitions and thus, in such generic situations, the problem of computing $S_R^{(m,n)}(A:B)$ reduces to the problem of finding the optimal configuration $\tilde{g}(v)$.
It is this problem that we focus on in this paper.

A particularly important group element that will arise in the optimal configuration $g(v)$ for \Eqref{statmech} is uniquely defined as the element $X \in S_{mn}$ with the largest Cayley distance to the identity, $d(X,\id)$, that lies on the joint Cayley geodesics defined by: $d(g_A,X) + d(X,\id) = d( g_A,\id)$ and $d(g_B,X) + d(X,\id) = d(g_B,\id)$. $X$ contains $2n$ cycles of length $m/2$, roughly speaking, corresponding to the intersection of the cycles in $g_A$ and $g_B$.  
The specific form of $g_A,g_B$ and $X$ can be found in  Appendix~\ref{app:specific}.

\begin{figure}[t]
    \centering
    \includegraphics[scale=.35]{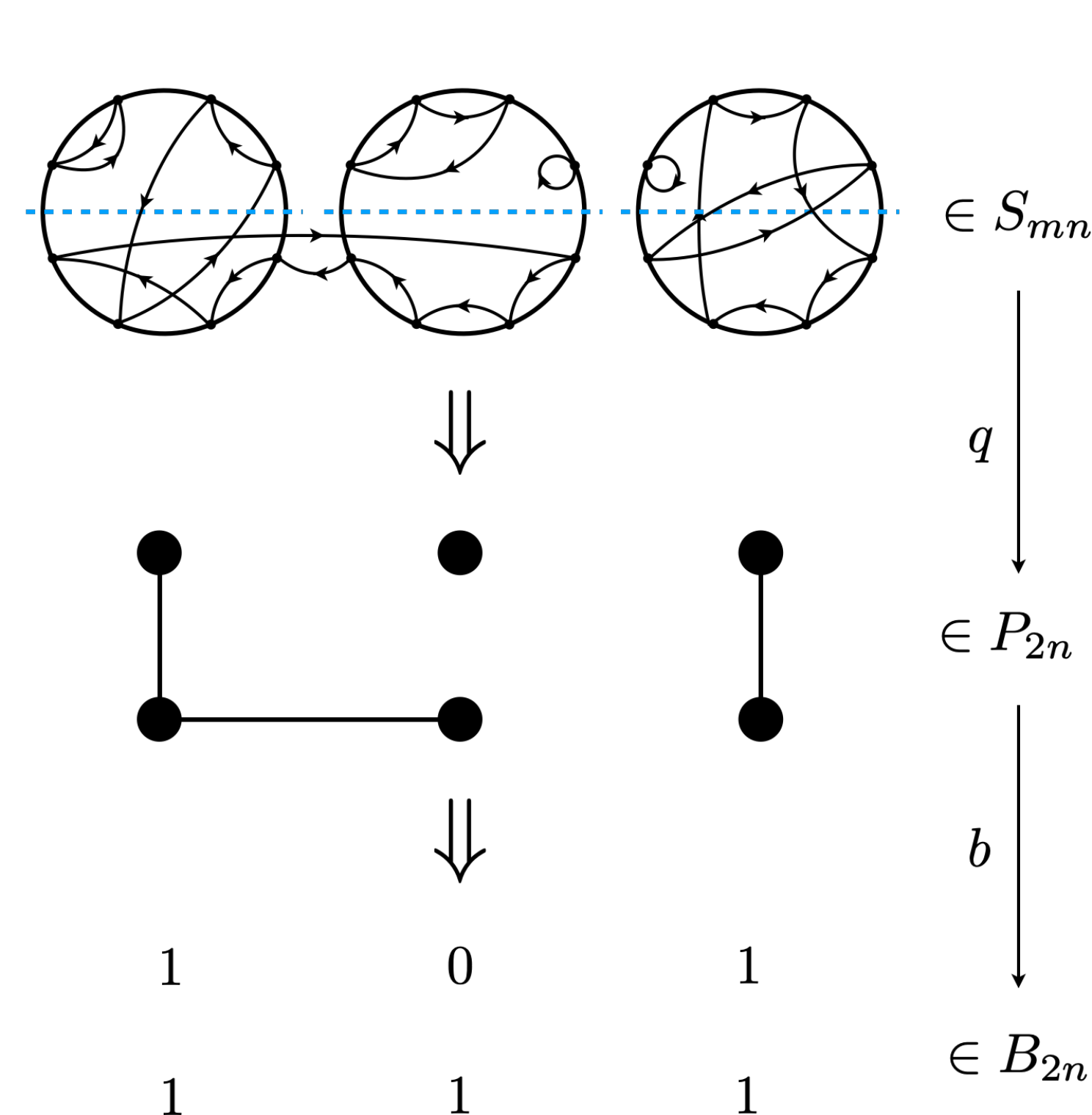}
    \caption{An illustration of the coarse-graining process used to establish our main result. We start from $g\in S_{mn}$ as shown at the top of the figure. For the convention used here to represent $g$: Each circle represents the $m$ replica, and each circle is further replicated $n$ times (here $m=8$ and $n=3$). The actions of $g$ are represented as directed arrow connecting different replicas. See also Appendix~\ref{app:specific}.
    To pass from $S_{mn}$ to $P_{2n}$, we divide the $mn$ elements into $2n$ blocks, each containing the $m/2$ elements in the corresponding half circle. We join two blocks if and only if any element in them can be reached by repeated actions of $g$. To pass from $P_{2n}$ to $B_{2n}$, we simply label each block by a number in $\{0,1\}$. That is, we label a block by $0$ if it is not connected to any other block (such a block is called a \emph{singlet}), and by $1$ if otherwise. The result is a bit string of length $2n$.}
    \label{fig:coarse-grain}
\end{figure}
In order to find the optimal configuration $g(v)$, with minimal free energy, we will relate this statistical mechanics model to a sequence of increasingly coarse grained  models where we throw out some information contained in $g(v)$.
See \figref{fig:coarse-grain}.
As a result the free energy can only decrease (as we prove rigorously), but then -- crucially -- we will prove that the minimal free energy in the most coarse-grained model \emph{upper bounds} the free energy of the original model! 
This loop of inequalities allows us to use the tractable, most coarse-grained model to find the minimal free energy.
This approach was inspired by the solution to the single tensor case studied in Appendix~B of Ref.~\cite{Akers:2021pvd} where a similar sequence of steps was followed. Below we first describe the various coarse-grained variables, then we define the various stat mech/optimization problems we are interested in. An important ingredient is the theory of set partitions which form a lattice, in the sense of a partially ordered set, with binary operations $\vee$ and $\wedge$. Boolean algebras will also make an appearance.  We review the necessary background material in \secref{sec:ppb}. 

Given a permutation $g \in S_N$ we can associate a set partition $P(g) \in P_N$ by mapping the cycles to subsets of $\mathbb{Z}_N$, simply forgetting how elements within each cycle are permuted.
In our case $N= nm$. We then further coarse-grain this set partition further by \emph{blocking} $P(g)$ into partitions of the $2n$ blocks in $P(X)$. This reduction effectively removes dependence on $m$. More specifically, for each $g \in S_{mn}$ we associate an element in $P_{2n}$ via $q_X(g) : S_{mn} \rightarrow P_{2n}$ defined as:
\begin{equation}
\label{eq:q(g)}
q_X(g) \equiv (P(g) \vee P(X)) / \sim \qquad x \sim y \,\,\,\, {\rm iff} \,\,\,\,  \{x, y \} \vee P(X) = P(X)
\end{equation}
where $\vee$ is the least upper bound operation on the lattice (defined in \secref{sec:ppb}) and the quotient is applied element wise to each element within the partition. We introduce a distance on set partitions\footnote{See \secref{sec:ppb} for some properties of this metric. See Ref.~\cite{MONJARDET1981173} for some discussions of this metric. }:
\begin{equation}
\label{eq:P_metric}
d(q_1,q_2) = \#(q_1) + \#(q_2) - 2 \#(q_1 \vee q_2)
\end{equation}
where $\#(\cdot)$ now counts the number of subsets in the partition. Note that $\#(g)=\#(P(g))$. This distance will replace the Cayley distance function in \Eqref{statmech}. For the problem in $P_{2n}$, the boundary elements will be fixed to partitions: $q_A = q_X(g_A)$ and $q_B = q_X(g_B)$ and $q_C = \id \equiv q_X(\id)$. Again we list these specific partitions in Appendix~\ref{app:specific}. The partition distances are $d(q_A,q_B) = 2(n-1)$ and $d(q_{A,B},\id) = n$.

At the next level, we introduce Boolean variables $b = \{b^k\}_{k=1}^{2n} \in (\mathbb{Z}_2)^{\otimes 2n}\equiv B_{2n}$ that detect the presence of a singlet (subsets with size $1$) in $q \in P_{2n}$.
Define:
\begin{equation}
\label{eq:bk(q)}
b^k(q) = 2-\#(u_k \vee q),
\end{equation}
where $u_k$ for $k \in \{ 1,2, \ldots 2n \}$, is the partition with a singlet at $k$ plus a block of size $2n-1$ containing with the rest of the elements.
Thus, $b^k = 0$ implies there is a singlet at location $k$; otherwise $b^k =1$. 
The distance between two Boolean strings $b_{1,2}$ is:
\begin{equation}
\label{eq:B_metric}
d(b_1,b_2) = \sum_{k =1}^{2n} | b_1^k - b_2^k |.
\end{equation}
In this case, the boundary conditions are $b_{AB} =11 \ldots 1$ and $b_C = 00 \ldots 0$. 

The set of walks/paths in a graph between $v_1$ and $v_2$ will be denoted $\mathcal{P}_{v_1,v_2}$, and consist of a sequence of edges $L \subset E$ that join at common vertices
and start (end) at $v_1$ ($v_2$). 
The set of edge disjoint paths will be denoted $\widehat{\mathcal{P}}_{v_1,v_2}$ -- these are paths with no repeated edges (vertices may be repeated).
The vertices $v_1$ and $v_2$ may be replaced by sets of vertices with paths starting/ending on any vertex in the respective set. 
For a graph $G=\{V,E\}$ and a subset of edges $E'\subset E$, denote $V_G[E']=\{v\in V: \{x,v\}\in E' \text{ or } \{v,y\}\in E'\}$ as the set of vertices that lie in $E'$. 
Similarly, $E_G[V']=\{ \{x,y\}\in E: x\in V' \text{ or } y\in V'\}$ are the set of edges that have \emph{some}
vertex in $V' \subset V$.
The subscript $G$ in $E_G[\cdot]$ and $V_G[\cdot]$ shall often be omitted when the referencing graph is unambiguous.
Lastly, for a function on edges $\rho : E \rightarrow \mathbb{R}$ we define $\rho(E') = \sum_{e \in E'} \rho(e)$ for a subset $E' \subset E$. 

Define the various optimization problems based on the refinements $g \rightarrow q \rightarrow b$:
\begin{definition}[\emph{Reflected entropy optimization}]
\label{def:R}
\begin{equation}
R = \min_{ g}  R(g) \,\,, \qquad R(g) \equiv \sum_{e = \{x,y\} \in E } w(e) d(g(x),g(y))
\end{equation}
where $g : V \rightarrow S_{mn}$ and such that $g(A) = g_A$ , $g(B) = g_B$ and $g(C) = \id$. 
\end{definition}

\begin{definition}[\emph{Set partition optimization}]
\label{def:Q}
\begin{equation}
Q = \min_{q} Q(q) \,\, , \qquad Q(q) = \sum_{e =  \{x,y\} \in E} w(e) d(q(x),q(y))
\end{equation}
where $q : V \rightarrow P_{2n}$ and such that $q(A) = q_A$ and $q(B) = q_B$ and $q(C) = \id$.
\end{definition}

\begin{definition}[\emph{Boolean optimization}]
\label{def:B}
\begin{align}
B &= \min_{b} B(b) \\  B(b) &= \min_{r}   \sum_{e \in E}  w(e) r(e)  \\
&{\rm subject \,\, to}\, \quad \forall L \in \mathcal{P}_{A,B}\,\, : \sum_{e \in L} r(e) \in 2(n-\delta_{b_L,b_{AB}})  + 2 \mathbb{Z}_{\geq 0}\\
&{\rm and \,\, }\, \qquad \forall e \in E \,\, : r(e) \geq \left\lceil \frac{1}{2} d(b(x),b(y)) \right\rceil 
\end{align}
where
$r : E  \rightarrow \mathbb{Z}_{\geq 0}$ and $b : V \rightarrow B_{2n}$ such that
$b(A) = b(B) = b_{AB} \equiv 11 \ldots 1$ and $b(C) = 00 \ldots 0$, where $b_L\equiv \wedge_{v\in V[L]}b(v)$ denotes the piecewise \emph{and} operation on Boolean algebra along the path $L$
\footnote{That is for $L = \{\{v_0,v_1\},\{v_1,v_2\},\cdots,\{v_{i-1},v_i\}\}$: $b_L = b(v_1)\wedge b(v_2)\wedge\cdots\wedge b(v_i)$ and $(b_1\wedge b_2)^k \equiv b_1^k \wedge b_2^k$. Recall that $b^k$ denotes the $k$-th element of the bit string $b$.}.
\end{definition}

\begin{definition}[\emph{Integer program}]
\label{def:I}
\begin{align}
I &=  \min_{\rho, \sigma}   \sum_{e \in E}  w(e) (\sigma(e) + \rho(e)) \\ 
\label{eveness}
&{\rm subject \,\, to}\, \quad \forall L \in \mathcal{P}_{A,B}\,\, : \sum_{e \in L} (\sigma(e) + \rho(e)) \in   2(n - \delta_{\rho(L), 0})  + 2 \mathbb{Z}_{\geq 0} 
\\ 
&{\rm and \,\, }\, \qquad \forall L \in \mathcal{P}_{AB,C} \, \, : \sum_{e \in L} \rho(e) \geq n  
\end{align}
where $\rho(L)\equiv \sum_{e\in L}\rho(e)$ and $\rho,\sigma: E \rightarrow  \mathbb{Z}_{\geq 0}$. 
\end{definition}

Our main result is to prove that the optimal solution to these programs is given by a particular minimal and multiway cut program that we now define.
A \emph{cut} (or \emph{cut set}) $r$ for two vertices (or sets of vertices), $A:B$ is defined as a subset $r\subset V$ such that $A\subset r$ and $B\subset r^c$.\footnote{We will sometimes denote a cut that divide boundary vertices $A:B$ as $r_{A:B}$, or simply as $r_A$ when $B=A^c$.}
A \emph{cut surface} $\mu(r)$ for a given subset $r \subset V$ comprises the subset of edges that lie on the ``boundary'' of $r\leftrightarrow r^c$, in the sense that for all $e \in \mu(r)$ the pair $\{ x, y \} = e$ satisfies $x \in r$ and $y \in r^c$ or vice versa. In other words:
\begin{equation}
\label{eq:mu_def}
    \mu(r) = E[r]\cap E[r^c].
\end{equation} 

\begin{definition}[\emph{Minimal cut}]
\label{def:cut}
The minimal cut for $A:B$ partitions two sets of boundary vertices $\partial = A \sqcup B$ and minimizes:
\begin{align}
\mathcal{A}(A:B) & = \min_{r} \mathcal{A}(r) \\
\nonumber
\mathcal{A}(r) &= \sum_{e\in \mu(r)}w(e) \equiv w(\mu(r))
\end{align}
over all subsets $r \subset V$, such that the boundary vertices $A \subset r$ and $B \subset r^c$. 
\end{definition}

\begin{definition}[\emph{Triway cut}]
\label{def:tri-cut}
Fix some tensions $\mathbf{t}=\{t_{A:B},t_{B:C},t_{C:A}\}$ with $t_{A:B}, t_{B:C},t_{C:A} > 0$ and boundary vertices $\partial = A \sqcup B \sqcup C$. The triway cut problem minimizes:
\begin{align}
\mathcal{A}_{\mathbf{t}}(A:B:C) &= \min_{(\alpha,\beta,\gamma)} \mathcal{A}_{\mathbf{t}}(\alpha:\beta:\gamma) \\
 \nonumber
\mathcal{A}_{\mathbf{t}}(\alpha:\beta:\gamma) &= t_{A:B}\sum_{e\in\mu(\alpha:\beta)}w(e)+t_{B:C}\sum_{e\in\mu(\beta:\gamma)}w(e)+t_{C:A}\sum_{e\in\mu(\gamma:\alpha)}w(e) \\
\nonumber &= t_{A:B}\,w(\mu(\alpha:\beta)) + t_{B:C}\,w(\mu(\beta:\gamma)) + t_{A:B}\,w(\mu(\gamma:\alpha))
\end{align}
over all subsets $(\alpha,\beta, \gamma)$ such that
 $\alpha\sqcup\beta\sqcup\gamma=V$ with $A\subset\alpha, B\subset\beta, C\subset\gamma$
and where $\mu(r_1:r_2) = E[r_1]\cap E[r_2]$.
\end{definition}
\begin{remark}
The triway cut is a special case of a \emph{multiway cut} and we will often refer to it as such. The equal tension case is the standard multiway cut problem, which arises here when $n=2$ (see e.g., Ref.~\cite{doi:10.1137/S0097539792225297,doi.org/10.1002/net.3230210106}). In the context of AdS/CFT, multiway cuts have been studied previously in \cite{Gadde:2022cqi,Penington:2022dhr,Gadde:2023zzj,Gadde:2023zni}. 
\end{remark}

\begin{figure}
\centering
\includegraphics[scale=.35]{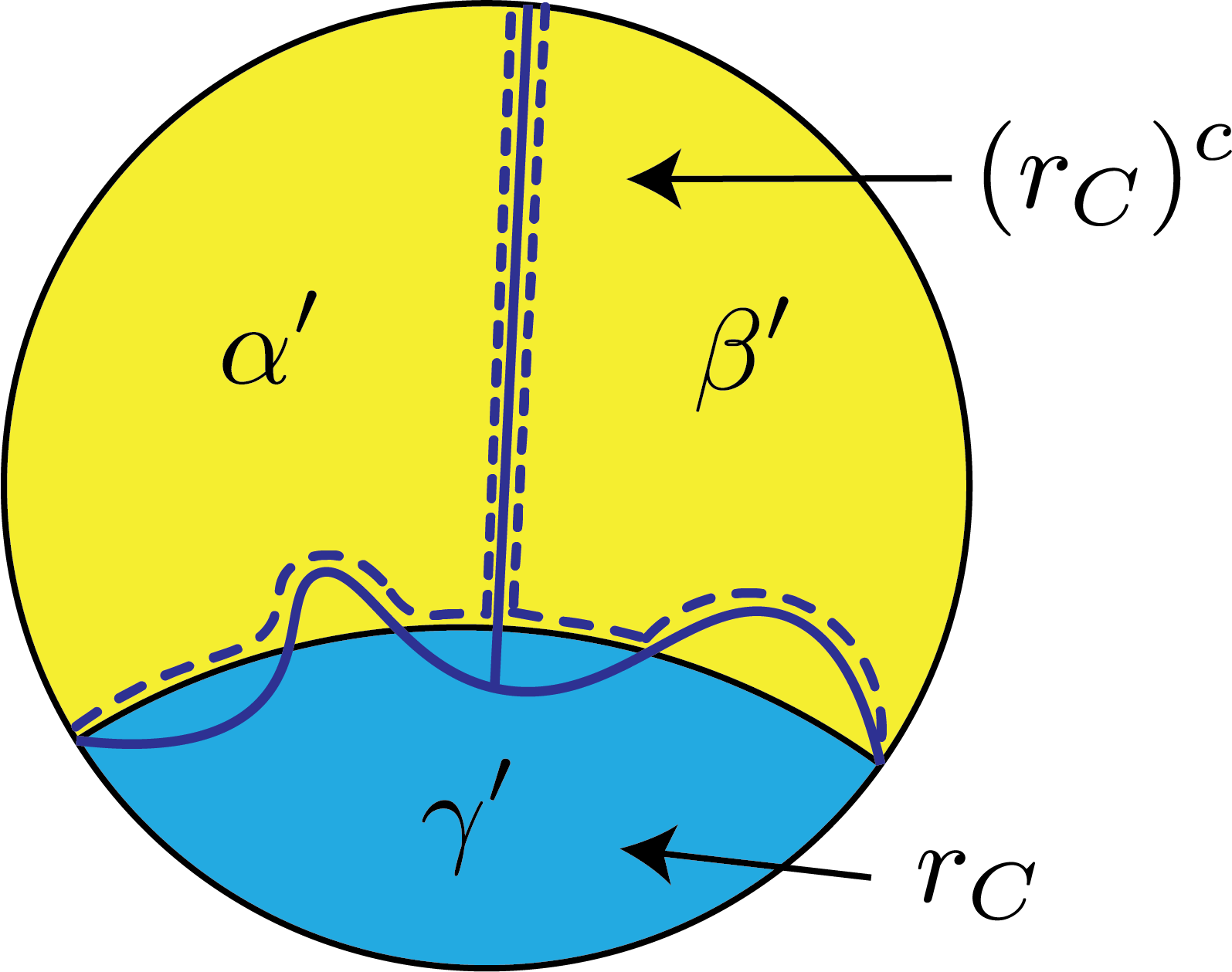}
\caption{The entanglement wedge of $C$ is denoted $r_C$ (blue) and the complement is $(r_C)^c$ (yellow). For any triway cut configuration $(\alpha',\beta',\gamma')$ (demarcated by solid lines), we can construct a better triway cut $(\alpha,\beta,\gamma)$ (separated by dashed lines) such that $r_C \subset \gamma$. }
\label{fig:lemma3}
\end{figure}

We note that minimal cuts must always lie inside some optimal multiway cut region: 
\begin{lemma}
\label{cutinmcut}
Given an optimal solution to a minimal cut problem for $C: AB$, $r_C \subset V$ then there exists
an optimal solution $(\alpha,\beta,\gamma)$  to the triway cut problem for $(A,B,C)$ with $t_{A:C} = t_{B:C}$  and such that $r_C \subset \gamma$.
\end{lemma}
\begin{proof}
 Consider any optimal triway cut $(\alpha',\beta',\gamma')$. Construct $\alpha = \alpha' \cap (r_C)^c$,  $\beta = \beta' \cap (r_C)^c$  and $\gamma =\gamma' \cup r_C$ (see \figref{fig:lemma3}). We note:
\begin{align}
 \sum_{e \in \mu(\alpha:\beta)} w(e) \le \sum_{e \in \mu(\alpha':\beta')} w(e)
\end{align}
since $\mu(\alpha:\beta)\supseteq\mu(\alpha':\beta')$ by construction and $w(e)\ge0$. Also,
\begin{align}
\begin{split}    
    &\sum_{e\in \mu(\alpha:\gamma)}w(e) + \sum_{e\in \mu(\beta:\gamma)}w(e) = w(\mu(\gamma)) \\
   = \, &w(\mu( \gamma' )) + w(\mu(r_C)) - w( \mu(\gamma \backslash r_C^c)) \leq  w(\mu( \gamma' )) 
  \end{split}
\end{align}
where in the last inequality we used the fact that $\gamma \backslash r_C^c$ is a cut for $C:AB$ and that $r_C$ is a minimal such cut.
Putting these inequalities together we have:
\begin{equation}
\mathcal{A}(\alpha,\beta,\gamma) \leq \mathcal{A}(\alpha',\beta',\gamma')
\end{equation}
implying equality and that $(\alpha,\beta,\gamma)$ is a minimal triway cut satisfying the properties stated in the Lemma. 
\end{proof}

The main theorem of this paper establishes:
\begin{theorem}
\label{thm:RPBI}
The minimum of each of the programs defined above are determined by an optimal solution to the multiway cut problem
with $\mathbf{t}=(t_{A:B},t_{B:C},t_{A:C}) = (2(n-1),n,n)$:
\begin{equation}
R -  \mathcal{A}(AB:C) n(m-2) = Q = B = I = \mathcal{A}_{\mathbf{t}}(A:B:C)
\end{equation}
\end{theorem}
\begin{proof}
Consider some minimal cut $r_C$ for $C:AB$. 
There is an optimal solution to the triway cut problem with $r_C \subset \gamma$ by Lemma~\ref{cutinmcut}.
Set 
\begin{equation}
\label{eq:unique_g}
g(x) = \begin{cases}
    g_A, \quad &x\in \alpha \\
    g_B, \quad &x\in \beta \\
    \id, \quad &x\in r_C \\
    X, \quad &x\in \gamma\setminus r_C
\end{cases}
\end{equation}
We then estimate:
\begin{equation}
\label{rless}
R \leq \mathcal{A}(AB:C) d(\id, X) +  \mathcal{A}_{\mathbf{t}}(A:B:C)
\end{equation}
by direct computation and since $d(X,g_A) = d(X,g_B) = n$ and $d(g_A,g_B) = 2(n-1)$ give the correct tensions. 

After this we prove a chain of inequalities in the other way, $R \geq \mathcal{A}(AB:C) d(\id, X)  + Q$ (Lemma~\ref{lem:RtoP'} and Lemma~\ref{lem:RtoP}) , and $Q \geq B$ (Lemma~\ref{lem:PtoB2} and Corollary~\ref{lem:PtoB}), and $B \geq I$ (Lemma~\ref{lem:BtoI})
and finally $I \geq  \mathcal{A}_{\mathbf{t}}(A:B:C)$ (Theorem~\ref{thmttt}).
Together with \Eqref{rless} this implies equality through the chain. 
\end{proof}

Theorem~\ref{thm:RPBI} only asserts that the optimal value of the program $R$ are equivalent to that of the triway cut. It need not be the same as the one that was constructed above and there may be other degenerate solutions that achieve the minimum -- Indeed one generally expects a huge degeneracy when a phase transition happens. Such phase transitions in tensor networks are usually signaled by a degenerate minimal surface. The following theorem states that the optimal solution constructed above is unique when the system is far away from such transitions.

\begin{theorem}
\label{thm:uniqueness}
    Let $g:V\to S_{mn}$ be an optimal solution to the reflected entropy permutation group optimization problem $R$ on graph $G$ with a unique solution $(\alpha,\beta,\gamma)$ to the triway cut problem for $(A,B,C)$ and a unique solution $r_C$ to the minimal cut problem $C:AB$. Then $g$ is unique and of the form as given in \Eqref{eq:unique_g}.
\end{theorem}
We prove this theorem in \secref{sec:uniqueness} after we establish the necessary Lemmas that lead to Theorem~\ref{thm:RPBI}.

\section{Preliminaries}
\label{sec:prelim}

\subsection{Min cuts and Max flows}
\label{sub:minmax}

We discuss a version of max-flow min-cut that is convenient here. 
Consider the undirected max flow problem written as a linear program and in terms of edge disjoint paths:
\begin{align}
\label{program:cP}
 \max_{c} & \sum_{L \in \widehat{\mathcal{P}}_{A,B} } c(L)\, \\
{\rm subject \,\,\, to:} & \quad \forall e \in E \quad \sum_{P \in \widehat{\mathcal{P}}_{A,B}} c(L) \mathbf{1}_{L}(e) \leq w(e)
\label{cpfeas}
\end{align}
where $\mathbf{1}_{L}(e) = 1$ if $e \in L$ and $0$ otherwise;  and where $c : \widehat{\mathcal{P}}_{A,B} \rightarrow \mathbb{R}_{\geq 0}$ and we maximize over all such maps. Any $c(L)$ that satisfies the constraint above is called \emph{feasible} while any feasible $c$ that achieves the maximum is called \emph{optimal}. 

Let us derive the dual program. For each constraint we introduce
a function $\rho: E \to \mathbb{R}_{\ge0}$ and minimize over $\rho(e)$: 
\begin{equation}
= \max_{c} \min_{\rho} \left( \sum_{L \in \widehat{\mathcal{P}}_{A,B} } c(L) + \sum_{e \in E} \rho(e) \left( w(e) - \sum_{P \in \widehat{\mathcal{P}}_{A,B}} c(L) \mathbf{1}_{L}(e) \right)  \right)
\end{equation}
The minimum is $-\infty$ with $\rho \rightarrow \infty$ for some $e$ if $c$ is not feasible. If $c$ is feasible the minimum is $\rho(e) =0$. Thus, this gives back the original problem after maximizing over $c$. We can now exchange the maximum and the minimum since the function is linear in $c$ at fixed $\rho$ (and hence concave)
and linear in $\rho$ at fixed $c$ (and hence convex), so we can use the von Neumann's minimax theorem:
\begin{equation}
=  \min_{\rho}   \max_{c}\left(\sum_{e \in E} \rho(e)  w(e) +  \sum_{L \in \widehat{\mathcal{P}} }c(L) (1 -\sum_{e \in L} \rho(e) )  \right)
\end{equation}
Maximizing over $c$ gives $+ \infty$ unless $\sum_{e \in P} \rho(e)  \geq 1$ in which case we find $c(L) = 0$. Thus, we can restrict to the set of \emph{dual feasible} $\rho$
satisfying this later constraint. 

In summary, the dual program is defined in terms of a map $\rho : E  \rightarrow \mathbb{R}_{\geq 0}$:
\begin{align}
\label{rhomin}
 \min_{\rho} & \sum_{e \in E } \rho(e) w(e) \, \\
{\rm subject \,\,\, to:} & \quad \forall L \in \widehat{\mathcal{P}}_{A,B} \quad \sum_{e \in L} \rho(e) \geq 1
\label{tosat}
\end{align}
As we just derived, this problem satisfies strong duality.\footnote{Strong duality means that the optimal values of the original and dual program are equal.} 
We can also impose the further constraints that all $L \in \mathcal{P}_{A,B}$ (not necessarily edge disjoint) also satisfy $\sum_{e \in L} \rho(e) \geq 1$ without changing
the minimum, since it does not change the set of feasible $\rho$. 

It is possible to show that this last problem is a min-cut problem. 
A proof of this fact can be found in standard references in linear programming, e.g., Ref.~\cite{Schrijver:LP}.
We will take a slightly different approach here. 
We consider the integer version of \Eqref{tosat} where we further impose $ \rho \in \mathbb{Z}_{\geq 0}$. This does indeed give the equivalent optimal value as \Eqref{rhomin}, although this is not obvious and there could well have been an \emph{integrality gap.}
We will now demonstrate it is equivalent to the min-cut problem, or that the integrality gap vanishes for this problem. 
To construct the cut we consider an optimal $\rho$ and define:
\begin{equation}
r_A = \{ x : d_\rho(x,A) = 0 \}
\end{equation}
where $d_\rho(x,y)$ is the graph distance  induced by $\rho(e)$, which is the minimal distance between $x$ and $y$ as measured by the edge function $\rho(e)$. See Appendix~\ref{app:subG}. 

It is clear that
\begin{equation}
\rho(e) \geq \mathbf{1}_{\mu(r_A)}(e)
\end{equation}
by integrality.  Then for any other feasible $\rho'$:
\begin{equation}
 \sum_{e \in E } \rho'(e) w(e)  \geq \sum_{e \in E } \rho(e) w(e)  
 \geq w(\mu(r_A))
\end{equation}
Also any path from $A$ to $B$ must pass $\mu(r_A)$ so that $\rho'(e) = \mathbf{1}_{\mu(r_A)}(e)$ is feasible and we
have the opposite inequality. 

We now give a quick derivation of the RT formula for RTNs. We pick the cut $C:AB$, although this discussion is general. We compute the $m$-th R\'enyi entropy for $\rho_{AB}$ at large $\chi$ which involves finding the minimum of the free energy:
\begin{equation}
\label{mingg}
\min_{ g} \sum_{e = \{x,y\} \in E} w(e) d(g(x),g(y))
\end{equation}
where $g : V \rightarrow S_m$ and $g(AB) = \tau_m = (12 \ldots m)$ and $g(C) = \id$. Use an optimal solution to the flow problem $c(L)$, inserting \Eqref{cpfeas} into the objective function of \Eqref{mingg}:
\begin{align}
\sum_{e = \{x,y\} \in E} w(e) d(g(x),g(y)) & \geq 
 \sum_{L \in \widehat{\mathcal{P}}_{AB:C}} c(L) \sum_{e = \{x,y\} \in L}  d(g(x),g(y)) \\& \geq d(\tau_m, \id)  \sum_{L \in \widehat{\mathcal{P}}_{AB:C}} c(L) \\
&= (m-1) \mathcal{A}(AB:C)
\end{align}
where in the first inequality we used the feasibility condition, and in the second we used the triangle inequality for the Cayley metric repeatedly along the path.
The opposite inequality follows by considering $g(x) = \id$ for $x \in r_{AB:C}^c$ and $g(x) = \tau_m$ for $x \in r_{AB:C}$. Thus, R\'enyi entropies are all computed by
the same minimal area cut. This is the well known result that the entanglement spectrum of random tensor networks is flat, and determined by minimal cuts. 

We conclude this subsection by giving a proof that the integrality gap of \Eqref{program-intro} is determined by the ratio between $\mathcal{A}(A:B:C)$ and $\frac{1}{2}(\mathcal{A}(AB:C)+\mathcal{A}(B:AC)+\mathcal{A}(C:AB))$. First we show that the value of the following integer program 
\begin{align}
\begin{split}
    \min_\rho& \sum_{e\in E} \rho(e)w(e) \\
    \text{subject to}:&\quad \forall L \in \mathcal{P}_{A,B}\cup\mathcal{P}_{A,C}\cup\mathcal{P}_{B,C}: \sum_{e\in L}\rho(e)\ge 1
\end{split}
\end{align}
is given by the minimal triway cut. 
We restrict our discussion here to the case where every connected component of $G$ is connected to at least one of $A$, $B$, or $C$, since adding any disconnected components to $G$ does not change the area of an optimal triway cut. 

Consider an optimal $\rho$ and define $r_A = \{x\in V: d_\rho(x,A)=0\}$ and similarly define $r_B$ and $r_C$. $(r_A,r_B,r_C)$ must be disjoint otherwise we violate the path constraint. Furthermore $r_A\cup r_B \cup r_C = V$: If otherwise define $r_D\equiv V\backslash r_A \backslash r_B \backslash r_C$. $r_D$ must be connected to at least one of $r_A, r_B$ or $r_C$, since otherwise $r_D$ will be totally disconnected from the boundary. Without loss of generality suppose it is $r_A$. We have 
\begin{align}
\label{eq:trimin-ineq}
\begin{split}
   \rho(e) &\ge \left( \mathbf{1}_{\mu(r_A:r_B)} + \mathbf{1}_{\mu(r_A:r_C)}  + \mathbf{1}_{\mu(r_B:r_C)} + \mathbf{1}_{\mu(r_A:r_D)} + \mathbf{1}_{\mu(r_B:r_D)} + \mathbf{1}_{\mu(r_C:r_D)}  \right) (e) \\
   & \ge \left( \mathbf{1}_{\mu(r'_A:r_B)} + \mathbf{1}_{\mu(r_A':r_C)}  + \mathbf{1}_{\mu(r_B:r_C)} \right) (e) \equiv \rho'(e)
\end{split}
\end{align}
where the first inequality follows from integrality. In the second line we defined $r_A'=r_A\cup r_D$ and this inequality is strict for $e\in\mu(r_A:r_D)$. $\rho'(e)$ is clearly feasible since any path from $A$ to $B$ must pass through $\mu(r_A':r_B)$ and similarly for $B:C$ and $C:A$. Since $\mu(r_A:r_D)\neq \emptyset$, we have $\sum_e \rho (e) w(e) > \sum_e \rho' (e) w(e)$, which is a contradiction since we have assumed $\rho$ to be optimal.
Therefore $r_D=\emptyset$ and $r'_A=r_A$. Since $\rho$ is optimal we must have
\begin{align}
    \min_\rho \sum_{e\in E} \rho(e)w(e) = \sum_{e\in E} \rho'(e)w(e) = \mathcal{A}(r_A:r_B:r_C) \ge \mathcal{A}(A:B:C)
\end{align}
On the other hand from any minimal triway cut $(\alpha,\beta,\gamma)$ we can construct $\rho^{\prime\prime}(e)=\left( \mathbf{1}_{\mu(\alpha:\beta)} + \mathbf{1}_{\mu(\alpha:\gamma)}  + \mathbf{1}_{\mu(\beta:\gamma)} \right) (e)$ which is clearly feasible. Hence $\mathcal{A}(A:B:C) \ge \sum_{e} \rho'(e)w(e)$ and the value of program is equivalent to minimal triway cut.

If one relaxes the integer constraint of $\rho$, then it allows us to construct a new solution from any optimal $\rho$:
\begin{align}
\tilde{\rho}(e) = \frac{1}{2} \left( \mathbf{1}_{\mu(r_A:r_A^c)} + \mathbf{1}_{\mu(r_B:r_B^c)} + \mathbf{1}_{\mu(r_C:r_C^c)} \right)(e)
\end{align}
$\tilde{\rho}(e)$ violates the integer constraint but is still feasible.
Using a similar argument as above one can show that in this scenario we get
\begin{align}
\begin{split}
    \min_\rho \sum_{e\in E} \rho(e)w(e) &= \sum_{e\in E} \tilde{\rho}(e)w(e) = \frac{1}{2}\left(\mathcal{A}(r_A:r_A^c)+\mathcal{A}(r_B:r_B^c)+\mathcal{A}(r_C:r_C^c)\right) \\
    &\ge \frac{1}{2}\left( \mathcal{A}(A:BC) + \mathcal{A}(B:AC) + \mathcal{A}(C:AB) \right)
\end{split}
\end{align}
and the optimal value of the program is given by a sum of minimal cut areas. We have determined the optimal value of the program \Eqref{program-intro} before and after the linear relaxation,  which agrees with the right-hand side of \Eqref{eq:MG-bound}.

\subsection{Permutations, Partitions and Boolean variables}

\label{sec:ppb}

In \secref{sec:sum} we introduced a coarse graining procedure which takes permutations to partitions. We further coarse grain our partitions by blocking them into partitions of the blocks in $P(X)$. We can apply this procedure given any fixed permutation $g_0$, instead blocking with $P(g_0)$. We discuss this more general procedure here and prove some important bounds.

The collection of set partitions $P_N$ admits a natural lattice structure.
There is a partial order within $P_N$: Given two elements $p,q\in P_N$ we say $p\ge q$ if every subset of $p$ can be expressed as the union of some subsets in $q$, or in other words, $q$ is a ``finer'' version of $p$ formed by further dividing blocks in $p$.
The finest (smallest) element in $P_N$ is the identity permutation $\id \equiv \{\{1\},\ldots,\{N\}\}$, and the coarsest (largest) is $\{\mathbb{Z}_N\}$.
The \emph{join} of $p$ and $q$, denoted by $p\vee q$, is defined to be the least upper bound of $p$ and $q$, i.e. $p \vee q$ is the smallest element $x$ such that $x\ge p$ and $x\ge q$. Conversely, the \emph{meet} of $p$ and $q$, denoted by $p\wedge q$, is defined to be the greatest lower bound of $p$ and $q$, i.e. $p \wedge q$ is the largest element $y$ such that $y\le p$ and $y\le q$. 
In the lattice of $P_N$ the meet is simply the set of non-empty pairwise intersections of $p$ and $q$; whereas the join can be thought of as the partition that arises from the connected orbits generated by $p$ and $q$.
For further information on the lattice of set partitions we refer the reader to Appendix A of Ref.~\cite{Akers:2021pvd}.

Given some elements $g,g_0 \in S_N$ we define:
\begin{equation}
\label{eq:q_g0}
q_{g_0}(g) \equiv (P(g) \vee P(g_0)) / \sim \,\, \in\, P_{\#(g_0)}
\end{equation}
where $/\sim$ is the set quotient operation defined element wise on each block in the partition and using the equivalence $x \sim y$ if $x,y$ are in the same block of $P(g_0)$.

Recall the distance measure on the set of partitions $P_{N}$:
\begin{equation}
\label{eq:P_metric_2}
d(p,p') \equiv \#(p) + \#(p') - 2 \#(p \vee p')
\end{equation}
We verify some properties. 
\begin{lemma}
The distance on set partitions $d(p,p')$ is a metric: it is positive, symmetric, vanishes iff $p=p'$ and satisfies the triangle inequality. Additionally:
\begin{itemize}
\item[(a)] There is the estimate:
\begin{equation}
d(p_1,p_2) \geq | \#(p_1) - \#(p_2) |
\end{equation}
\item[(b)] For all partitions $r$ then:
\begin{equation}
\label{addr}
d(p_1, p_2) \geq d( p_1 \vee r, p_2 \vee r)
\end{equation}
\item[(c)] It bounds the Cayley distance by the associated set partitions $P(g)$, or coarse grained versions $q_{g_0}(g)$: 
\begin{equation}
\label{gqbound}
\#(g)  - \#(g \vee g')  \geq   \#(q_{g_0}(g))  - \#(q_{g_0}(g) \vee q_{g_0}(g'))  
\end{equation} 
and
\begin{equation}
\label{gPq}
d(g,g') \geq d( P(g), P(g')) \geq d ( q_{g_0}(g), q_{g_0}(g'))
\end{equation}
\end{itemize}
\end{lemma}
\begin{proof}
Firstly the distance is positive since $\#(p) - \#(p \vee p') \geq 0$ and similarly for $p \leftrightarrow p'$. 
Equality implies that $\#(p') = \#(p) = \#(p \vee p')$ which is only possible if $p' \leq p$ and also $p' \geq p$ which implies equality of the partitions. 
We have the triangle inequality:
\begin{equation}
d(p,p')  + d(p',p'') \geq d(p,p'') 
\end{equation}
which follows from semimodularity 
\footnote{Semimodularity here means $\#(p_1\wedge p_2)+\#(p_1\vee p_2)\ge \#(p_1) + \#(p_2)$.}:
\begin{align}
\label{eq:triangle_semimodular}
\begin{split}
&\#(p')-\#(p \vee p') - \#(p' \vee p'') + \#(p \vee p'') \\& \qquad \geq
\#( (p' \vee p) \wedge (p' \vee p'') ) -\#(p \vee p') - \#(p' \vee p'') 
+ \#(p \vee p'' \vee p') \geq 0
\end{split}
\end{align}
where the first inequality follows simply from $p' \leq (p' \vee p) \wedge (p' \vee p'')$
and $p' \vee p'' \leq p \vee p'' \vee p' $.

(a) Using  $\#(p_1 \vee p_2) \leq \#(p_1)$ or $\#(p_1 \vee p_2) \leq \#(p_2)$ we derive $d(p_1,p_2) \geq | \#(p_1) - \#(p_2) |$.

(b) For any $r \in P_N$:
\begin{align}
\begin{split}
&d(p_1, p_2) - d( p_1 \vee r, p_2 \vee r) \\
=&\#(p_1)+\#(p_2)-\#(p_1\vee r)-\#(p_2\vee r) + 2\#(p_1\vee p_2\vee r) -2\#(p_1 \vee p_2) \\
\ge &  \#(p_1) + \#(p_2) - \#( (p_1 \vee p_2) \wedge (p_1 \vee r) ) - \#( (p_1 \vee p_2) \wedge (p_2 \vee r) ) \geq 0
\end{split}
\end{align}
where we used semi-modularity in the first inequality\footnote{Specifically: $-\#( p_1 \vee p_2) - \#(p_1 \vee r) + \#(p_1 \vee p_2 \vee r) \geq - \#( (p_1 \vee p_2) \wedge (p_1 \vee r))$
and $1 \leftrightarrow 2$. } and we used $p_i \leq  (p_1 \vee p_2) \wedge (p_i \vee r)$ in the second inequality.

(c) In general we start with some $g \in S_N$. Then we will remove fine-grained topological information (roughly speaking, we lose information on the genus expansion) by
moving to partitions. We can do this using the general bound:
\begin{align}
d(g, g')  &= - d(e,g) + d(e, g') + 2(\#(g') - \#(g' \vee g) ) + 2 G_{g'} (g) \nonumber  \\
\label{gGbound}
&\geq \#(g) + \#(g') - 2 \#( g \vee g')  
= d(P(g), P(g'))
\end{align}
where $\#(g'\vee g)\equiv \#(P(g')\vee P(g))$ is the number of orbits generated by the joint action of $g$ and $g'$, and $G_{g'}(g)$ is the genus of the admissible surface associated for $g$ based over $g'$. The inequality then follows from the non-negativity of the genus. See theorem 7 of Ref.~\cite{Akers:2021pvd}.

Consider semimodularity applied to $p = P(g) \vee P(g_0)$ and $p' = P(g) \vee P(g')$ giving:
\begin{equation}
 \#(P(g))+ \#(p \vee p') \geq \#(p \wedge p')+ \#(p \vee p')  \geq \#(p) + \#(p') = \#(q_{g_0}(g)) + \#(g \vee g')
\end{equation}
where the first inequality follows since $P(g) \leq  p$ 
and $P(g) \leq  p' $ 
implying the same of the meet $p \wedge p'$.
Note that 
$\#(p \vee p') = \#(q_{g_0}(g) \vee q_{g_0}(g'))$,$\,\, \#(p) = \#(q_{g_0}(g))$ and $\#(p') = \#(g \vee g')$.
Thus, we derived the bound:
\begin{equation}
\#(g)  - \#(g \vee g')  \geq   \#(q_{g_0}(g))  - \#(q_{g_0}(g) \vee q_{g_0}(g'))  
\end{equation} 
which implies:
\begin{equation}
d( P(g), P(g')) \geq d ( q_{g_0}(g), q_{g_0}(g'))
\end{equation}
\end{proof}

Consider $q=q_X(g)\in P_{2n}$ for some $g\in S_{mn}$. 
We can further classify $q$ by the singlets in $q$.
A singlet is a block of size $1$ so there are $2n$ possible singlets. 
We define $\#_1(q)$ as the number of singlets in a given partition. 
We define $s(q)$ as the (unique) largest partition with singlets in the same location as the singlets of $q$.  That is $s(q) \geq q$ and $\#(s(q)) = \#_1(q) + 1 - \delta_{\#_1(q) , 2n}$. 
Also $\#_1(q) = \#_1(s(q))$.

We define the \emph{singlet distance} on the set of $s$ as:
\begin{equation}
\label{eq:S_metric}
d_1(s_1, s_2) = \#_1( s_1) + \#_1( s_2) - 2 \#_1(s_1 \vee s_2) \geq d(s_1, s_2)
\end{equation}
where the later inequality is only not saturated if one of $s_1$ or $s_2$ is $\id$ (but not both.)
We also have $d_1(s_1,s_2) \leq d(s_1,s_2) + 1$ with equality iff one and only one of $s_1,s_2$ is $\id$.

We can bound the difference in $q$ by $s$:
\begin{lemma}
\label{lem:qtos}
\begin{equation}
\label{ddd}
d(q_1,q_2) \geq \lceil d_1(s_1, s_2)  /2 \rceil  \geq d_1(s_1,s_2)/2
\end{equation}
where $s_i := s(q_i)$.
\end{lemma}

\begin{proof}
If $q_1$ and $q_2$ do not contain any singlet, then $s_1=s_2=\mathbb{Z}_{2n}$ and $d(s_1,s_2)=0$ so the estimation is true.
Now let $C_{12}$ be the common singlets of $q_1$ and $q_2$, $C_1$ be the singlets in $q_1$ that do not overlap with the singlets in $q_2$ and similarly for $C_2$.
Note that $d(q_1,q_2)\ge d(q_1\vee r, q_2\vee r)$ from \Eqref{addr}. We take $r=s_1\vee s_2$. 
We can express $r$ as the unique element that is fully connected in $K=\mathbb{Z}_{2n}\backslash C_1\backslash C_2\backslash C_{12}$ and singlets elsewhere.
See:
\begin{equation}
\begin{matrix}
    \includegraphics[scale=.3]{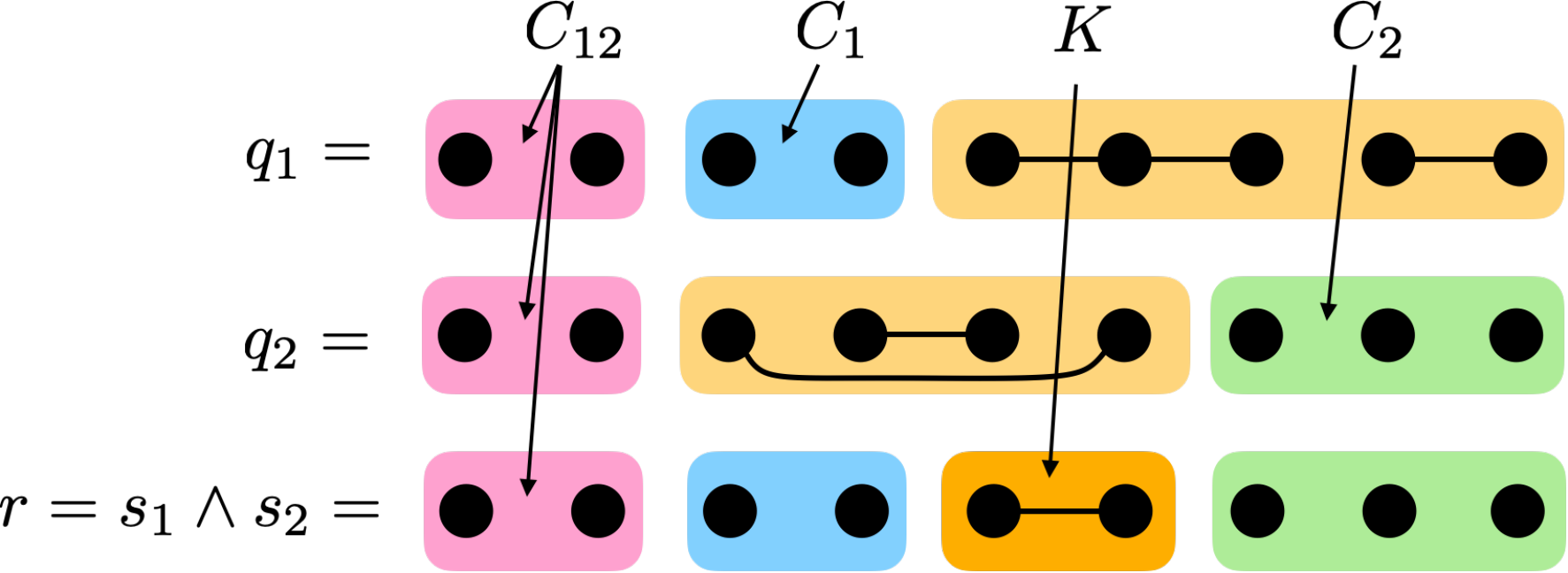}
\end{matrix}
\end{equation}
We can compute
\begin{align}
        d(q_1\vee r,q_2 \vee r) &= \#(q_1 \vee r) + \#(q_2 \vee r) - 2\#(q_1\vee q_2\vee r)\\
        &=|C_1|+|C_2|+\#(q_1\vee r)|_{K\cup C_2}+\#(q_2\vee r)|_{K\cup C_1}-2\#(q_1\vee q_2\vee r)|_{K\cup C_1\cup C_2}
    \nonumber
\end{align}
where we have used $\#(q)|_C$ to indicate the number of cycles of $q$ when restricted to a subset of elements $C\subset \mathbb{Z}_{2n}$.
Note that the common singlets $|C_{12}|$ cancels out in the calculation.
To proceed further, we define $C'_1\subset C_1$ to be the elements in $C_1$ that are not connected to $K$ via $q_2\vee r$ and similarly for $C_2'$.
Then we have
\begin{align}
 \#(q_1\vee r)|_{K\cup C_2} &= 1 + \#(q_1\vee r)|_{C_2'}   \\
 \#(q_2\vee r)|_{K\cup C_1} &= 1 + \#(q_2\vee r)|_{C_1'}  \\
 \#(q_1\vee q_2 \vee r)|_{K\cup C_1\cup C_2} &= 1 + \#(q_1\vee r)|_{C_2'} + \#(q_2\vee r)|_{C_1'}
\end{align}
which readily follows from the structure of $q_i\vee r$. See:
\begin{equation}
\begin{matrix}
    \includegraphics[scale=.3]{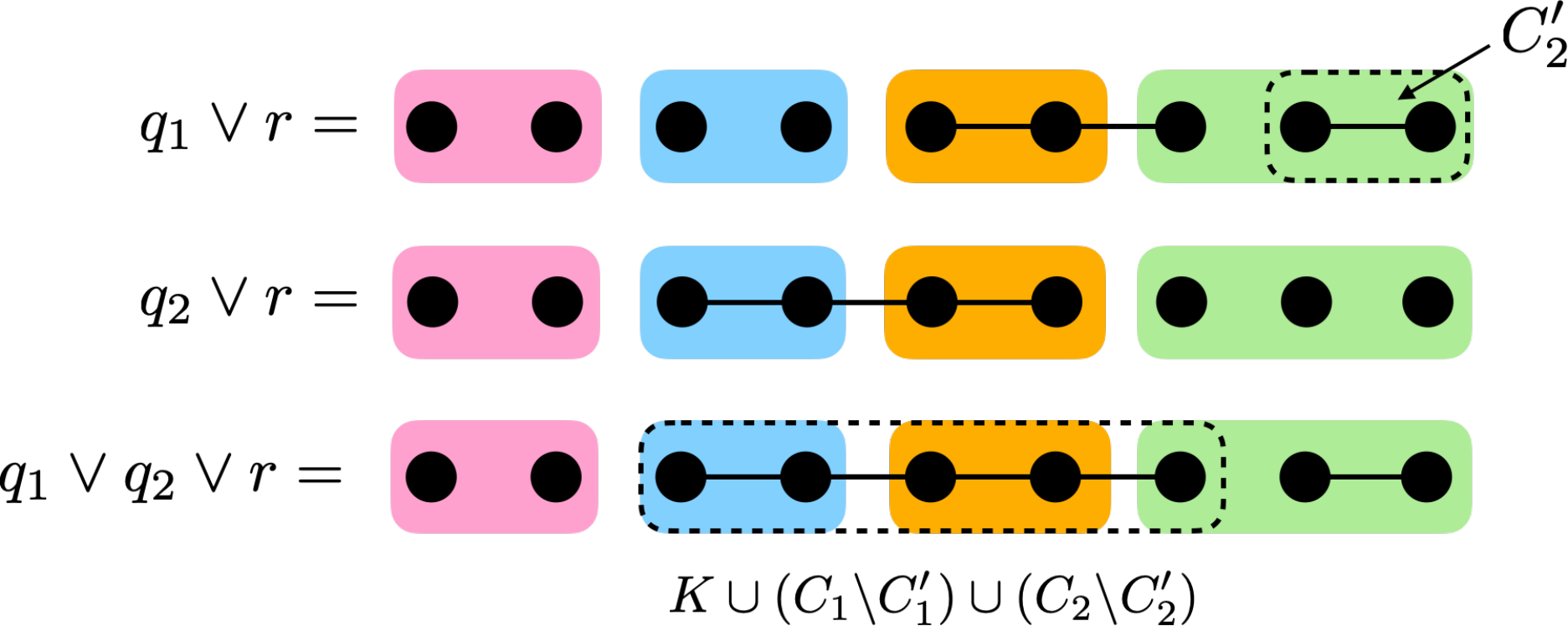}
\end{matrix}
\end{equation}
(In the example given here $C_1'=\emptyset$)
Hence,
\begin{equation}
    d(q_1\vee r, q_2 \vee r) = |C_1| + |C_2| -\#(q_1\vee r)|_{C_2'} - \#(q_2\vee r)|_{C_1'}
\end{equation}
Also note that $(q_2 \vee r)|_{C_1'}$ has no singlets
so that $ \#( q_2 \vee r)|_{C_1'} \leq \lfloor |C_1'|/2 \rfloor \leq \lfloor |C_1|/2 \rfloor$ since this number is maximized by forming the largest number of doublet blocks or doublet and a single triplet block in $C_1'$ (and similarly for $1 \leftrightarrow 2$.)
This gives the estimate:
\begin{equation}
d( q_1 , q_2 ) \geq \lceil |C_1|/2 \rceil +  \lceil |C_2|/2 \rceil \geq\lceil  (|C_1| + |C_2|)/2 \rceil
\end{equation}
Now $|C_1| + |C_2| = d_1(s_1, s_2)$ (see the proof of the Lemma following immediately after).
\end{proof}

Note that the set of $s$, $s(P_{2n})\subset P_{2n}$, forms a lattice (albeit not a sub-lattice of $P_{2n}$ \footnote{Since $s(P_{2n})$ is not closed under $\vee$.}), under the new meet and join defined by $s_1 \wedge_B s_2 \equiv s_1\wedge s_2$ and $s_1 \vee_{B} s_2 \equiv s( s_1 \vee s_2)$. 
Define the units $u_k\in s(P_{2n})$ to be the partition with one singlet at location $k$. That is, $u_k = \{\{k\},\mathbb{Z}_{2n} \backslash \{k\}\}$. 
It then follows that every element in $s(P_{2n})$ (other than $\mathbb{Z}_{2n}$) can be expressed as the meet of a string of $u^k$'s, i.e. $s=u_{k_1}\wedge\cdots\wedge u_{k_i}$.

The lattice $(s(P_{2n}),\vee_B,\wedge_B)$ is isomorphic to Boolean algebra $B_{2n}$ of bit-strings where singlets in $s$ are $0$'s, and non-singlets are $1$'s.
More specifically, identify $s$ (and thus $q$) using a binary variable $b^{k}$ for $k = 1, \ldots 2n$. 
That is:
\begin{equation}
\label{eq:def:b^k}
b^{k}(q) =\#_1(  s(q) \vee u_k) = \#_1(q \vee u_k)
\end{equation}
so that $b^{k}(q)=0$ if $q$ has a singlet at the $k$-th element and $b^{k}(q)=1$ otherwise. 
The bit-string $b=\{b^k\}_{k=1}^{2n}$ then forms a lattice where $\vee$ is the pair-wise or operation and $\wedge$ is the pair-wise and operation.
Then the singlet distance on $P_{2n}$ is simply the Hamming distance on bit-strings:
\begin{lemma}
\begin{equation}
    \label{d1s}
d(b_1, b_2) \equiv \sum_{k=1}^{2n} | b^k_1 - b^k_2 | = d_1(s_1,s_2)
\end{equation}
    where $b^k_{1,2} = b^k(s_{1,2})$.
\end{lemma}
\begin{proof}
Consider the $k$-th element in a partition of $\mathbb{Z}_{2n}$.
Then $s_1\vee s_2$ has a singlet at position $k$ iff $s_1$ and $s_2$ both have a singlet at position $k$.
Let $S_i$ be the set of singlets in $s_i$, then we have
    \begin{align}
    \begin{split}
        d_1(s_1,s_2) &= \#_1(s_1) + \#_1(s_2) - 2\#_1(s_1\vee s_2) \\
        &= |S_1| + |S_2| - 2|S_1\cap S_2| = |C_1| + |C_2| 
    \end{split}
    \end{align}
    where $C_1$ are the elements in $S_1$ that are not contained in $S_2$ and likewise for $C_2$.
    By definition $C_1\cap C_2=\emptyset$ and 
    \begin{align}
        |C_1|+|C_2| = |C_1\cup C_2| = \sum^{2n}_{k=1}|b^k_1-b^k_2|
    \end{align}
    since $|b^k_1-b^k_2|=1$ signals $k\in C_1\cup C_2$ and $|b^k_1-b^k_2|=0$ otherwise.
\end{proof}
\begin{corollary}
\label{lem:qtob}
\begin{equation}
    d(q_1,q_2) \ge \lceil d(b_1,b_2)/2\rceil \ge d(b_1,b_2)/2
\end{equation}
where $b^k_{1,2}=b^k(s_{1,2})$.
\end{corollary}
\begin{proof}
    Substitute \Eqref{d1s} into Lemma~\ref{lem:qtos}.
\end{proof}
\begin{remark}
$d(b_1,b_2)$ clearly satisfies all the properties of a metric. 
We will also sometimes work with $\lceil d(b_1, b_2)  /2 \rceil$ which also satisfies the properties of a metric, since $\lceil a \rceil + \lceil b \rceil \geq \lceil a + b \rceil$. 
\end{remark}

We need one final result, which generalizes the triangle inequality and was already proven in Ref.~\cite{Akers:2021pvd}.

\begin{lemma}
\label{lem:tAtB}
Consider two partitions $t_A, t_B \in P_N$ and consider the bi-partite graph $G$ formed from $\#(t_A)$ black vertices and $\#(t_B)$ white vertices joined
with $t_A \wedge t_B = \id$ edges for each block in $t_A$ and $t_B$ that intersect. If $G$ is a cycle graph then:
\begin{equation} 
\label{boundsat}
d(t_A, q) + d(q, t_B) \geq d(t_A,t_B) + 2 (1-\delta_{ \#_1(q), 0} )
\end{equation}
for any $q \in P_N$. 
\end{lemma}
\begin{proof}
\begin{figure}[ht]
    \centering
    \includegraphics[scale=.35]{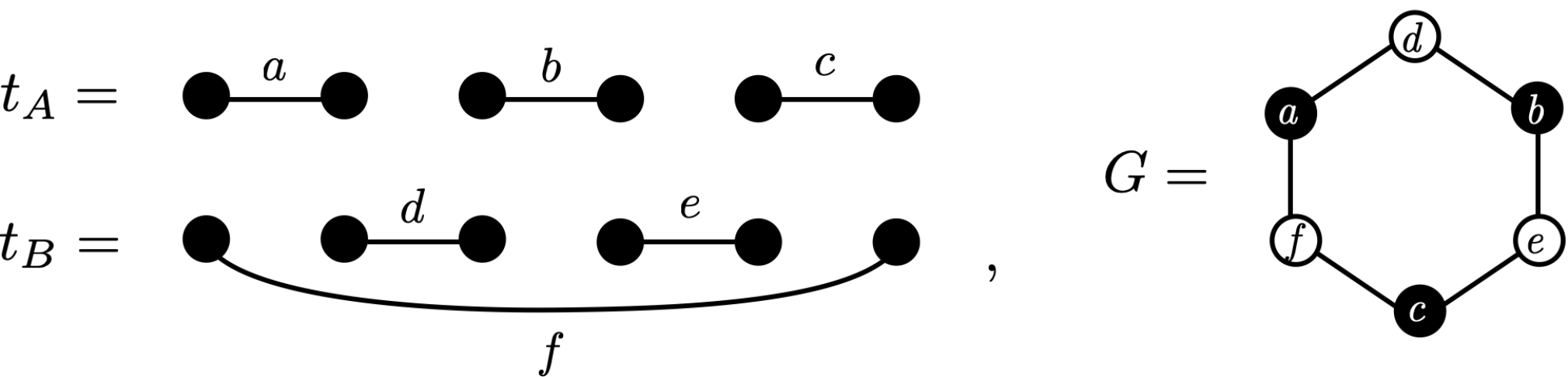}
    \caption{Two partitions $t_A,t_B\in P_N$ with $t_A\wedge t_B=\id$ and the bi-partite graph $G$ constructed from the subsets (labeled $\{a,b,\ldots,f\}$) of $t_A$ and $t_B$.}
    \label{fig:modular_pair}
\end{figure}
An example graph is shown in \figref{fig:modular_pair}. Blocks in $t_A$ can be found by combining the edges that intersect a fixed black vertex. 
From this figure it is clear that $t_A$ and $t_B$ are made of doublets.  Because it has a single cycle $t_A$ and $t_B$ fail to be a modular pair\footnote{A modular pair saturates the semimodularity condition.} by $1$:
\begin{equation}
d(t_A,\id) + d(t_B,\id) - d(t_A,t_B) = 2(\#(t_A \vee t_B) + \#(t_A \wedge t_B)- \#(t_A) - \#(t_B) ) = 2
\end{equation}
since $\#(t_A \vee t_B) = 1$. 
A modular pair $(t_A',t_B)$ is associated to a tree graph \cite{birkhoff1967lattice}, and so any unit $u_k$ will break a single edge $t_A' = t_A \wedge u_k$ of the cycle
and produce a tree with:
\begin{equation}
d(t_A',\id) + d(t_B,\id) - d(t_A',t_B) =  2(\#(t_A' \vee t_B) + \#(t_A' \wedge t_B)- \#(t_A') - \#(t_B) ) = 0
\end{equation}
where we used the fact that $t_A' \wedge t_B = u_k \wedge \id = \id$. 
Thus if $q = q \wedge u_k$ then:
\begin{align}
\begin{split}
d(t_A, q) + d(t_B,q) &= d(t_A', q) + d(t_B,q) \geq d(t_A',t_B) = d(t_A',\id) + d(t_B,\id) \\
& = d(t_A,\id) + d(t_B,\id) = d(t_A,t_B) + 2
\end{split}
\end{align}
Since this is true for any $u_k$ it is true if $\#_1(q) \neq 0$ as required. 
\end{proof}
\begin{remark}
We will apply this Lemma to $t_A = q_A$ and $t_B = q_B$ where these partitions satisfy the required properties, as can be seen by their explicit form in Appendix~\ref{app:specific}.
While the above example only slightly generalizes the case of interest, it points towards the basic structure that makes our results work. In particular, note
that $\#_1(\id) \neq 0$ and indeed the bound for $q=\id$, \Eqref{boundsat}, is saturated in this case. So the bounds we derive (here for partitions) are tight, which is important for the possibility
of forming a collapsing chain in Theorem~\ref{thm:RPBI}.
It is also important that the pair $(t_A,t_B)$ does not form a tree (they are not a modular pair) since otherwise we would not get minimal triway cuts -- we would simply get minimal cuts.
\end{remark}

\section{Proof of the main theorem}
\label{sec:main}
In this section, we establish our proof for the main results (Theorem~\ref{thm:RPBI} and Theorem~\ref{thm:uniqueness}) of the paper. First, in \secref{sec:ptop}, we show that the permutation group optimization problem $R$ (Definition~\ref{def:R}) can be coarse grained into a set partition optimization problem $Q$ (Definition~\ref{def:Q}). Next, in \secref{sec:ptoint} we bound $Q$ by a mixed Boolean-integer program $B$ (Definition~\ref{def:B}), and then an integer program $I$ (Definition~\ref{def:I}). We these results at hand, we relate the value of $I$ to multiway cuts in \secref{sec:inttomulti}, thus proving the collapsing chain $R-\mathcal{A}(AB:C)\ge Q \ge B \ge I \ge \mathcal{A}(A:B:C)$ as required by Theorem~\ref{thm:RPBI}. We prove the uniqueness of the solution in \secref{sec:uniqueness}.
Note that this section only establishes our results on even integer $m>2$.  We will deal with the problem of analytically continuing $m\to 1$ in \secref{sec:cont}.

\subsection{From permutations to partitions}
\label{sec:ptop}

We seek:
\begin{equation}
R = \min_{g} R(g)\,, \qquad R(g) =  \sum_{ e=\{ x, y\} \in E}w(e) d(g(x),g(y))
\end{equation}

Consider a minimal cut $r_{AB} \subset V$ associated to the division $AB:C$. Define $C' = \{x\in r^c_{AB}:\{x,y\}\in \mu(r_{AB})\}$ to be the vertices in $r^c_{AB}$ that border $r_{AB}$. 
Define a new amputated graph $G' = (V',E') = (r_{AB} \cup C' ,E_G [r_{AB}])$. 
We can estimate $R$ using a new model, written in terms of the coarse grained $q_X(g)$ with respect to element $X$, on this new graph:

\begin{lemma}
\label{lem:RtoP'}
We have the following estimate: 
\begin{equation}
\label{lemff}
R(g) \geq \mathcal{A}(AB:C) d(X,\id) +  Q'(q)  
\end{equation}
where $q : V' \rightarrow P_{2n}$ defined via $q(v) \equiv q_X(g(v))$ and where:
\begin{equation}
Q'(q) = \sum_{e = \{x,y\} \in E'} w(e) d( q(x), q(y)) 
\end{equation}
\end{lemma}
\begin{remark}
\label{remark:RtoP'}
Minimizing over $g$ gives $R \geq \mathcal{A}(r_{AB}) d(X,\id) + Q'$ with $Q'\equiv \min_q Q'(q)$, where the boundary conditions on $q$ for this later model is such that vertices in $A(B)$ have a fixed permutation $q = q_A(q_B)$ and the permutations on $C'$ are fixed as $q=\id$.
To pass from the program $Q'$ (on the amputated graph $G'$) to the program $Q$ (on the original graph $G$), as required in Definition~\ref{def:Q} and Theorem~\ref{thm:RPBI}, we simply need
to show that $Q'\ge Q$. This is proven in Lemma~\ref{lem:RtoP} after the proof of Lemma~\ref{lem:RtoP'}. 
\end{remark}

\begin{proof}
Given the minimal cut $r_{AB}$ we can consider an optimal solution to the linear program $c(L)$ in \Eqref{program:cP} for edge-disjoint paths:
\begin{equation}
\mathcal{A}(r_{AB})  = \sum_{L \in \widehat{\mathcal{P}}_{AB,C}} c(L) 
\end{equation}
where for all $e \in E$ then $w(e) \geq \sum_{L} c(L) \mathbf{1}_{L}(e)$.
Consider a single such path $L\in\mathcal{P}_{AB,C}$.
Such a path must pass through the minimal cut $\mu(r_{AB})$ at least one time.
Denote the vertex $\nu$ to be the first vertex along $L$ that connects $r_{AB}$ to $C'$.
That is for the first edge $e$ in $L \cap \mu(r_{AB})$ we have $e = \{ \nu', \nu \}$ where
$\nu' \in C'$. Since $C'$ lies on the minimal cut, this is guaranteed to exist, see \figref{fig:proof-partition}.
We split up $L=L_{AB}\sqcup L_{C}$ into two connecting paths at the common vertex $\nu$, where $L_{AB} \in \widehat{P}_{AB,\nu}$ and $L_{C} \in \widehat{P}_{\nu,C}$.
We estimate the contribution of $L$ to be:
 \begin{align} 
 \label{great}
\sum_{e = \{x,y\} \in L} d (g(x),g(y))
 &=\sum_{e =\{x,y\} \in L_{C} } d(g(x),g(y))+  \sum_{ e =\{x,y\}  \in L_{AB}}  d(g(x),g(y))
 \\
& \geq  d(g(\nu), \id) + \sum_{e =\{x,y\}  \in L_{AB}} d(P(g(x)),P(g(y)))
 \label{great2} 
\end{align}
\begin{figure}
\centering
\includegraphics[scale=.35]{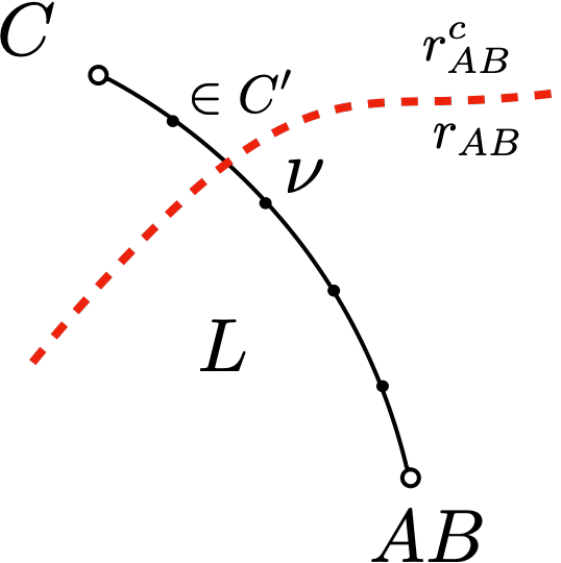}
\caption{We split up a path $L\in \widehat{P}_{AB,C}$ at the cut surface $\mu(r_{AB})$, through a vertex $\nu\in r_{AB}$ that lies immediately inside $r_{AB}$.}
\label{fig:proof-partition}
\end{figure}
To arrive at \Eqref{great2} we have used \Eqref{gPq} for the part of the path that intersects $r_{AB}$ and then repeated uses of the (Cayley distance)  triangle inequality in the complement.

We can re-arrange the sum: 
\begin{align}
\sum_{e =\{x,y\} \in L_{AB}} d(P(g(x)),P(g(y))) 
= \,&\#(g_{A,B}) - 2\#(g_{A,B} \vee g(x_2)) + \#(g(\nu))  \\ & + \sum_{\alpha=2}^{|L_{AB}|}  2\left( \#(g(x_\alpha)) - \#(g(x_\alpha) \vee g(x_{\alpha+1})) \right)
\label{ssum}
\end{align}
where $x_\alpha$ is the sequence of vertices that connects the path $L$. With $x_1 \in AB$ and $x_{|L_{AB}|+1} = \nu$. 
Note that if there are no edges in $L_{AB}$ then the sum \Eqref{ssum} does not contribute.
Now use the bound \Eqref{gqbound} along with the equality $\#(g_{A,B}) = \#(q_X(g_{A,B}))$ and $\#(g \vee g_{A,B}) = \#( q_X(g) \vee q_X(g_{A,B}))$, we turn all $g$'s in the sum of \Eqref{ssum}
into $q_X(g)$'s plus the remainder $\#(g(\nu)) - \#(q_X(g(\nu)))$. That is:
\begin{align}
\label{finqq}
\sum_{e=\{x,y\} \in L_{AB}} d(P(g(x)),P(g(y)))
\geq \#(g(\nu)) - \#(q(\nu)) + 
\sum_{e=\{x,y\} \in L_{AB}} d(q(x),q(y))
\end{align}
(recall that $q_X(g(x)) \equiv q(x)$).
Hence we have the following estimate:
\begin{align}
\label{great3}
    \sum_{e=\{x,y\}\in L} d(g(x),g(y)) \ge d(\id,X) + d(\id,q(v))+\sum_{e=\{x,y\}\in L_{AB}} d(q(x),q(y))
\end{align}
where we have used the identity
\begin{equation}
d(\id,g(\nu)) + \#(g(\nu)) - \#(q(\nu)) = nm -  \#(q(\nu)) = d(\id,X) + d(\id,q(\nu))
\end{equation}
Now write:
\begin{align}
 \label{firstlast}
 R(g) &=  \sum_{e =\{x,y\}} \left(w(e) -  \sum_{ L} c(L) \mathbf{1}_{L}(e) \right) d(g(x),g(y))  
 + \sum_{L} c(L) \sum_{e =\{x,y\} \in L} d(g(x),g(y))  \\
 &\ge \sum_{e =\{x,y\} \subset r_{AB}}\left(w(e) -  \sum_{ L} c(L) \mathbf{1}_{L}(e) \right) d(q(x),q(y)) \nonumber \\
 &\qquad + \sum_L c(L) \left( d(\id,X) + d(\id,q(v))+\sum_{e=\{x,y\}\in L_{AB}} d(q(x),q(y))\right)
 \end{align}
where we applied \Eqref{great3} to all paths weighted by $c(L)$ in the last term of \Eqref{firstlast}, along with \Eqref{gPq} in the first term of \Eqref{firstlast} on the edges that are entirely inside $r_{AB}$ and finally dropping all other edges in $r_{AB}^c$ (note that the bracketed term in the first part of \Eqref{firstlast} is positive
by feasibility of $c(L)$). Using:
\begin{equation}
\sum_L c(L) d(\id,X) = \mathcal{A}(AB:C) d(\id,X)
\end{equation}
we arrive at \Eqref{lemff}.
\end{proof}

\begin{remark}
We will make use of the saturation conditions for the inequality \eqref{lemff} in our proofs for uniqueness of the solution. It is useful to record them here for reference later. For edges $e$ in $r_{AB}^c$ the saturation of \Eqref{firstlast} requires:
\begin{align}
    \sum_{e=\{x,y\}\in r_{AB^c}} \left(w(e)-\sum_L c(L)\mathbf{1}_L(e)\right)d(g(x),g(y)) = 0
\end{align}
We also have, from \Eqref{great2},
\begin{align}
    \sum_{e=\{x,y\}\in L_C}d(g(x),g(y)) = d(g(v), \id)
\end{align}
where $L_C\in \mathcal{P}_{v,C}$ is a subpath of $L\in \mathcal{P}_{AB,C}$ with $c(L)>0$. $L_C$ connects a vertex $v$ immediately inside $r_{AB}$ to $C$. 
For the edges $e$ in $r_{AB}$, saturation of \Eqref{firstlast} requires $\sum_{e\in r_{AB}} \left(w(e)-\sum_L c(L)\mathbf{1}_L(e)\right)d(g(x),g(y)) 
    = \sum_{e\in r_{AB}} \left(w(e)-\sum_L c(L)\mathbf{1}_L(e)\right)d(q(x),q(y))$.
Since $d(g(x),g(y))\ge d(q(x),q(y))\ge 0$ by \Eqref{gPq}, this condition holds locally, i.e.,
\begin{equation}
    d(g(x),g(y))=d(q(x),q(y))
\end{equation}
for all edges $e=\{x,y\}$ where $w(e)-\sum_L c(L)\mathbf{1}_L(e)>0$. If the minimal cut $r_{AB}$ is unique, then it holds for all $e$ strictly inside $r_{AB}$.
\end{remark}

\begin{lemma}
\label{lem:RtoP}
    Consider the program
    \begin{equation}
        Q=\min_q Q(q),
    \end{equation}
    on the original graph $G$,
    where we minimize over all $q:V\to P_{2n}$ with $q(A)=q_A, q(B)=q_B$ and $q_C = \id$. Then we have
    \begin{equation}
        R \ge \mathcal{A}(AB:C)d(X,\id) + Q
    \end{equation}
\end{lemma}
\begin{proof}
Consider the program $Q'=\min_{q} Q'(q)$, as defined in Lemma~\ref{lem:RtoP'} and Remark~\ref{remark:RtoP'}. Let $q'$ be an optimal solution to this problem. 
We construct from $q'$ a feasible solution $q:V\to P_{2n}$ to the $Q$ problem by 
\begin{align}
    q(v) = 
    \begin{cases}
    q'(v), \quad &v\in V'    \\
    \id, \quad &v\in V \backslash V'
    \end{cases}
\end{align}
It is clear that $q$ satisfies all the boundary conditions and $Q(q)=Q'(q')$, since setting $q(v)=\id$ on the new vertices does not increase the free energy.
Minimizing over all $q:V\to P_{2n}$ we have $Q'>Q$.
Then from Remark~\ref{remark:RtoP'} we have $R\ge \mathcal{A}(AB:C)d(X,\id) + Q' \ge \mathcal{A}(AB:C)d(X,\id) + Q$.
\end{proof}

\subsection{From partitions to integer program}
\label{sec:ptoint}

In this section we move from partitions $q$, to Boolean variables $b$, to an integer program. 

\begin{lemma}
\label{lem:checkfeas}
For all paths $L \in  \mathcal{P}_{A,B}$  (not necessarily edge disjoint) we have the estimate:
\begin{equation}
\label{difflr}
\sum_{e =\{x,y\} \in L} d(q(x),q(y)) \geq 2(n - \delta_{b_L,b_{AB}} )\,, \qquad \sum_{e = \{x,y\}\in L}  d(q(x),q(y)) \in 2 \mathbb{Z}
\end{equation}
where $b_{AB} = 11\cdots 1\in B_{2n}$ and
\begin{equation}
b_L = \bigwedge_{x \in V[L]} b(x),
\end{equation}
and where $b:V \to B_{2n}$ is defined via $b(v)^k = b^k(q(v))$.
\end{lemma}
\begin{proof}
If $b_L = b_{AB}$ then we must have $b(x)=b_{AB}=11\cdots 1$ for the entire path, so that $\#_1(q(x)) = 0$ for all $x\in L$ (since singlets are preserved through the meet operation). We use
the triangle inequality repeatedly to show:
\begin{equation}
\sum_{e =\{x,y\} \in L} d(q(x),q(y)) \geq d(q_A,q_B) = 2(n-1)
\end{equation}
Conversely, if $b_L \neq b_{AB}$ then there must be some $x\in V$ along the path such that $\#_1(q(x)) \neq 0$.
We apply the triangle inequality about this point:
\begin{equation}
\sum_{e =\{x,y\} \in L} d(q(x),q(y)) \geq d(q_A, q(x)) + d(q(x), q_B) \geq 2n
\end{equation}
where we used the improved triangle inequality derived in Lemma~\ref{lem:tAtB}, and which applies when $q(x)$ has a singlet somewhere. 

More generally the deficit in the triangle inequality is always an even integer:
\begin{equation}
d(q_1,q_2) + d(q_2,q_3) - d(q_2,q_3) = 2( \#(q_2 \vee q_3) + \#(q_2) - \#(q_1 \vee q_2) -\#(q_2 \vee q_3) )
\end{equation}
 Thus the difference in the left and right-hand side of the inequality in \Eqref{difflr} must
be an even integer. 
\end{proof}

We have the estimate in terms of an integer program:
\begin{lemma}
\label{lem:PtoB2}
Given some fixed $b: V\to B_{2n}$, define the integer program:
\begin{align}
\label{linstate}
\begin{split}
B(b) \equiv & \min_{r}   \sum_{e \in E}  w(e) r(e)  \\
&{\rm subject \,\, to}\, \quad \forall L \in \mathcal{P}_{A,B}\,\, : \sum_{e \in L} r(e) =  2(n-\delta_{b_L,b_{AB}})  + 2 \mathbb{Z}_{\geq 0}\\
&{\rm and \,\, }\, \qquad \forall e =\{x,y\} \in E \,\, : r(e) \geq \left\lceil \frac{1}{2} d(b(x),b(y))  \right\rceil 
\end{split}
\end{align}
where $r  \in \mathbb{Z}_{\geq 0}$.  
Then for any $q: V\to P_{2n}$ satisfying the boundary condition $q(A)=q_A,~ q(B)=q_B$ and $q(C)=\id$, we have the bound $ Q(q) \geq B(b^k(q))$.
\end{lemma}
\begin{proof}
We simply consider $r(e) = d(q(x),q(y))$ for $e =\{x,y\}$.
We use the bounds in Lemma~\ref{lem:checkfeas} and Corollary~\ref{lem:qtob} to check feasibility for the $B(b)$ problem. 
\end{proof}

\begin{corollary}
    \label{lem:PtoB}
    Consider the program:
    \begin{align}
        B\equiv \min_b B(b),
    \end{align}
    where we now minimize over all $b:V\to B_{2n}$ with $b(A)=b(B)=b_{AB}=11\ldots1$ and $b(C)=00\ldots0$.
    Then we have $Q\ge B$.
\end{corollary}

\begin{proof}
Consider an optimal solution $q$ to program $Q$. Then from Lemma~\ref{lem:PtoB2} it is clear that $B(b^k(q))\le Q(q)$ is feasible. Minimizing over all $b$ then gives $Q\ge B$.
\end{proof}

We will now introduce an edge variable $\rho(e) =   \left\lceil \frac{1}{2} d(b(x),b(y))  \right\rceil $ for $e = \{x,y\}$.  We note that we can write:
\begin{equation}
\label{sPrho}
\delta_{b_L,b_{AB}} = \delta_{\rho(L), 0},
\end{equation}
where we recall the shorthand notation $\rho(L)\equiv \sum_{e\in L}\rho(e)$.
Thus we have the new integer program:
\begin{lemma}
\label{lem:BtoI}
For integer $n \geq 2$, consider the integer non-linear program:
\begin{align}
\label{nlip}
\begin{split}
I =  & \min_{\rho, \sigma}    \sum_{e \in E}  w(e) (\sigma(e) + \rho(e)) \\ 
&{\rm subject \,\, to}\, \quad \forall L \in \mathcal{P}_{A,B}\,\, : \sum_{e \in L}( \sigma(e) + \rho(e)) \in   2(n - \delta_{\rho(L), 0})  + 2 \mathbb{Z}_{\geq 0} \\ 
&{\rm and \,\, }\, \qquad \forall L \in \mathcal{P}_{AB,C} \, \, :  \sum_{e \in L} \rho(e) \geq n  
\end{split}
\end{align}
for $ \rho, \sigma \in \mathbb{Z}_{\geq 0}$.
Then $\min_q Q(q) \geq \min_b B(b) \geq I$.
\end{lemma}
\begin{proof}
 Again we just have to check feasibility with $\rho(e) =  \left\lceil \frac{1}{2} d(b(x),b(y)) \right\rceil $
 and $\sigma(e) = r(e) -  \rho(e) \geq 0$, which again follows from Lemma~\ref{lem:checkfeas}, Lemma~\ref{lem:qtos} and \Eqref{sPrho}. We also need
 repeated use of the triangle inequality for the metric $ \left\lceil \frac{1}{2} d(b_1,b_2) \right\rceil$:
 \begin{equation}
 \sum_{e \in L} \rho(e) \geq \left\lceil \frac{1}{2} d(b_{AB},b_C ) \right\rceil = n
 \end{equation}
 for all $L \in \mathcal{P}_{AB,C}$. 
\end{proof}

\begin{remark}
In passing from the Boolean program $B$ to the integer program $I$, we have split the integer variable $r(e)$ into two variables $\rho(e)$ and $\sigma(e)$. 
In particular, the region associated to $\rho=0$ plays a special role in the program, as the constraint on a path is weaker if it stays entirely within the said region. It turns out that this region corresponds to the ``squeezed entanglement wedge'' (backreacted EW accommodating the non-zero tension $t_{A:B}$) in the solution to the reflected entropy optimization problem.
The $\sigma$ variable can be set to be non-zero only inside the $\rho=0$ wedge, which sources a total of $2(n-1)$ cuts separating $A$ and $B$. We will see that these cut surfaces collapse to a single domain wall corresponding to the (squeezed) EW cross-section for an optimal solution.
\end{remark}

\subsection{From integer program to multiway cuts}
\label{sec:inttomulti}
The main theorem we would like to prove here is that the program \Eqref{nlip} is equivalent to a multiway cut problem:
\begin{theorem}
\label{thmttt}
The minimum of the non-linear integer program \Eqref{nlip} is achieved by an optimal solution to the multiway cut problem for $A:B:C$
with $t_{A:B} = 2(n-1)$ and $t_{A:C} = t_{B:C} = n$. 
\end{theorem}

We do this in two steps. In the first step we prove the existence of a subgraph $G' = (V',E')$ with $A,B \subset V'$ such that either (i) $A$ and $B$ are disconnected 
on $V'$ and there is an optimal solution to \Eqref{nlip} where $\rho(e) = \sigma(e) = 0$ for all $e \subset V'$ or (ii) $A$ and $B$ are connected and there is in optimal solution $(\rho,\sigma)$
to \Eqref{nlip} where $\rho(e) = 0 $ and $\sigma$ can be described as a set of $2(n-1)$ cuts separating $A,B$ for all $e \subset V'$.
These solutions then seed a multiboundary/intersecting cut problem in the remaining graph made from the remaining edges: $E^c = E\backslash E'$ and vertices  $V^c \equiv (V \backslash V') \cup (AB)'$
where $(AB)' = \{x\in V':x\in V_G[\mu(V')]\}$ are the vertices in $G'$ that lie on the cut surface $\mu(V')$.
Both cases (i) and (ii) will be treated together, using a vertex valued variable $k(x)$ instead of an edge variable.  In the second step we solve this multiboundary cut problem. 

More explicitly, the first step is the following Lemma:

\begin{figure}[ht]
    \centering
    \includegraphics[scale=.4]{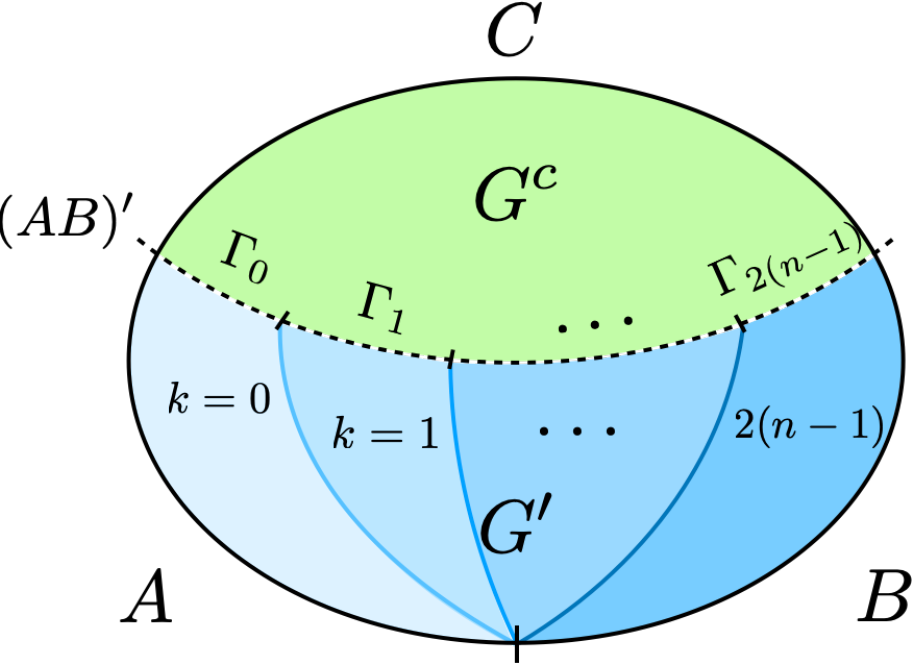}
    \caption{The original graph $G$ is split into two subgraphs $G'$ and $G^c$ that share common vertices $(AB)'$ (dashed line). On $G'$ we define a vertex variable $k$ associated to $\sigma$, whose value on $(AB)'$ seeds an $\ell$-intersecting cut problem on the complementary graph $G'$ with boundary vertices $\Gamma_k$.}
    \label{fig:G'-Gc}
\end{figure}

\begin{lemma}
\label{lem:subG}
There exists a subgraph $G' = (V',E')$ of $G$ 
where $A,B \subset V'$ such that there is an optimal solution $(\rho,\sigma)$ to \Eqref{nlip} where:
\begin{enumerate}
\item There exists a function $k: V' \rightarrow \mathbb{Z}_{\geq 0}$ with $k(x) \leq 2(n-1)$ such that:
for all edges $e  \in E'$ we have:
\begin{equation}
\sigma(e) = |k(x) - k(y) | \,, \qquad \rho(e) =0
\end{equation}
where $k(x) = 0$ for $x\in A$ and $k(x) = 2(n-1)$ for $x\in B$.
\item For the remaining edges we consider the $\ell$-\emph{intersecting cut problem} with $\ell=2(n-1)$  (defined immediately below)
on the complementary reduced subgraph $G^c = (V^c, E^c)$
, with boundary vertices $\Gamma_k = \{ x \in  (AB)' : k(x) = k \}$  and $C$. Where the optimal minimum $M$ of this later problem bounds:
\begin{equation}
\label{IgeqM}
I \geq M  + w( \mu( V')) + \sum_{e  =\{x,y\} \in E'} w(e) |k(x) - k(y) |
\end{equation}
\end{enumerate}
\end{lemma}
See \figref{fig:G'-Gc} for an exemplary configuration.
We prove this in \secref{sec:struct}.
Recall the notation, used in \Eqref{IgeqM}, $f(E') = \sum_{e \in E'} f(e)$ for some function $f$ on the edges and some subset $E' \subset E$.

\begin{definition}[$\ell$-intersecting cut problem]
\label{def:kint}
Given an integer $\ell \geq 1$ and weighted graph $G = (V,E)$ with $\ell+2$ sets of boundary vertices $\{\Gamma_k; k = 0 \ldots \ell\}$ and $C$, define the following $k$-\emph{intersecting cut problem}:
\begin{align}
\begin{split}
M \equiv  &\min_{\varrho} M(\varrho) \qquad M(\varrho) =   \sum_{e \in  E}  w(e) \varrho(e) \\ 
&{\rm subject \,\, to}\, \quad \forall k,k' = 0, \ldots \ell \,\, \forall L \in \mathcal{P}_{\Gamma_k ,\Gamma_{k'}}\,\, : \varrho(L) \in  |k-k'| + \mathbb{Z}_{\geq 0} \\ 
&{\rm and \,\, }\, \qquad \forall k = 0, \ldots \ell \,\, \,\, \forall L \in \mathcal{P}_{\Gamma_k,C} \, \, : \varrho(L) \in \ell/2  + \mathbb{Z}_{\geq 0}
\end{split}
\end{align}
where $\varrho(e) \in \mathbb{Z}_{\geq 0}/2$. 
\end{definition}

\begin{figure}[ht]
    \centering
    \includegraphics[scale=.4]{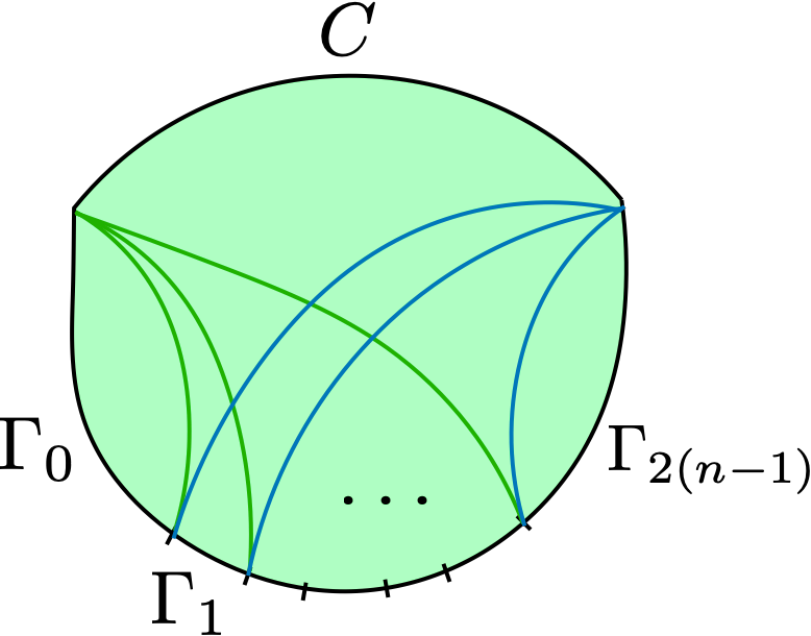}
    \caption{The solution to the $\ell$-intersecting cut problem as described in Lemma~\ref{lem:optk}. The dark green curves represent the cut surfaces of $\alpha_k$ and the dark blue curves represent the cut surfaces of $\beta_k$.}
    \label{fig:Gc-sol}
\end{figure}
We solve this problem recursively in \secref{sec:intcut}. The result is:
\begin{lemma}
\label{lem:optk}
The optimal value of the $\ell$-\emph{intersecting cut problem} is:
\begin{equation}
M = \frac{1}{2} \sum_{k=0}^{\ell-1} \left( w(\mu(\alpha_k))  + w(\mu(\beta_k))  \right)
\end{equation}
where $\alpha_k$ is a minimal cut for $(\Gamma_0 \cup \ldots \Gamma_k) : (\Gamma_{k+1} \cup \ldots \Gamma_\ell \cup C)$ and 
$\beta_k$ is a minimal cut for $(\Gamma_{k+1} \cup \ldots \Gamma_\ell):(\Gamma_0 \cup \ldots \Gamma_k \cup C )$ and $\alpha_k\cap\beta_k=\emptyset$. 
\end{lemma}

See \figref{fig:Gc-sol} for an example. Given the above two Lemmas we now prove Theorem~\ref{thmttt}. 
\begin{proof}[Proof of Theorem~\ref{thmttt} and Theorem~\ref{thm:RPBI}]

Use the function $k(x)$ in Lemma~\ref{lem:subG} to define bulk regions:
\begin{equation}
\gamma_k =\{ x \in V': k(x) \leq k \}
\end{equation}
for $k = 0, \ldots ,2(n-1)-1$ and
such that:
\begin{equation}
\label{wqq}
 \sum_{e  =\{x,y\} \in E'} w(e) |k(x) - k(y) | =  \sum_{k=0}^{2(n-1)-1} \sum_{ e \in \mu_{G'}( \gamma_k)} w(e)
 \end{equation}
 where $\mu_{G}(\cdot)$ denotes the edge cut function on some graph $G=\{V,E\}$, that is $\mu_{G}(r) \equiv E_{G} [r] \cap E_{G} [V\backslash r] $
 for $r \subset V$. 
 
The region $\gamma_k$ contains $A$ and shares vertices with $\alpha_k$ constructed from Lemma~\ref{lem:optk} on the subgraph $G^c$.
Similarly $V' \backslash \gamma_k$ contains $B$ and shares vertices with $\beta_k$. 
Thus we can combine them so that $A_k \equiv \gamma_k \cup \alpha_k$ forms a cut for $A: BC$ on the original graph $G$ and
 $B_k \equiv (V' \backslash \gamma_k) \cup \beta_k$ forms a cut for $B: AC$. These cuts are non-intersecting but share edges in $E'$ that are $\mu_G(\gamma_k)$. In particular:
 \begin{equation}
\mu_G(A_k) = \mu_{G'}(\gamma_k) \cup \mu_{G^c}(\alpha_k) 
\label{Aqproof}
 \end{equation}
 and similarly for $B_k$. 
 
We give an explicit proof of \Eqref{Aqproof} for completeness. Recall that $(AB)'  = V' \cap V^c$ and $E = E' \sqcup E^c$.  We note that $(AB)' \cap \alpha_k = (AB)' \cap \gamma_k$ by the boundary conditions and this implies that $A_k \cap V' = \gamma_k$ and  $A_k \cap V^c = \alpha_k$ further implying
that $E_{G'}[A_k] = E_{G'} [\gamma_k]$ since all $\{x,y\} \in E'$ satisfy $\{x,y \} \subset V'$.  Similarly $E_{G^c}[A_k] = E_{G^c} [\alpha_k]$. 
Also note that $V \backslash A_k = (V' \backslash \gamma_k) \cup (V^c \backslash \alpha_k)$. Putting this together:
\begin{align}
\begin{split}
\mu_G(A_k)  &=   E_G [A_k]  \cap   E_G [V\backslash A_k]   \\
& =  \left(  E_{G'} [A_k] \cap   E_{G'} [V \backslash A_k] \right) \cup \left(  E_{G^c} [A_k]  \cap   E_{G^c} [V \backslash A_k] \right)  \\
& =  \left(  E_{G'} [\gamma_k] \cap   E_{G'} [V' \backslash \gamma_k] \right) \cup \left(  E_{G^c} [\alpha_k]  \cap   E_{G^c} [ V' \backslash \alpha_k] \right) 
\end{split}
\end{align}
as required. A similar argument applies to $\mu_G(B_q)$.

 Thus, we can write the estimate from Lemma~\ref{lem:subG} and using \Eqref{wqq} as:
 \begin{align}
 \begin{split}
 I &\geq \left(\frac{1}{2} \sum_{k=0}^{2(n-1)-1} w(\mu_G( A_k)) +  w_G(\mu( B_k)) \right)
 + w( \mu_G( V')) \\ & \geq \min_{k} \big(  (n-1) \left(   w(\mu_G( A_k )) +  w(\mu_G( B_k))\right) + w( \mu_G( V'))  \big)
 \end{split}
 \end{align}
 The right-hand side can be written as a sum of two triway cuts, with tensions as given in the statement of this theorem. That is:
 \begin{align}
 \begin{split}
      I &\geq \min_k\left(  \frac{(n-1)}{n} \mathcal{A}( A_k : B_k : (V \backslash (A_k \cup B_k))) +\frac{1}{n} \mathcal{A}(\gamma_k : (V' \backslash \gamma_k) : (V \backslash V') ) \right) \\
 &\geq \mathcal{A}(A:B:C)
  \end{split}
 \end{align}
Equality follows from the collapsing chain mentioned in Theorem~\ref{thm:RPBI} that we have also now established. 
\end{proof}

The proofs we give for Lemma~\ref{lem:subG} and Lemma~\ref{lem:optk}  are  quite lengthy,
so we present these in Appendix~\ref{app:subG} and Appendix~\ref{app:optk} respectively.

\subsection{Uniqueness theorem}
\label{sec:uniqueness}

In this subsection we give a proof for Theorem~\ref{thm:uniqueness}, i.e., the solution to the permutation group optimization problem $R$ (Definition~\ref{def:R}) is the unique solution given by the triway cut if both the triway cut and the entanglement wedge of the graph are unique.
Though this result may seem  trivial, it lays the foundation for performing the analytic continuation $m\to 1$, the main result in \secref{sec:cont}. 

We only consider the case where every connected component of the graph $G$ is path connected to at least one boundary region. The reason is obvious -- If such a disconnected region exists then one can always set the permutations on such region to be \emph{any} fixed element on $S_{mn}$ and we have a degeneracy of $|S_{mn}|$. In any case, such disconnected parts of the graph can be factored out of the computation of reflected entropy. We can thus deal with these trivially. 

\begin{lemma}
  \label{lem:uniqueEW-id}
Let $g:V\to S_{mn}$ be an optimal solution to the permutation group program $R$. If the  minimal cut $r_{AB}$ for $AB$ is unique, then $g(v)=\id$ for all vertices $v\in r_{AB}^c$.
\end{lemma}
\begin{proof}
    We know from Theorem~\ref{thm:RPBI} that $R = \mathcal{A}(AB:C) d(\id,X) + Q'$.
    We also have, from Lemma~\ref{lem:RtoP'} that 
    \begin{equation}
    \label{eq:Rg-ineq}
        R(g) \ge \mathcal{A}(AB:C) d(\id,X) + Q'(q_X(g))
    \end{equation}
    If $g$ is optimal, this inequality must be saturated since otherwise $Q'(q_X(g))$ has a smaller objective.

    Consider an optimal solution $c:\mathcal{P}_{AB:C}\to \mathbb{R}_{\ge 0}$ to the dual max-flow problem \Eqref{program:cP}.
    The saturation condition of \Eqref{eq:Rg-ineq} demands that (see the remark before Lemma~\ref{lem:RtoP})
    \begin{equation}
    \label{eq:geodesic}
        \sum_{e=\{x,y\}\in L_C} d(g(x),g(y)) = d(g(v),\id)
    \end{equation}
    where $L_C$ is the subpath of $L\in \mathcal{P}_{AB:C}$ such that $c(L)>0$ and $L_C$ connects $C$ to the vertex $v$ immediately inside $r_{AB}$ (see \figref{fig:proof-partition}).
    In other words, it states that the elements along the path $L_C$ must be on the group geodesic of $\id$ to $g(v)$.
    Another saturation condition demands that
    \begin{equation}
        \sum_{e=\{x,y\}\subset r_{AB}^c}\left( w(e)-\sum_{L\in \mathcal{P}_{AB,C}} c(L)\mathbf{1}_L(e) \right)d(g(x),g(y))=0
    \end{equation}
    where we sum over all the edges outside the entanglement wedge $r_{AB}$.
    Since both terms in the sum are semi-positive, for any $e=\{x,y\} \subset r_{AB}^c$ it must be that either 
    \begin{equation}
    \label{eq:either_or}
        w(e)-\sum_{L\in\mathcal{P}_{AB,C}} c(L)\mathbf{1}_L(e)=0, \quad \text{or} \quad d(g(x),g(y))=0
    \end{equation}
    
    Now pick any vertex $v\in r_{AB}^c$. We claim that there must exist a path $L'\in \mathcal{P}_{C,v}$ such that the $d(g(x),g(y))=0$ for all adjacent vertices $x,y$ along the path. If this is true, the triangle inequalities along $L'$ then require $d(g(v),\id)=0$ and thus $g(v)=\id$. This holds true for any $v\in r^c_{AB}$ so we are done.
    
    Suppose by contradiction that such a path does not exist. Define the region $r_C$ to be the set of vertices that can be reached by some path with vanishing Cayley distance from $C$. It is clear that $g(x)=\id$ for all $x\in r_C$ by a similar argument as above. $r_C$ is a cut for $C$ since any vertex in $r_C$ is path connected to $C$.
    Consider now the cut surface $\mu(r_C)$. By definition, for all edges $\{x,y\}\in \mu(r_C)$ we must have $d(g(x),g(y))>0$, and thus $w(e)-\sum_{L} c(L)\mathbf{1}_L(e)=0$ by \Eqref{eq:either_or}. We compute
    \begin{equation}
        \sum_{e\in \mu(r_C)} w(e) = \sum_{L\in \mathcal{P}_{AB:C}} c(L) = \mathcal{A}(AB:C)
    \end{equation}
    since all paths starting from $C$ to $AB$ must pass through $\mu(r_C)$. Moreover, any such path with $c(L)>0$ can only intersect $\mu(r_C)$ once, since it cannot re-enter $r_C$ once it has left the cut, as this will violate the geodesic condition \Eqref{eq:geodesic}. Hence, $\mu(r_C)$ is a minimal surface and $r_C$ is a minimal cut. It is different from $r_{AB}^c$ since $v\in r_{AB}^c$ and $v\notin r_C$. This is a contradiction because we have assumed that the minimal cut is unique.
\end{proof}

\begin{lemma}
  \label{lem:dq-sat}
  Let $q:V\to P_{2n}$ be an optimal solution to the set partition optimization problem $Q$ on graph $G$ with a unique triway cut. Then for any edge $e=\{x,y\}\in E$:
  \begin{equation}
    \label{eq:dq-sat}
    d(q(x),q(y)) = \left(2(n-1)\mathbf{1}_{\mu(r_A:r_B)}+n\mathbf{1}_{\mu(r_B:r_C)}+n\mathbf{1}_{\mu(r_A:r_C)}\right)(e),
  \end{equation}
  where $(r_A,r_B,r_C)$ is the optimal triway cut for $(A:B:C)$.
\end{lemma}
\begin{proof}
    Since $q$ is optimal we know that the chain in Theorem~\ref{thm:RPBI} collapses, i.e.,
  \begin{equation}
    Q(q)=B\left(\{b^k(q)\},r\right)=I\left(\rho,\sigma\right),
  \end{equation}
  are all optimal, where $r(e)=d(q(x),q(y))$, $\rho(e) = \lceil \frac{1}{2}d(b(x),b(y)) \rceil$ and $\sigma(e)=(r-\rho)(e)$ for $e=\{x,y\}\in E$. Since we have $r(e)=(\rho+\sigma)(e)=d(q(x),q(y))$, it suffices to prove \Eqref{eq:dq-sat} on $\rho+\sigma$ for an optimal solution $(\rho,\sigma)$ of the integer program $I$.
  Now, from Lemma~\ref{lem:prop0} we know that for any optimal $(\rho,\sigma)$ one can always construct a new optimal $(\rho',\sigma')$ with $\rho+\sigma=\rho'+\sigma'$ that satisfies the conditions of Lemma~\ref{lem:subG}. Now we write
  \begin{align}
    I(\rho,\sigma) &= \frac{1}{2}\sum_{k=0}^{2(n-1)-1}\left(w(\mu(A_k))+w(\mu(B_k))\right)+w(\mu(V')) \\
        \label{eq:triway-sum}
           &=\sum_{k=0}^{2(n-1)-1}\left(\frac{n-1}{2n}\mathcal{A}(A_k:B_k:(V\backslash A_k\backslash B_k)+\frac{1}{2n}\mathcal{A}(\gamma_k:(V'\backslash\gamma):(V\backslash V')\right) \\
           &=\mathcal{A}(A:B:C) \equiv \mathcal{A}(r_A:r_B:r_C)
  \end{align}
  where we have used Lemma~\ref{lem:optk} and replaced the bound by equality since $(\rho,\sigma)$ is optimal. The cut surfaces $A_k,B_k$ and $\gamma_k$ are defined as in the proof of Theorem~\ref{thmttt}.
  Now since $\mathcal{A}(A:B:C)$ is the unique optimal triway cut, both terms in the sum of \Eqref{eq:triway-sum} must be equal to $\mathcal{A}(r_A:r_B:r_C)$ for all $k$. Hence we have
  \begin{equation}
    A_k = \gamma_k = r_A, \quad B_k = (V'\backslash \gamma_k) = r_B, \quad (V\backslash A_k\backslash B_k) = (V\backslash V') = r_C 
  \end{equation}
  which implies that within $V'$,
  \begin{equation}
    \label{eq:rhosigma1}
    \rho'(e)=0, \quad\sigma'(e) = 2(n-1)\mathbf{1}_{\mu(r_A:r_B)}(e)
  \end{equation}
  and outside $V'$,
  \begin{equation}
    \label{eq:rhosigma2}
    (\rho'+\sigma')(e) = \mathbf{1}_{\mu(V')}(e) + (n-1)(\mathbf{1}_{\mu(\alpha_k)} + \mathbf{1}_{\mu(\beta_k)})(e) = n(\mathbf{1}_{\mu(r_A:r_C)} +\mathbf{1}_{\mu(r_B:r_C)})(e),
  \end{equation}
  since any other configuration cannot reproduce an optimal value of $M$.
  Combining \Eqref{eq:rhosigma1} and \Eqref{eq:rhosigma2} together, we have proven the Lemma.
\end{proof}

\begin{corollary}
\label{lem:qsat}
    Let $q:V\to P_{2n}$ be an optimal solution to the set partition optimization problem $Q$ on graph $G$ with a unique triway cut. Then for any vertex $v\in V$:
  \begin{equation}
    \label{eq:q(v)-sat}
    q(v) =
    \begin{cases}
      q_A, \quad &v \in r_A \\
      q_B, \quad &v \in r_B \\
      \id, \quad &v \in r_C
    \end{cases}
  \end{equation}
  where $(r_A,r_B,r_C)$ is the optimal triway cut for $(A:B:C)$.
\end{corollary}
\begin{proof}
  Since $(r_A,r_B,r_C)$ is an optimal triway cut, if $v\in r_A$ then it is connected to $A$ by some path $L$ in $r_A$. The reason is as follows: Suppose that there exists a region $r\subset r_A$ not connected to $A$. Then $r$ must be connected to at least one of $r_B$ or $r_C$ otherwise $r$ will be a totally disconnected region. Without loss of generality, suppose it is $B$. We can then construct a new triway cut $(r_A\backslash r,r_B\cup r, r_C)$ whose weight is smaller than the original. This is a contradiction since we assumed $(r_A,r_B,r_C)$ is optimal.
  
  From Lemma~\ref{lem:dq-sat} we know that $\sum_{\{x,y\}\in L}d(q(x),q(y))=0$ along the path $L$. Then by repeated use of triangle inequality we can show that $d(q(v),q_A)=0$ and thus $q(v)=q_A$. Similar arguments also apply to subregion $r_B$ and $r_C$.
\end{proof}

We are now ready to complete the proof for our main result of this subsection. To begin with, we restate Theorem~\ref{thm:uniqueness} in terms of the notations used here:
\begin{theorem}
\label{thm:uniqueness2}
  Let $g:V\to S_{mn}$ be an optimal solution to the reflected entropy permutation group optimization problem $R$ on graph $G$ with a unique triway cut $(r_A,r_B,r_C)$ and a unique entanglement wedge $r_{AB}$. Then for any vertex $v\in V$:
  \begin{equation}
    \label{eq:g(v)-sat}
    g(v) =
    \begin{cases}
      g_A, \quad &v \in r_A \\
      g_B, \quad &v \in r_B \\
      X, \quad &v \in  r_{AB} \setminus r_A \setminus r_B  \equiv r_X \\
      \id, \quad & v\in r^c_{AB}
    \end{cases}
  \end{equation}
\end{theorem}
\begin{proof}[Proof of Theorem~\ref{thm:uniqueness2} and Theorem~\ref{thm:uniqueness}]
  There are four regions that we need to take care of. It is immediate that the region $r^c_{AB}$ follows trivially from Lemma~\ref{lem:uniqueEW-id}.
  For the other regions: Note that since $g$ is optimal for $R$, $q_X(g)$ must be optimal for $Q$, therefore $q_X(g)$ must coincide with \Eqref{eq:q(v)-sat}. Moreover, the saturation conditions require $d(g(x),g(y))=d(P(g(x)),P(g(y)))=d(q(x),q(y))$ for all $e=\{x,y\}$ completely within $r_{AB}$ (see the remark before Lemma~\ref{lem:RtoP}).
  
  Let's first consider $r_A$. The must be a path $L$ completely within $r_A$ such that $d(g(x),g(y))=d(q(x),q(y))=0$ along adjacent vertices $\{x,y\}\in L$ (cf. the proof of Corollary~\ref{lem:qsat}). Applying triangle inequality along the path we see $d(g(v),g_A)=0$ which implies $g(v)=g_A$. A similar argument holds for $v\in r_B$.

  For the remaining region $r_X$, We split the contribution of $R(g)$ as
  \begin{align}
    R(g) =  \sum_{e\in E, e\subset r_{AB}} w(e) d(g(x),g(y)) + \sum_{v \in \partial r_{AB}} w(e) d(g(v),\id)
  \end{align}
  where $\partial r_{AB}\equiv V[\mu(r_{AB})]\cap r_{AB}$ are the vertices on the cut surface $\mu(r_{AB})$ and within $r_{AB}$. We have used the fact that $g(v)=\id$ for $v\in r_{AB}^c$.
  Since $q_X(g(v))=\id$ for $v\in r_{AB}$ we have $\#(g\vee X)=\#(g)$, and $P(g(v))$ lies on the geodesic between set partitions $P(X)$ and $P(\id)$ on $P_{mn}$:
  \begin{equation}
    d(P(X),P(g(v)))+d(P(g(v)),P(\id)) = \#(X)-2\#(X\vee g(v))+\#(\id) =  d(P(X),P(\id))
  \end{equation}
  We claim that this geodesic condition naturally extends to Cayley distances on $S_{mn}$.
  The reasoning is as follows: The region $r_X$ must be path connected to either $A$ or $B$, since if it is only connected to $C$ then we know such region is subset of $r_{AB}^c$ and thus must be empty. Now consider a connected component $r_i$ of this region. Without loss of generality we may set $g=g_i$ for all $v\in r_i$. Consider an edge $\{x,y\}$ on the common cut surface of $r_A$ (or $r_B$ respectively) and $r_i$ such that $g(x)=g_i$ and $g(y)=g_{A/B}$. We know that since $g$ is optimal and $x,y\in r_{AB}$,
  \begin{equation}
    d(g_i,g_{A/B}) = d(P(g_i),P(g_{A/B}))+ 2G_{g_{A/B}}(g_i) = d(q_X(g_i),q_{A/B}),
  \end{equation}
  which implies that the genus $G_{g_{A/B}}(g_i)=0$. But since $P(g_{A/B})\ge P(X)$ we have $G_{g_{A/B}}(g_i)\ge G_{X}(g_i)$ \footnote{This fact follows from the construction of admissible surfaces: Given an admissible surface of $g_i$ based on $g_A$ (or $g_B$) one can ``pinch'' the cycles of $g_A$ (resp. $g_B$) to form cycles of $X$ without destroying any connections. See Appendix A. of Ref.~\cite{Akers:2021pvd} for details.} and thus $G_{X}(g_i)=0$. Also $G_{g_i}(\id)=0$ trivially. Thus $d(P(g_i),P(\id))=d(g_i,\id)$ and $d(P(g_i),P(X))=d(g_i,x)$ and we have
  \begin{align}
  \begin{split}
    R(g) &=  \sum_{e\in E, e\subset r_{AB}} w(e) d(g(x),g(y)) + \sum_{v \in  \mu(r_{AB}) \cap r_{AB}} w(e) (d(X,\id)-d(X,g(v))) \\
    &=  \sum_{e\in E, e\subset r_{AB}} w(e) d(g(x),g(y)) + \mathcal{A}(AB:C)d(X,\id) - \sum_{v \in \partial r_{AB}} w(e) d(X,g(v))
    \end{split}
  \end{align}
  Now since $g$ is optimal we know $R(g)=Q(q_X(g))+\mathcal{A}(AB:C)d(X,\id)$. Using $d(g(x),g(y))=d(q(x),q(y))$ within $r_{AB}$ we obtain
  \begin{equation}
    \sum_{v \in \partial r_{AB}} w(e) d(X,g(v)) = 0 
  \end{equation}
  and thus $g(v)=X$ for $v\in r_{AB}$ on the cut surface $\mu(r_{AB})$. Since Cayley distance vanishes within $r_X$ we must have $g(v)=X$ within the whole region. This completes our proof.
\end{proof}

\begin{remark}
  In general, when using the permutation optimization to determine the moments of the reflected density matrix, one must sum over all the configurations that minimizes the free energy functional, i.e. \Eqref{statmech}.
    Physically speaking, what we have proven here is that such minimum is unique as long as one is sufficiently far away from the phase transitions, which is signalled by a degenerate minimal surface. In fact, there are two different kinds of phase transitions in play here -- one being the entanglement wedge phase transition, in which the reflected entropy suffers a discontinuous jump; and the other being the EW cross-section phase transition, in which the there are two cross-sectional candidates with the same area. In the latter case, the change of reflected entropy across the transition is still continuous. Our result here only applies when the system is far away from both kind of transitions. 
    
    However, based on various evidences, we conjecture that $S^{(n)}_R$ still converges as stated in Theorem~\ref{thm:1} for the second kind of the transitions. In this case the multiplicity factor is no longer unity but some integer $d^{(n)}_m>1$. We conjecture that $d^{(n)}_m$ is \emph{independent} of $m$ in this scenario. Therefore, the techniques we will be using for analytic continuation still carries over. The multiplicity factor will introduce a $\ln d^{(n)} \sim O(\chi^0)$ correction to the entropy, which goes away as one takes $\chi\to\infty$.
    Our uniqueness assumption on the triway cut problem would then be unnecessary in Theorem~\ref{thm:1}. 
\end{remark}

\section{Continuation in $m$}

\label{sec:cont}

To finish the proof of Theorem~\ref{thm:1}, we need to ``analytically continue'' away from integer $m/2$.
Since the answer we have for $R(g)$ only depends trivially on $m$ through normalization (it is independent upon correctly normalizing the reflected density matrix),
it seems like this task would be simple. For example, one might expect a simple application of Carlson's theorem. Unfortunately applying Carlson's theorem
after taking the limit $\chi \rightarrow \infty$ turns out to be rather difficult -- a fact that is often not discussed in the AdS/CFT literature.  The basic issue is that
one does not know if the natural analytic continuation of ${\rm Tr} ( \sigma_{AA^\star}^{(m)})^n$ in $m$ remains an analytic function upon 
taking the limit $\chi \rightarrow \infty$.  Compounding this difficulty is the need to divide the function one wishes to continue by the expected answer - this division is necessary for convergence, but will often introduce non-analytic dependence on $m$. Indeed it is well known that partition functions do not remain analytic in the thermodynamic limit due to phase transitions and the condensation of zeros. 
That said, such phase transitions in $m$ are not present here since the expected answer has a simple analytic dependence on $m$.
Even so, it is not obvious how to establish pointwise convergence of the sequence as $\chi \rightarrow \infty$ and this is necessary in order to establish analyticity of the limit. 

Thus we follow a different approach here -- that of the method of moments. One immediate difficulty is that the quantity of interest ${\rm Tr} ( \sigma_{AA^\star}^{(m)})^n$  is not obviously a moment for $m$ (it is a moment for $n$, but here we are fixing $n$ to be an integer.) However we can write it as a moment for some operator spectral measure \cite{Dutta:2019gen,Akers:2023obn}:
\begin{equation}
{\rm Tr}_{AA^\star} ( \sigma_{AA^\star}^{(m)})^n = \left< 1_{AB}^{\otimes n} \right| \Sigma_A^\dagger \varrho^{m/2} \otimes  \varrho^{m/2} \Sigma_A \left|  1_{AB}^{\otimes n} \right>
\end{equation}
where $\varrho = \rho_{AB}^{\otimes n}$ and $\Sigma_A$ is the usual $n$-twist operator. 
Consider the operator:
\begin{equation}
 \mathcal{O} =  \varrho \otimes \varrho 
\end{equation}
and the associated spectral measure $E_\lambda$ from which we define a new measure $\mu_\Psi(\lambda) = \left< \Psi \right| E_\lambda \left| \Psi \right> $ with $ \left| \Psi \right> = \Sigma_A \left|  1_{AB}^{\otimes n} \right> $. More specifically we can write:
\begin{equation}
\mathcal{O} = \sum_i \lambda_i \left| v_i \right> \left< v_i \right|
\end{equation}
and the measure for finite $\chi$ will have a discrete decomposition:
\begin{equation}
 d\mu_\Psi(\lambda)  = \sum_{i} |\left< \Psi \right| \left. v_i \right>|^2 \delta( \lambda - \lambda_i) d\lambda
\end{equation}
We recover the quantity of interest:
\begin{equation}
{\rm Tr}_{AA^\star} ( \sigma_{AA^\star}^{(m)})^n   = \int_0^\infty d \mu_\Psi(\lambda) \lambda^{m/2}
\end{equation}
We have computed the moments for $m/2$ integer away from the phase transition, schematically:
\begin{equation}
\label{eq:schematic-convergence}
\overline{ \int_0^\infty d \mu_\Psi(\lambda) \lambda^{m/2} } ~\underset{\chi\to\infty}{\longrightarrow}~ \chi^{ - 2(n-1) \mathcal{A}(A:B:C) + n(m-2) \mathcal{A}(AB:C)}
\end{equation}
We have not computed the $m=0$ moment. This makes the moment problem somewhat more involved since we have less
control over the limit of the measure $\mu_\Psi$ near $\lambda =0$. 
Instead, we will work with a weak form of measure concentration for general random tensor networks that works as long as the entanglement wedge is appropriately unique:

\begin{definition}[Unique entanglement wedge]
\label{def:unique-EW}
Given a random tensor network with boundary $\partial = AB \sqcup C$ then the entanglement wedge $r_{AB}$ is called \emph{unique} if for any other $AB:C$ cut $r \neq r_{AB}$:
\begin{equation}
\mathcal{A}(AB:C) \leq \mathcal{A}(r : r^c) - g
\end{equation}
for some non-zero gap $g > 0$. 
\end{definition}

\begin{lemma}[Weak measure concentration]
Consider a random tensor network state with boundary $\partial = AB\sqcup C$ and with a unique entanglement wedge
for $AB$, then:
\begin{equation}
\label{eq:wmc2}
\overline{ \| \rho_{AB} - \pi_{\rho_{AB}} /\chi^{\mathcal{A}(AB:C)} \|_1} = \mathcal{O}(\chi^{-g/2})
\end{equation}
where $\pi_{\rho_{AB}}$ is the projection operator onto the support of $\rho_{AB}$.
\end{lemma}
\begin{remark}
The result presented here is a much weaker bound than the standard measure concentration where the probability of finding a minimum non-zero eigenvalue away from the peak is exponentially small in $\chi^{\#}$. 
Such measure concentration results were proven for a single random tensor \footnote{See e.g. \cite{10.1093/acprof:oso/9780199535255.001.0001}, Theorem 5.6.} -- multiple random tensors seem much harder to work with, hence we have not succeeded in deriving the much stronger result.
This weaker result is however sufficient for our purposes. 
\end{remark}
\begin{proof}
We start with a computation of the projected trace distance on the subspace $\pi_{\rho_{AB}} \mathcal{H}_{AB}$
where ${\rm Tr} \, \pi_{\rho_{AB}} \leq \chi^{\mathcal{A}(AB)} $ for any tensor network state.
Then from H\"older's inequality we have:
\begin{align}
\begin{split}
\left\| \rho_{AB} - \pi_{\rho_{AB}}/\chi^{\mathcal{A}(AB)} \right\|_{1}^2 & \leq \chi^{\mathcal{A}(AB)} \| \rho_{AB} - \pi_{\rho_{AB}}/\chi^{\mathcal{A}(AB)}  \|_2^2
= \chi^{\mathcal{A}(AB)} {\rm Tr} (\rho_{AB} - \pi_{\rho_{AB}}/\chi^{\mathcal{A}(AB)})^2 \\
&=  \chi^{\mathcal{A}(AB)} \left( {\rm Tr} \rho_{AB}^2 - 2 /\chi^{\mathcal{A}(AB)} + {\rm Tr} \pi_{\rho_{AB}}/\chi^{2\mathcal{A}(AB)} \right) \\
&\leq \chi^{\mathcal{A}(AB)}  {\rm Tr} \rho_{AB}^2  - 1
\end{split}
\end{align}
Averaging gives the required estimate:
\begin{equation}
\chi^{\mathcal{A}(AB)}  \overline{{\rm Tr} \rho_{AB}^2}  - 1
= \mathcal{O}(\chi^{-g})
\end{equation}
by explicit computation of the replica statistical mechanics model, and application of the gap condition in Definition~\ref{def:unique-EW}.

\end{proof}

To proceed, we define a specific form of convergence in probabilities as follows. We say that:
\begin{equation}
 \alpha_{\chi} \mathop{\rightarrow}^{\rm Pr} c
 \end{equation}
if for all $\epsilon>0$:
\begin{equation}
\lim_{\chi \rightarrow \infty} {\rm Pr}( |\alpha_{\chi}-c | \geq \epsilon)  = 0
\end{equation}
 Where $\alpha_{\chi}$ is a sequence of real valued random variables (on different probability spaces equipped with the Haar measure with dimension determined by $\chi$)
 and $c$ is simply a constant.  
We now state our main result of this section:
\begin{lemma}
\label{conv:prob}
For integer $n$:
\begin{equation}
\label{cont:prob}
S_{R}^{(n)}(A:B) - \ln \chi \left(\frac{1}{n-1} \mathcal{A}_{\mathbf{t}}(A:B:C) - \frac{n}{n-1} \mathcal{A}(AB:C) \right) \,\, \mathop{\rightarrow}^{\rm Pr} \,\, 0
\end{equation}
as $\chi \rightarrow \infty$. 
\end{lemma}
The proof we give for Lemma~\ref{conv:prob} is rather lengthy and technical so we present it in Appendix~\ref{app:conv}.
We now use the above result to prove Theorem~\ref{thm:1}. 

\begin{proof}[Proof of Theorem~\ref{thm:1}]

The only thing remaining is to move from convergence in probability to convergence in mean after dividing by $\ln \chi$. 
Certainly:
\begin{equation}
\frac{S_{R}^{(n)}(A:B)}{\ln \chi} -  \left(\frac{1}{n-1} \mathcal{A}_{\mathbf{t}}(A:B:C) - \frac{n}{n-1} \mathcal{A}(AB:C) \right) \,\, \mathop{\rightarrow}^{\rm Pr} \,\, 0
\end{equation}
as $\chi \rightarrow \infty$, follows from Lemma~\ref{conv:prob}. We now show that $\frac{S_{R}^{(n)}(A:B)}{\ln \chi} $ is a uniformly bounded random variable.
We use monotonicity, established in Ref.~\cite{Dutta:2019gen} for integer $n \geq 2$:
\begin{equation}
S_{R}^{(n)}(A:B) = S_{n}(AA^\star)_{\rho_{AB}^{1/2}} \leq S_{n}(AA^\star)_{\rho_{ABC}^{1/2}}  =  S_{R}^{(n)}(A:BC) = 2 S_n(A)
\end{equation}
Using the Swingle bound gives:
\begin{equation}
\frac{S_{R}^{(n)}(A:B) }{\ln \chi} \leq 2 \mathcal{A}(A:BC)
\end{equation}
Both of these bounds
apply to all instances of the ensemble. Convergence in probability for a bounded random variable, implies convergence in mean, and so we are done.
\end{proof}

\section{Discussion}
\label{sec:disc}

In conclusion, we have rigorously demonstrated that the $(m,n)$-R'enyi reflected entropies in random tensor networks are computed by saddles involving triway cuts as shown in \figref{fig:network-multi-cut} for arbitrary $m$ and integer $n$. Moreover, there is a natural analytic continuation of the triway cut problem for non-integer $n$ which leads to the holographic proposal relating to the entanglement wedge cross section. We now comment on various aspects of our work.

\subsection*{Bit Threads}

As discussed earlier, bit threads provide a vivid picture of the entanglement structure of holographic states. Based on the structure of bit thread configurations, Ref.~\cite{Cui:2018dyq} conjectured that the three party entanglement structure of holographic systems is dominated by bipartite entanglement. However, a version of this conjecture is in conflict with the $S_R=2EW$ proposal \cite{Akers:2019gcv}, for which we have found further evidence in this paper. We would nevertheless like to discuss a way to see how one can get as close to bit threads as possible.

As reviewed in \secref{sub:minmax}, the RT formula can be recast as a max flow problem by convex duality. Moreover, the natural form of the min-cut problem from the RTN perspective is an integer program where the optimization is over domain walls with integer energy costs, representing the different permutations that contribute to the free energy minimization problem. For the entanglement entropy, the problem can be relaxed to a linear program over the real numbers since there isn't an integrality gap for this problem. Once relaxed, convex duality naturally leads to a max-flow problem, i.e., bit threads.

In the context of reflected entropy however, we showed that the triway cut problem is equivalent to an integer program with a non-trivial integrality gap. The existence of such an integrality gap for the triway cut problem was in fact previously discussed in Ref.~\cite{Harper:2019lff}. While the R\'enyi reflected entropies are computed by triway cuts, which are related to the entanglement wedge cross section, the non-integer program allows relaxation to surfaces that are related to the original RT surfaces as shown in \figref{fig:relax} (e.g. for $n=2$). For RTNs, this is identical to the mutual information, as found from assuming the mostly bipartite conjecture. Thus, we can think of the relaxation of the integer program as the ``incorrect" step that led to bit threads, as well as the mostly bipartite conjecture.

It has been conjectured that one can amend the mostly bipartite story from bit threads by considering a generalized ``hyperthread'' optimization program \cite{Harper:2021uuq} (see also \cite{Harper:2019lff}). Hyperthreads are threads connecting between multiple $(k\ge 3)$ boundary regions and it is conjectured that the optimal value of a $k$-thread program may be a measure of $k$-partite entanglement. In the case of 3-threads, the optimal configuration saturates a minimal triway cut. In this sense, our results serve as a firsthand bridge that connects the 3-thread problem to a concrete quantum information measure. We think it should be possible to derive the hyperthread optimization as a dual problem of the various integer programs appearing in this paper -- indeed the triway cut does not admit a bit thread dual but it may still be able to be ``dual'' to a more exotic program such as the hyperthreads. On the other hand, reflected entropy can be naturally generalized to accommodate $k$-party systems \cite{Bao:2019zqc}. It would be interesting to see if the formalism we developed in this paper extends to such case and if there is possible connection to the $k$-thread programs. We leave these investigations to future works.

\subsection*{More General Tensor Networks}

There are various aspects of RTNs that make them good models of holography such as their relation to fixed-area states \cite{Akers:2018fow,Dong:2018seb,Dong:2019piw}. However, there are other aspects that are missing such as the lack of mutual backreaction between domain walls, as well as the commutativity of area operators \cite{Bao:2018pvs}. Thus, it is interesting to analyze the extent to which variants of random tensor networks can model holography. 

For instance, the RT formula can be reproduced by choosing tensors that are 2-designs  where the average is the same as the Haar average up to the second moment \cite{Hayden:2016cfa}. Random stabilizer tensor networks are an example which form at most a projective $3$-design \cite{klappenecker2005mutually,gross2007evenly,webb2015clifford,zhu2017multiqubit}. However,  they satisfy the bipartite dominance conjecture \cite{Nezami:2016zni} and thus, do not accurately model the reflected entropy for holographic states. Our paper certainly makes use of larger moments, so the discrepancy is no surprise. In fact, our results suggest that the integrality gap would become visible to a projective $4$-design, at least when computing the $(2,2)$ R\'enyi reflected entropy (involving canonical purifications of the density matrix $\rho_{AB}^2/{\rm Tr} \rho_{AB}^2$). The implications for the R\'enyi reflected entropy and the Markov gap are less clear, and we leave it as an important open question to understand at what level of $k$-design the Markov gaps becomes large, of order $\ln \chi$. 

Another possible generalization of the RTN that can be considered is to use non-maximally entangled edge states \cite{Cheng:2022ori}. In this case, the calculation for the $(m,n)$-R\'enyi reflected entropy is similar except the energy costs on the domain wall become functions of the entanglement spectrum of the edge states. In this case, it is harder to prove anything about the analytic continuation, but we expect similar saddles to play an important role.
Even more ambitious would be to use tensor networks with edge degrees of freedom, as in \cite{Akers:2024wab, Akers:2024ixq, Dong:2023kyr, Qi:2022lbd, Donnelly:2016qqt, Colafranceschi:2022dig}, with a possible payoff of more realistic backreaction.

Another generalization we can consider is that of hypergraph RTNs. Hypergraphs are a natural generalization of graphs where edges connecting 2 vertices are generalized to hyperedges potentially connecting more than 2 vertices. States that satisfy the RT formula on hypergraphs were discussed in Refs.~\cite{Bao:2020zgx}. Such states have a natural construction in terms of the RTN where hyperedges are formed by projecting onto GHZ states coupling multiple vertex tensors \cite{Walter:2020zvt}. One could then ask what the reflected entropy for such states is. Although we do not have a proof analogous to the one for usual RTNs, assuming that triway cut-like configurations dominate, we can compute the R\'enyi reflected entropy. The important ingredient is the generalization of the free energy optimization problem. The free energy cost of an edge in the RTN for vertex permutations $g_1$ and $g_2$ is weighted by $\#(g_2 g_1^{-1})$. It is easy to work out the generalization for hyperedges. For instance, for a 3-edge with vertex permutations $g_1$, $g_2$ and $g_3$, the free energy cost is $\#(g_2g_1^{-1}\vee g_3g_1^{-1})$, which measures the number of orbits of the relevant elements. 

\subsection*{Holographic States}

Having obtained rigorous results for RTNs, it is natural to guess that similar saddles, which are different from the naive saddles written down in Ref.~\cite{Dutta:2019gen}, will also play an important role in holography, resolving the issues with analytic continuation found in Ref.~\cite{Kusuki:2019evw}.
In particular, this would imply that geometric saddles involving backreacted versions of triway cuts would compute the $(m,n)$-R\'enyi reflected entropy in holographic states.

It is useful to contrast this with the case of the multi-entropy, a symmetric multipartite entanglement measure defined by Ref.~\cite{Gadde:2022cqi}.
In particular, for the case of three parties, the bulk dual was proposed to be a triway cut.
However, it was argued by Ref.~\cite{Penington:2022dhr} that despite the multi-entropy being computed by a triway cut in RTNs, the corresponding configurations were not gravitational saddles in AdS/CFT.
The basic reason is that in a smooth gravitational saddle, the local neighborhood of any point is a Euclidean ball.
This turns out to be true for the triway cut configuration at $n=2$.
However, for the triway cut configurations at $n>2$, the topology is such that the junction of the triway cut is not smooth.
A similar argument also holds for the entanglement negativity and the odd entropy \cite{Dong:2024gud}.

Given this, one may worry that while triway cuts are relevant for RTNs, they may not compute the $(m,n)$-R\'enyi reflected entropy in holographic states.
However, one can check that for the reflected entropy of $(m,n)$-R'enyi, the junction of the triway cut is consistent with the topology of a Euclidean ball, thus avoiding the above failure mode.
Thus, we still expect the topologies of the RTN saddles leading to triway cuts to continue to be relevant for holographic states.
A particular case where this is known to be true is that of the $(2,n)$-R\'enyi reflected entropy computed in Ref.~\cite{Milekhin:2022zsy}.

\subsection*{Effective Description}

In Ref.~\cite{Akers:2021pvd}, we discussed an effective description of the canonical purification as a superposition over states with different values of area for the squeezed cross sections. E.g., for hyperbolic RTNs we have
\begin{equation}\label{eq:CP_effective_hyperbolic} 
		\includegraphics[width=0.8
		\textwidth]{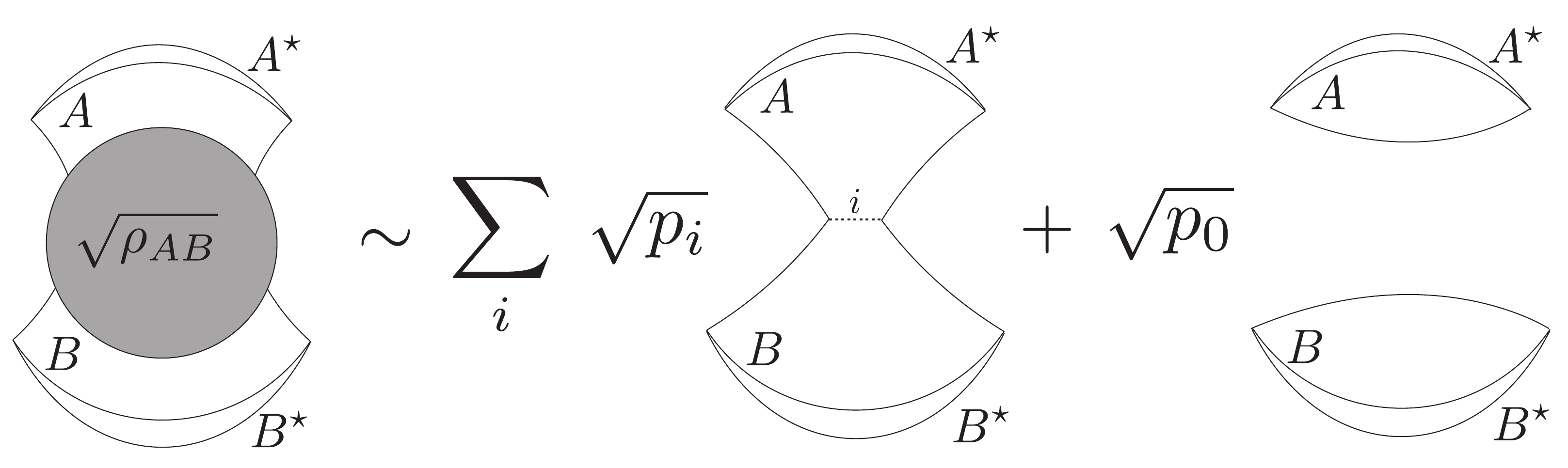}. 
	\end{equation} 
It is natural to expect a generalization of this result to general RTNs where the squeezed cross sections are solutions to the triway cut problem at different values of $n$. 

\begin{figure}
    \centering
    \includegraphics[scale=0.1]{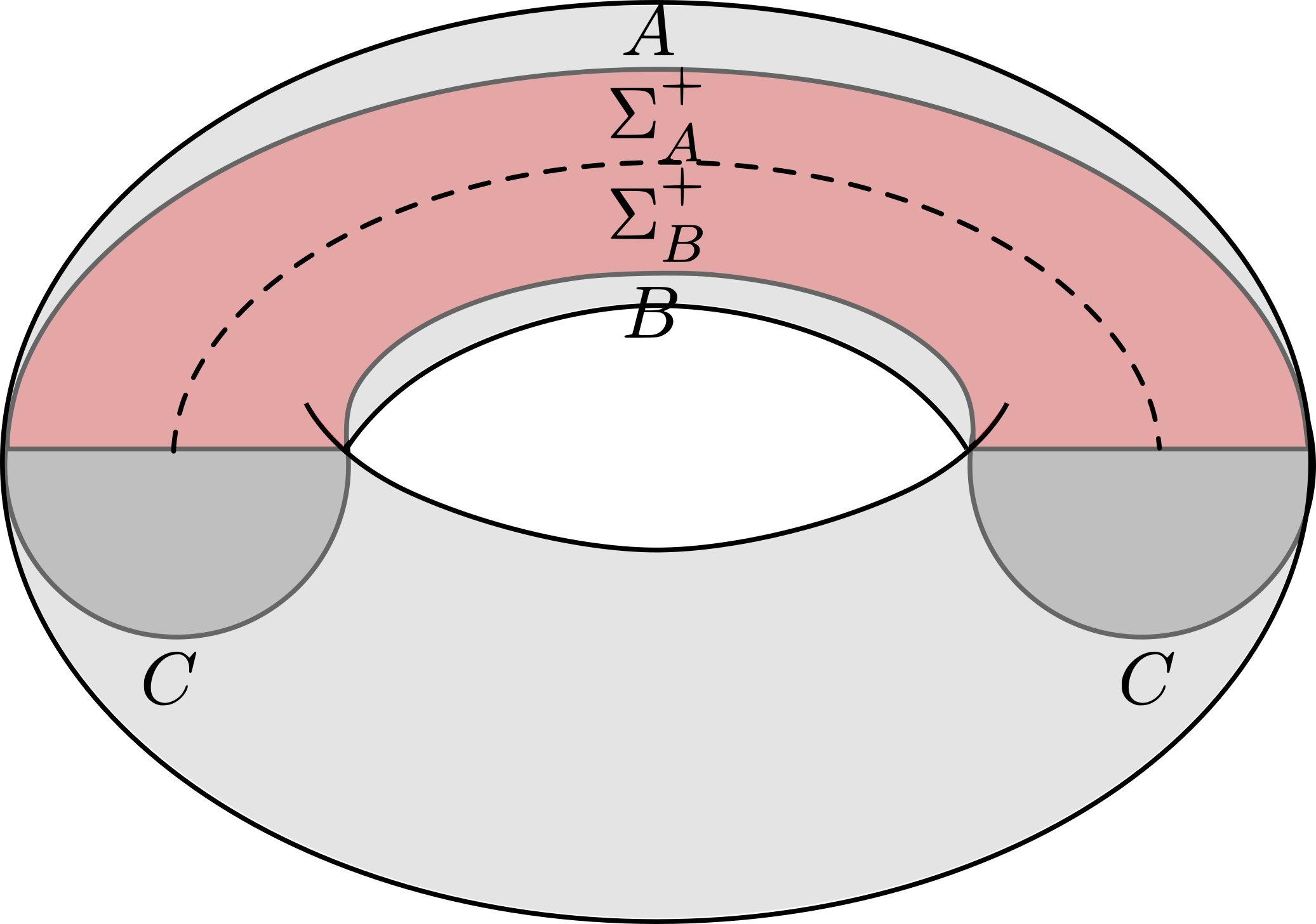}
    \caption{The purification for two intervals in the vacuum state at $m=2$ is a BTZ geometry where fixing the area of the entanglement wedge cross section induces a conical defect (dotted line) at the BTZ horizon.}
    \label{fig:torus_slice}
\end{figure}

While this effective description was proposed in the context of RTNs, where there is no backreaction and the superposition is simply over different sets of vertices in the graph, it is natural to ask if it can be understood in holographic systems as well. For the special case of two intervals in the vacuum state at $m=2$, it was found in Refs.~\cite{Yin:2022toc,Milekhin:2022zsy} that the R\'enyi reflected entropies for arbitrary $n$ can be computed by the torus partition function. This makes analytic continuation easy and is related to the fact that for $m=2$, there is no independent $X$ element. The calculation of R\'enyi reflected entropy can then similarly be decomposed into fixed-area sectors of the entanglement wedge cross section, which correspond to different horizon areas in the BTZ saddle. It is then easy to see that the effective description with different squeezed cross sections is represented by these fixed-area states which induce varying conical defects at the horizon, which translate into the angle subtended at the triway cut at different values of $n$ (see \figref{fig:torus_slice}). 

\acknowledgments

TF would like to thank Alejandro Dominguez-Garcia for discussions. SL would like to thank Shiliang Gao, Jingwei Xu, and Kiran Kumar A.S. for useful discussions. PR would like to thank Abhijit Gadde for discussions. 
PR is supported in part by a grant from the Simons Foundation, by funds from UCSB, the Berkeley Center for Theoretical Physics; by the Department of Energy, Office of Science, Office of High Energy Physics under QuantISED Award DE-SC0019380, under contract DE-AC02-05CH11231 and by the National Science Foundation under Award Number 2112880. 
This material is based upon work supported by the Air Force Office of Scientific Research under award number FA9550-19-1-0360.

\appendix

\section{Specific group and permutation elements}
\label{app:specific}

In this appendix we give a description of the important permutation and partition elements that appear in our proof of the main theorem. 

We first describe the permutation group elements $g_A$ and $g_B\in S_{mn}$.
Denote a specific replica by $(k,\ell)$ where $k=1\ldots n$ and $\ell=1\ldots m$; $\tau_m^{[k]}$ as the $m$-cyclic permutation that shifts  $(k,\ell)\to(k,\ell+1)$ and likewise $\tau_n^{[\ell]}$ the $n$-cyclic permutation that shifts  $(k,\ell)\to(k+1,\ell)$.
We write 
\begin{equation}
g_B = \prod_{k=1}^{n} \tau_m^{[k]}, \qquad g_A = \gamma^{-1} g_B \, \gamma
\end{equation}
where $\gamma=\prod_{\ell>m/2}\tau^{[\ell]}_n$ shifts the lower ($\ell>m/2$) elements cyclically in the $n$-direction.
A graphical representation of these elements are shown in \figref{fig:gAgB}.
\begin{figure}[ht]
    \centering
    \includegraphics[scale=.27]{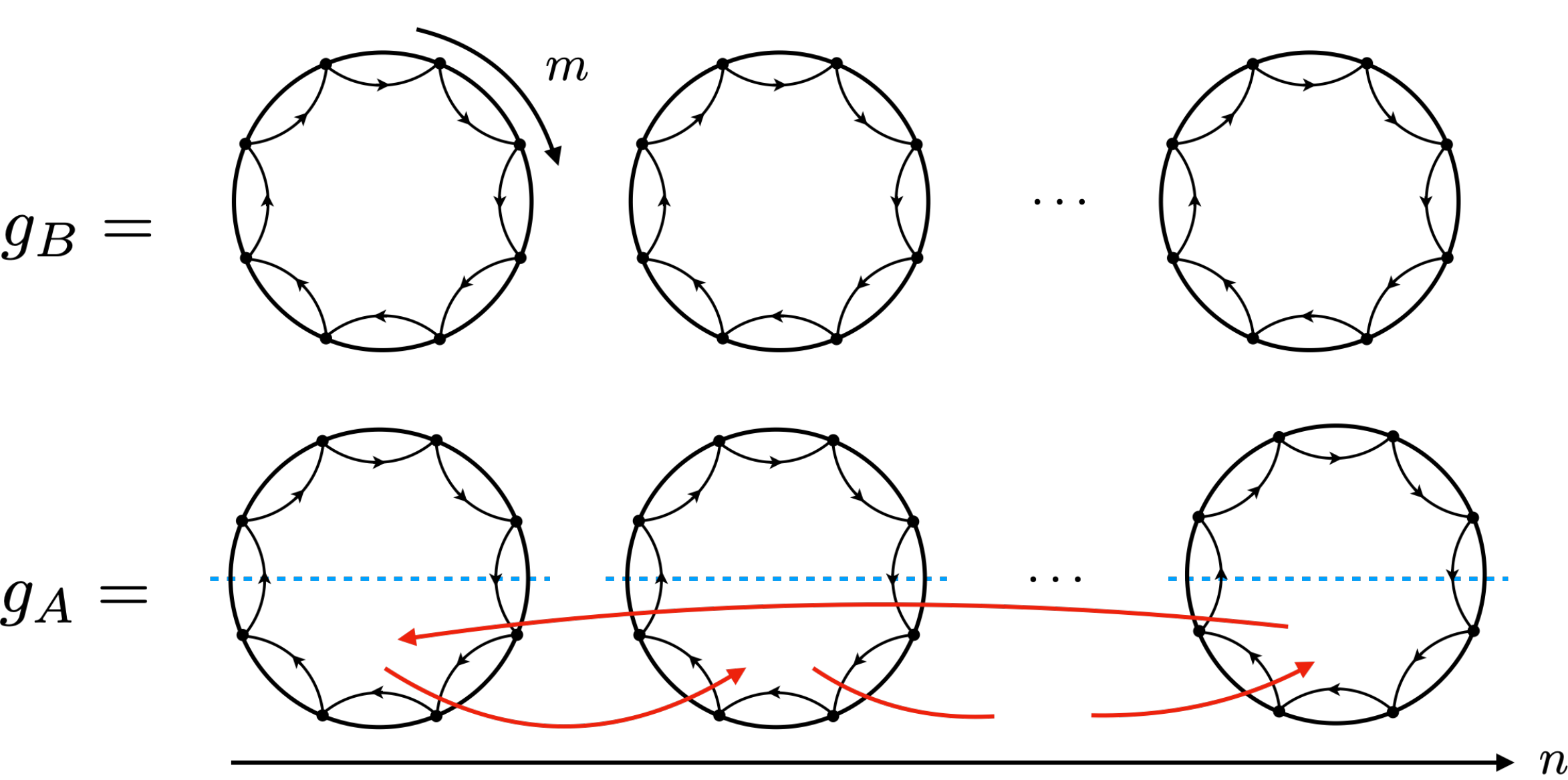}
    \caption{A graphical representation of $g_A$ and $g_B$. Each dot denotes a particular replica. Each circle represents the $m$ replicas and each circle is further replicated $n$ times. Going clockwise in each circle increases the $m$ replica number and going to the next circle to the right increases the $n$ replica number. The permutation is represented as directed arrows connecting different replicas. The element $g_A$ can be thought of as cutting open the $m$-circles in $g_B$ open, permuting the lower semicircle cyclically in the $n$-direction and then gluing them back together.}
    \label{fig:gAgB}
\end{figure}

There is a unique element $X\in S_{mn}$ that lies on the joint geodesic of $\id\to g_A$ and $\id \to g_B$ while being farthest away from $\id$. Its form is given by
\begin{equation}
X =  \prod_{k=1}^{n} (\tau_{m/2}^U \tau_{m/2}^L)_k
\end{equation}
where $\tau^U_{m/2}$ and $\tau^U_{m/2}$ defines a permutation in $S_{m}$ by cyclically permuting the upper $m/2$ the lower $m/2$ elements. See \figref{fig:X}.
\begin{figure}
    \centering
    \includegraphics[scale=.27]{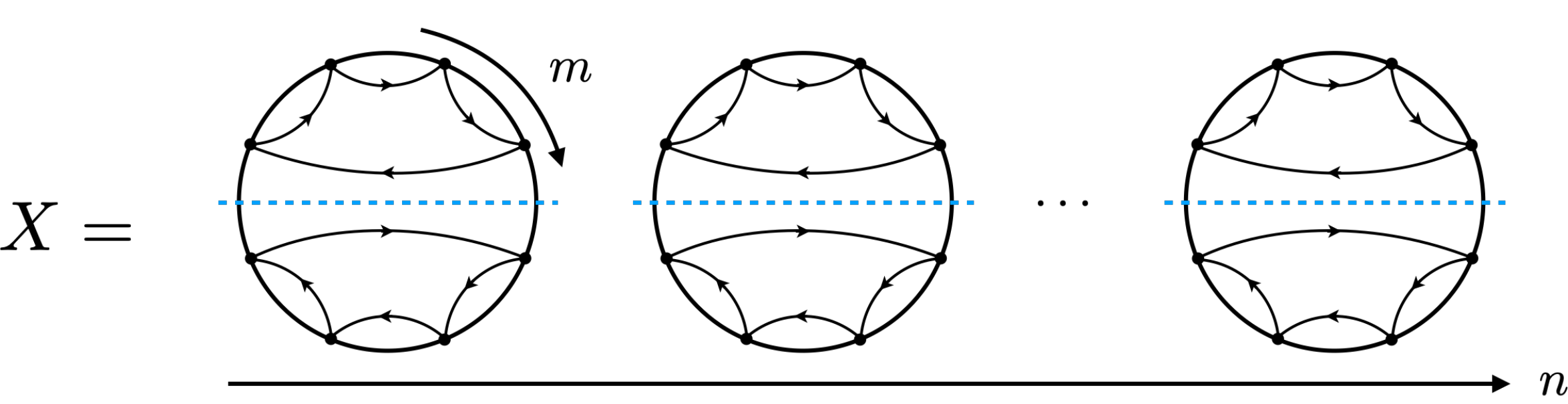}
    \caption{The element $X$ can be expressed as the product of maximum cyclic permutations $\tau^U_{m/2},\tau^L_{m/2}$ that act within each $m$-semicircles.}
    \label{fig:X}
\end{figure}
The Cayley distances between these elements and the identity are
\begin{equation}
    d(g_A,g_B) = 2(n-1), \quad d(\id,g_{A/B})=n(m-1), \quad d(\id,X)=n(m-2)
\end{equation}

For the sake of completeness we also include the coarse grained version of the permutation group elements $q_A\equiv q_X(g_A)$ and $q_B \equiv q_X(g_B)$ here:
\begin{equation}
    \begin{matrix}
        \includegraphics[scale=.35]{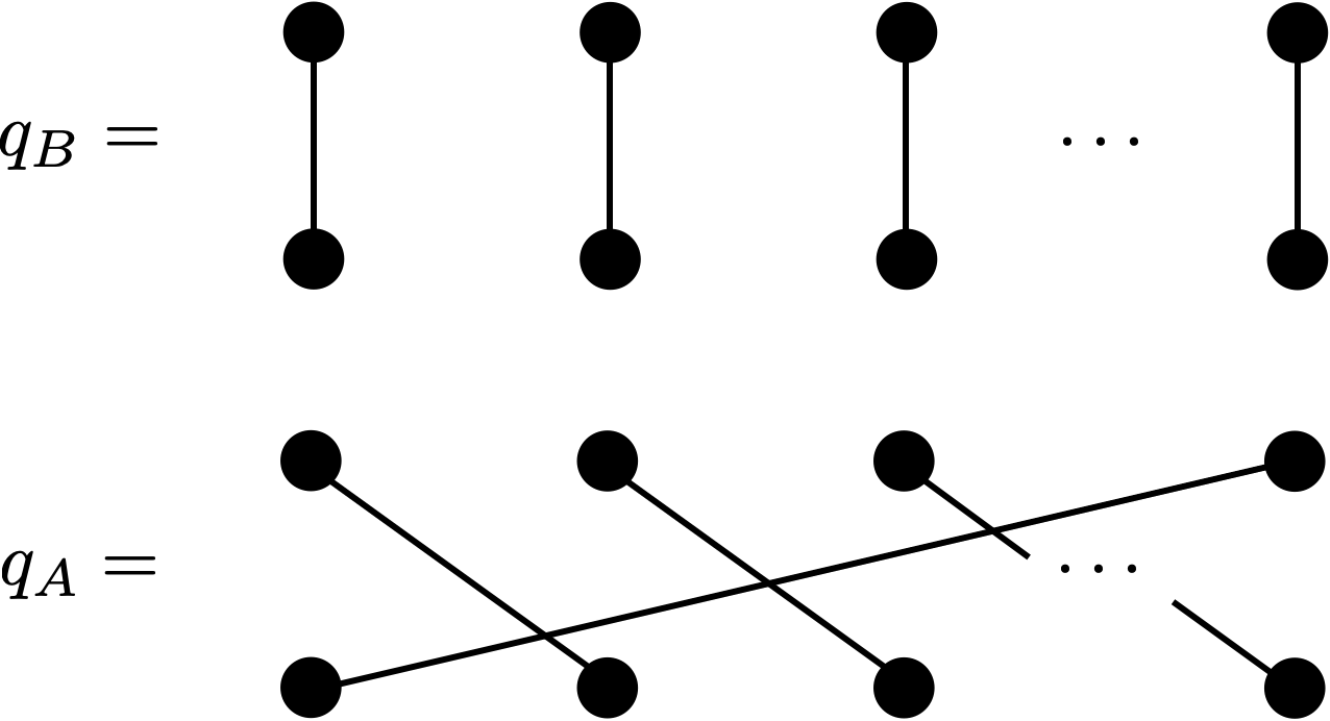}
    \end{matrix}
\end{equation}
where we have represented each dot as an element in $\mathbb{Z}_{2n}$, each corresponding to a cycle in the blocking permutation $X$, with position arranged in a similar fashion as in \figref{fig:X}. We connect two dots by a line if they belongs to the same subset in the set partition. 
Note that $q_X(X)=q_X(\id)=\{\{1\},\ldots,\{2n\}\}$, the finest element in $P_{2n}$ which is also denoted by $\id\in P_{2n}$.
The distances between these set partitions are
\begin{equation}
    d(q_A,q_B) = 2(n-1), \quad d(\id,q_A)=d(\id,q_B)=n
\end{equation}

\section{Proofs}

\subsection{Lemma~\ref{lem:subG}: Structure of optimal $\rho,\sigma$}
\label{app:subG}
\label{sec:struct}
For quick reference we sketch the integer program, Definition~\ref{def:I}, we are trying to solve:
\begin{equation}
\nonumber
\includegraphics[scale=1.2]{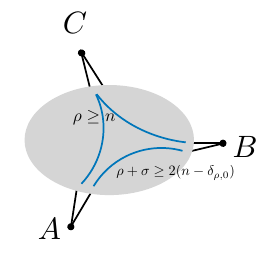}
\vspace{-.4cm}
\end{equation}
plus the additional even condition on paths $A \rightarrow B$ (see the original definition.)

Let us introduce some notation. For a path $L$, not necessarily edge disjoint, define:
\begin{equation}
\label{notation}
f(L) = \sum_{e \in L} f(e)
\end{equation}
for some function of the edges $f: E\to \mathbb{R}$. This is slightly different (but consistent with) the notation introduced previously, because of the possibility that
paths $L$ intersect an edge more than once. We sum each edge in the sequence defining $L$, which then includes such multiplicities. 

\begin{definition}
\label{rhobinding}
We say $e$ is $\rho$-\emph{binding} if either $\rho(e) =0$ or, when $\rho(e) >0$, there exists a path $L \in \mathcal{P}_{AB,C}$ with $e \in L$ such that
$\rho(L) = n$. 
\end{definition} 

For the constraint involving the $\sigma$ variable \Eqref{eveness} we need a refined notion:
\begin{definition}
An edge $e$ is called $\sigma$-\emph{binding within $E'$} for some $E' \subset E$ if either $\sigma(e) =0$ or, when $\sigma(e) > 0$, then there exists a path $L \in \mathcal{P}_{A,B}$ with $e \in L$ and $L \subset E'$ for which $\sigma(L) + \rho(L) = 2 (n - \delta_{\rho(L),0})$.  We say that $e$ is \emph{simply $\sigma$-binding} for the case $E' = E$. 
\end{definition}

Given some function $\rho$ on the edges $E$, we define the graph distance induced by $\rho$ as:
\begin{equation}
\label{eq:d_graph}
d_\rho(x,y) = \min_{L \in \mathcal{P}_{x,y} } \rho(L)
\end{equation}
and similarly for the distances between subsets of $V$. We define the distance to be $\infty$ should there be no such path $L$ in the minimization. 
Given some $\rho$, define the region:
\begin{equation}
\label{ab0}
(AB)_0 = \{ x\in V : d_\rho(x, AB)  = 0\}
\end{equation}
and similarly define
\begin{equation}
    E_0 = \{\{x,y\}\in E : x\in (AB)_0 \text{ and } y\in (AB)_0\}
\end{equation} to be the edges that lie entirely within $(AB)_0$
\footnote{Note that $E_0$ is not the same as $E[(AB)_0]$, which is the set of edges that have \emph{some} vertex in $(AB)_0$.}.

We will construct an optimal pair $(\rho,\sigma)$ with some nice properties and such that $V' = (AB)_0$ will be our choice for $V'$ in Lemma~\ref{lem:subG}.
We note however, at this point there is no obvious reason for $\rho(e) =0 $ for all edges within $(AB)_0$.

\begin{lemma}
\label{minilem}
Given some feasible  $(\rho,\sigma)$, consider an edge $e \in E_0$, 
then we have the estimate:
\begin{equation}
\label{Pe}
\forall \,  L \in \mathcal{P}_{AB,C} : e \in L \, \, \implies \,\,  \rho(L) - \rho(e) \geq n
\end{equation}
\end{lemma}
\begin{proof}
For all paths $L \in \mathcal{P}_{AB,C}$ passing through $e= \{ x, y\}$ in the order $AB \rightarrow x \rightarrow y \rightarrow C$ we 
we note that $\rho(L) \geq  \rho(L_{y,C})+ \rho(e) $ where $L_{y,C} \in \mathcal{P}_{y,C}$, after dropping the $AB\to x$ portion of the path. But since $y \in (AB)_0$ there is some path $L^\star_{AB,y}$ from $AB$ to $y$ with
with $\rho(L^\star_{AB,y}) = 0$. We can combine the two paths, denoted as $L^\star_{AB,y} \cup L_{y,C} \in \mathcal{P}_{AB,C}$,
and use this to estimate:
\begin{equation}
\rho(L) \geq  \rho(L_{y,C})+ \rho(e)= \rho(L^\star_{AB,y}\cup L_{y,C} ) + \rho(e)
\geq n  + \rho(e)
\end{equation}
by feasibility of $\rho$. 
\end{proof}
\begin{corollary}
For a feasible $(\rho,\sigma)$ 
an edge $e \in E_0$ is  $\rho$-binding iff $\rho(e) = 0$. 
\end{corollary}
\begin{proof}
The if statement is obvious from the definition of $\rho$-binding.

Now for the only if statement: Assume the edge $e$ is $\rho$-binding.
Then either $\rho(e)=0$ or $\exists L\in \mathcal{P}_{AB,C}$ containing $e$ such that $\rho(L)=n$.
If $\rho(e)=0$ then we are done.
Assuming that it is the other case, we use \Eqref{Pe} to obtain
\begin{equation}
n - \rho(e) \geq n \,\, \implies \,\, \rho(e) = 0
\end{equation}
which is a contradiction. So we must have $\rho(e)=0$.
\end{proof}

This motivates introducing:
\begin{equation}
\hat{E}_0 = \{ e \in E_0 : \rho(e) = 0 \}
\end{equation}
as the $\rho$-binding edges in $E_0$. Note that $\hat{E}_0 \subset E_0$.

\begin{remark}
One might have expected that for an optimal solution there are simply no edges that are not $\rho$-binding, otherwise one could
get a smaller free energy by making $\rho$ smaller. One cannot directly do this because one has to consider the other constraint for paths $\mathcal{P}_{A,B}$ that also involves $\rho$. We will eventually show this is possible. That is we will prove that $E_0 = \hat{E}_0$, but we cannot do this before we prove some results on the behavior of $\sigma$, our next goal. 
\end{remark}

Define the two $\sigma$ distance measures for paths within $(AB)_0$:
\begin{equation}
\label{hatnohat}
d_\sigma^0(x,y) = \min_{L \in \mathcal{P}_{x,y} : L \subset E_0} \sigma(L), \qquad \hat{d}_\sigma^0(x,y) = \min_{L \in \mathcal{P}_{x,y} : L \subset \hat{E}_0} \sigma(L)
\end{equation}
which can be thought of a distance measures on truncated graphs. We continue to define the distance to be $\infty$ should the set of such paths above be empty.

The next lemma pertains to the region $(AB)_0$ and edges $E_0$ and $\hat{E}_0$:
\begin{lemma}
\label{lem:prop0}
Given any optimal solution $(\rho',\sigma')$, there exists an optimal solution $(\rho,\sigma)$ such that $\rho+\sigma=\rho'+\sigma'$ and it satisfies the following properties that we prove sequentially:
\begin{enumerate}
\item[(a)] For all $e \in \hat{E}_0$ then $e$ is $\sigma$-binding within $\hat{E}_0$.
\item[(b)] 
There exists a function $k : (AB)_0 \rightarrow \mathbb{Z}$ with $0 \leq k(x) \leq 2(n-1)$ such that 
$\sigma(e) = |k(x) - k(y)|$ for all $e = \{x,y\} \in E_0$ and where:
\begin{equation}
\label{tokk}
k(x) = \begin{cases}  \hat{d}_\sigma^0(x,A)\,,  & \hat{d}_\sigma^0(x,A)  < \infty \\
 2(n-1)\,, & \hat{d}_\sigma^0(x,A)  = \infty \end{cases} 
 \end{equation}
and
\begin{equation} 
\label{tokk2}
  k(x) = \begin{cases}  2(n-1) - \hat{d}_\sigma^0(x,B)\,,  & \hat{d}_\sigma^0(x,B)  < \infty  \\
 0\,, & \hat{d}_\sigma^0(x,B) = \infty \end{cases}
\end{equation} 
If $L \in \mathcal{P}_{A,B}$ is non-empty, then for all $e \in E_0\backslash\hat{E}_0$ we have $\rho(e) \in 2+ 2\mathbb{Z}_{\geq 0}$. 
\item[(c)] All edges $e \in E_0$ are $\rho$-binding. In other words  $E_0 = \hat{E}_0$.
This implies that $\hat{d}_\sigma^0$ can be replaced by $d_\sigma^0$ in \Eqref{tokk} and \Eqref{tokk2}.
\end{enumerate}
\end{lemma}

\begin{proof}
(a) Consider some optimal $(\rho',\sigma')$. We use the notation $(AB)'_0$ and $E_0'$ for the region as in \Eqref{ab0}, defined with respect to $d_{\rho'}$, and $\hat{E}_0'\subset E_0'$ as the set of edges inside $(AB)'_0$ with $\rho'(e) =0$. 
Define $E_{\rm fail}$ as the edges $e \in \hat{E}_0'$ which are not $\sigma'$-binding
within $\hat{E}_0'$.
Then set 
\begin{equation}
(\rho(e), \sigma(e) ) = \begin{cases} \big( \rho'(e) + \sigma'(e) , 0 \big) \,, & e \in E_{\rm fail} \\
(\rho'(e),\sigma'(e))\,, & {\rm otherwise} \end{cases}
\end{equation}
Note that $(AB)_0 \subset (AB)'_0$
since $\rho(e) \geq \rho'(e)$. Also $E_0 \subset E_0'$ and $\hat{E}_0 \subset \hat{E}_0'$ for the same reason. And $\hat{E}_0 \cap E_{\rm fail} = \emptyset$ since $\rho(e)=(\rho'+\sigma')(e)>0 ~\forall e\in E_{\rm fail}$. 

We aim to show that $(\rho,\sigma)$
is (i) feasible, (ii) optimal and (iii) satisfies the statement under investigation in (a). We start with (iii): Consider any edge $e \in \hat{E}_0$ with $\sigma(e) \neq 0$. Note that $\sigma(e) = \sigma'(e)$ since only edges in $E_{\rm fail}$ are changed.
Since this edge is $\sigma'$-binding within $\hat{E}_0'$, there is a path $L \in \mathcal{P}_{A,B}$ intersecting $e$ such that $L \subset \hat{E}_0'$ and $\sigma'(L) = 2(n-1)$. 
We just need to show that $L \subset \hat{E}_0$ so that $e$ is $\sigma$-binding inside $\hat{E}_0$.
Note that $L \cap E_{\rm fail} = \emptyset$ since this would otherwise contradict the definition of $E_{\rm fail}$ ($L$ being a saturating path).
Thus $\rho(L) = \rho'(L)$ since these edges are not changed, and $\rho'(L) =0$ since $L \subset \hat{E}_0'$.
This implies that $\rho(e) =0$ for all $e \in L$  implying that $L \subset \hat{E}_0$, and we are done. For later use, we note that we just proved:
\begin{equation}
\label{toconverse}
\forall L \in \mathcal{P}_{A,B} : L \subset \hat{E}_0' \,\, {\rm then} \,\, L \cap E_{\rm fail} = \emptyset \implies \rho(L) = 0
\end{equation}
 (i) Since $\rho \geq \rho'$ any $L \in \mathcal{P}_{AB,C}$ is clearly still feasible
for $\rho$.  Also since $(\sigma' + \rho')(e)
= (\sigma + \rho)(e)$ for all edges, we need only check paths $L \in \mathcal{P}_{A,B}$ such that $\delta_{\rho'(L),0}=1$ and $\delta_{\rho(L),0} = 0$. 
The condition $\rho'(L) = 0$ implies that $L \subset \hat{E}_0'$,  but then the converse of \Eqref{toconverse} implies that $\rho(L) \neq 0  \implies L \cap E_{\rm fail} \neq \emptyset$ and thus $\sigma'(L) > 2(n-1)$ by the failure of binding.
Recall that the failure of saturating for the $\mathcal{P}_{A,B}$ paths costs $+2$ in the definition of the integer program, so actually $\sigma'(L) \geq 2n$. Thus $(\sigma + \rho)(L) = \sigma'(L) \geq 2n = 2(n- \delta_{\rho(L),0})$ as required. 
(ii) Optimality is obvious since $(\rho+\sigma)(e)=(\rho'+\sigma')(e)$. 

(b) We start with an optimal $(\rho',\sigma')$ satisfying (a). Note that any point $x \in (AB)_0'$ has at least one path to either $A$ or $B$ contained in $\hat{E}_0'$ by the definition of these regions. Assume that that $\hat{d}_{\sigma'}^0(x,A) < \infty$, or in other words, assume there is some path from $x$ to $A$ inside $\hat{E}_0'$, and define:
\begin{equation}
k(x) = \hat{d}_{\sigma'}^0(x,A)
\end{equation}
where we recall the definition of this hatted distances uses paths inside $\hat{E}_0'$. 
We now aim to find a consistent set of equations \Eqref{tokk} and \Eqref{tokk2} (with $\sigma$ replaced by $\sigma'$ for now).

We note that if $\hat{d}_{\sigma'}^0(x,B) = \infty$ then we must have $k(x) = 0$ or otherwise we will violate (a) with some non $\sigma'$-binding edge along the path from $x$ to $A$. Furthermore if  $\hat{d}_{\sigma'}^0(x,B) < \infty$ then consider the edge $e$ with $\sigma'(e) \neq 0$ along the piecewise minimal path $L: A \rightarrow x \rightarrow B$ that is closest to $x$ (either towards $A$ or towards $B$) 
\begin{equation}
2(n-1) \leq \sigma'(L) = \hat{d}_{\sigma'}^0(A,x) + \hat{d}_{\sigma'}^0(x,B) \leq \sigma'(L') = 2(n-1)
\end{equation}
where $L'$ is a saturating path for $e$ that exists by (a). The first inequality is feasibility and the second comes from deforming the minimal paths in the distance
functions to the path $L'$ (see \figref{fig:sigmadeform}).
\begin{figure}[ht]
    \centering
    \includegraphics[scale=.35]{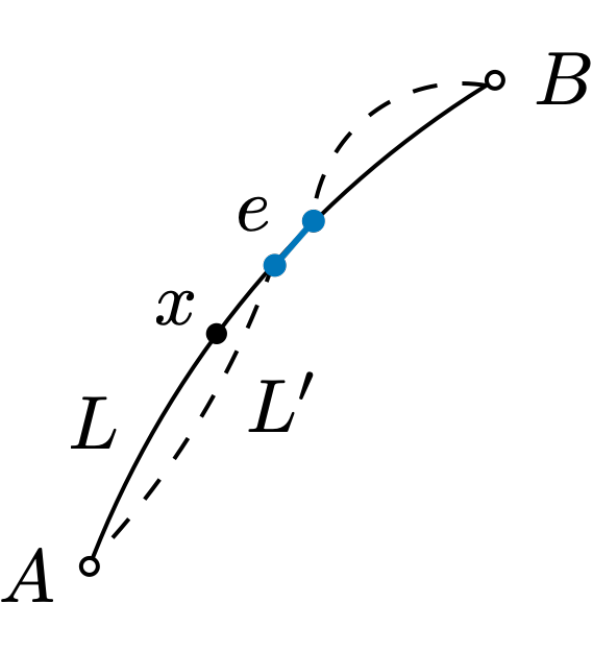}
    \caption{We deform a piecewise minimal path $L$ (solid line) containing the edge $e$ with $\sigma'(e)\ne 0$ and the vertex $x$ to a saturating path $L'$ (dashed line) containing $e$, which is guaranteed to exist by the conditions in (a).}
    \label{fig:sigmadeform}
\end{figure}
Thus:
\begin{equation}
\label{kx2n}
k(x) = 2(n-1) -  \hat{d}_{\sigma'}^0(x,B) 
\end{equation}
The only case we have not covered for \Eqref{tokk} and \Eqref{tokk2} is when $\hat{d}_{\sigma'}^0(x,A)  = \infty$. In this case we set $k(x) = 2(n-1)$ and this is consistent with \Eqref{kx2n}
since $\hat{d}_{\sigma'}^0(x,B)  =0 $ is again necessary in order to not violate (a). 

For edges $e  = \{x, y\} \in \hat{E}_0'$ we now aim to compute:
\begin{equation}
|k(x) - k(y)|
\end{equation}
There are three cases. Either (i) both $\{x, y\}$ have infinite distance to $B$, or (ii) both $\{x, y\}$ have infinite distance to $A$ or (iii) all distances are finite.
In the first two cases we have $k(x) = k(y)$. We also have $\sigma'(e) = 0$ for these cases due to (a). 

In the last case we can estimate using the triangle inequality:
\begin{equation}
\label{triin} 
| k(x) - k(y) | \leq \sigma'(e)
\end{equation}
If $\sigma'(e) =0$ then $k(x) = k(y)$, however if $\sigma'(e) > 0$, then (a) implies the existence of a saturating path $L$ such that:
\begin{equation}
2(n-1) = \sigma'(L) \geq \hat{d}_{\sigma'}^0(A,x)  + \sigma'(e) +  \hat{d}_{\sigma'}^0(y,B) 
\end{equation}
where $e = \{x,y\}$ and the saturating path behaves as $L : A \rightarrow x \rightarrow y \rightarrow B$. Thus:
\begin{equation}
 \sigma'(e)  \leq k(y) - k(x) \leq | k(y) - k(x) |
\end{equation}
which combining with \Eqref{triin} proves equality. We have now established equality $ \sigma'(e) = |k(x) - k(y)|$ for all three cases above (i-iii). 

We now consider edges $e =\{x,y\} \in E_0'$ that are not in $\hat{E}'_0$. We construct a new $(\rho,\sigma)$ that differs from the original ones on these edges:
\begin{equation}
( \rho(e) , \sigma(e) ) = \begin{cases} \big(\rho'(e) + \sigma'(e) - |k(x) - k(y)|,   |k(x) - k(y)| \big), & e \in E_0'\backslash \hat{E}_0' \\ (\rho'(e),\sigma'(e) ),  & {\rm otherwise} \end{cases}
\end{equation}
We need to show that $(\rho,\sigma)$ is (i) feasible (ii) optimal and (iii) satisfies the requirements in both (a) and (b). 

For (i) we first establish that $\rho(e) >0$ for edges $e =\{x,y\} \in E'_0\backslash \hat{E}_0'$. There are several cases to deal with depending
on whether the $\hat{d}_{\sigma'}$ distances are finite or not. See \figref{fig:prop0-b}.
\begin{figure}[ht]
    \centering
    \includegraphics[scale=.35]{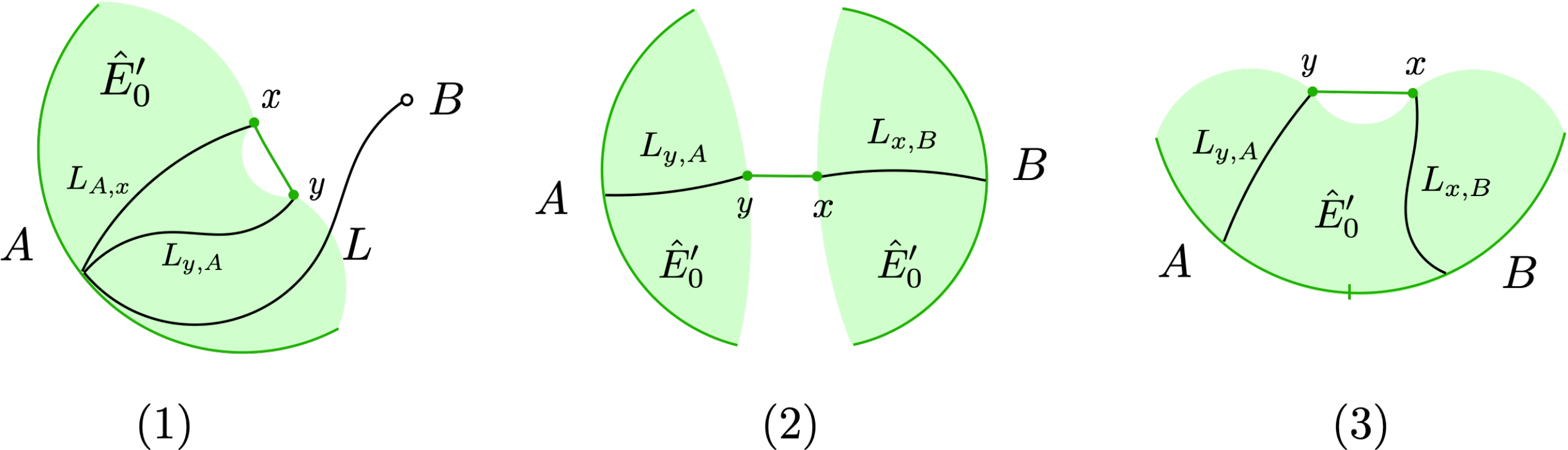}
    \caption{We construct a path from $A\to B$ containing $e=\{x,y\}$ using minimal paths within $\hat{E}_0'$ and use it to prove $\rho(e)\in 2+2\mathbb{Z}_{\ge>0}$. There are three different cases depending on whether the $\hat{d}_{\sigma'}$ distances of $x,y$ to the boundary vertices $A,B$ are finite or not. The three cases are treated differently as shown here.}
    \label{fig:prop0-b}
\end{figure}

\begin{enumerate}
    \item 
    If $x,y$ have infinite $\hat{d}_{\sigma'}^0$-distance both to $A$ or both to $B$ then
$|k(x) - k(y) | =0$ and $\rho(e) > 0$ since $\rho'(e) > 0$ on such an edge.
We will need to improve this bound for later use when establishing the statement in (b): Suppose that the $\hat{d}^0_{\sigma'}$ distances to $B$ is infinite. If there is any path $L \in \mathcal{P}_{A,B}$ then we can combine the minimal paths $L_{A,x} \subset \hat{E}_0'$ and $L_{y,A} \subset \hat{E}_0'$, that are used
to compute the respective distances $ \hat{d}_{\sigma'} (A,x) =  \hat{d}_{\sigma'} (A,y) =0$,
to form a path $L_{A,x} \cup e \cup L_{y,A} \cup L \in \mathcal{P}_{A,B}$ giving the estimate:
\begin{equation}
(\rho' + \sigma')( L_{A,x} \cup e \cup L_{y,B} \cup L) = \rho'(e) + \sigma'(e) + (\rho'+ \sigma')(L) \in 2\mathbb{Z}_{\ge0}
\end{equation}
since $\rho'$ and $\sigma'$ vanish on these minimal paths. Feasibility gives the even condition. But $(\rho' + \sigma')(L) \in 2\mathbb{Z}_{\ge0}$ also by feasibility, thus $\rho(e) \in 2 + 2 \mathbb{Z}_{\geq 0}$.
The other case where $\hat{d}^0_{\sigma'}(x,A)=\hat{d}^0_{\sigma'}(y,A)=\infty$ follows in a similar way.

\item
If $x$ has infinite $\hat{d}_{\sigma'}^0$-distance to $A$ and $y$ has infinite $\hat{d}_{\sigma'}^0$-distance to $B$ then we can construct a new path $L : A \rightarrow y \rightarrow x \rightarrow B$ using the minimal paths from $A \to y$ and the minimal
path from $x \rightarrow B$ both contained in $\hat{E}_0'$. Since these two minimal $\hat{d}_{\sigma'}^0$-distances vanishes for this new path we must have:
\begin{equation}
2n + 2 \mathbb{Z}_{\geq 0} \ni \sigma'(L) + \rho'(L) = \sigma'(e) + \rho'(e)
\end{equation}
where we used feasibility of $(\rho',\sigma')$ with the fact that $\rho'(e) > 0$. In this case we have $| k(x) - k(y) | = 2(n-1)$, implying that $\rho(e) \in 2 +2 \mathbb{Z}_{\geq 0}$. 

\item
If both $x,y$ have finite  $\hat{d}_{\sigma'}^0$-distances to $A$ and $B$ then, picking $k(y) \leq k(x)$, we can construct a path $L : A \rightarrow y \rightarrow x \rightarrow B$, as above, and feasibility now implies that:
\begin{align}
2n + 2 \mathbb{Z}_{\geq 0} \leq  \sigma'(L) + \rho'(L) &= \sigma'(e) + \rho'(e) + k(y) + 2(n-1) - k(x)  \\ & = \sigma'(e) + \rho'(e) + 2(n-1) - |k(x) - k(y)| 
\end{align}
or $\rho(e) \in 2 +2 \mathbb{Z}_{\geq 0}$. 
\end{enumerate}
We have covered all possibilities and established that $\rho(e) >0$ for such edges, and
furthermore we have established the final condition in (b). A corollary is that $(AB)_0 = (AB)_0'$ (since this only depends on the pattern of zeros in $\rho(e)$) and thus $E_0 = E_0'$ and $\hat{E}_0 = \hat{E}'_0$.

We now prove the rest of feasibility (i). Assume by contradiction that there is a path $L \in \mathcal{P}_{AB,C}$ such that $\rho(L) < n$, then this path must pass through
at least one edge $e$ with $e \in E_0'\backslash \hat{E}'_0$ because $\rho'$ was feasible so $L$ must go through one of the edges where $\rho$ is changed.
Consider the first such edge $e$ on the journey from $C \rightarrow AB$.
 It is possible to use this path to construct a new path $L'$ that uses $L$ from $C \rightarrow e$ and then follows $ e \rightarrow AB$ through $\hat{E}_0 = \hat{E}_0'$ (since both vertices in $e$ are in $(AB)_0'$ this is always possible). Thus this path new path $L'$ only intersects one edge with $e \in E'_0\backslash\hat{E}_0'$.  Then:
\begin{equation}
 \rho(L') \leq \rho(L) < n
\end{equation}
is still not-feasible for $\rho$. We use Lemma~\ref{minilem} which states that for such an edge $e$ and feasible $\rho'$:
\begin{equation}
\label{rhopp}
\rho(L') > \rho'(L') - \rho'(e) \geq n
\end{equation}
which is a contradiction. The first inequality in \Eqref{rhopp} follows from $\rho(e) > 0$, and the fact that this is the only edge that is changed relative to $\rho'$ along the path.
For paths $L \in \mathcal{P}_{A,B}$ all paths that pass through a deformed edge 
we maintain $\delta_{\rho(L),0} = \delta_{\rho'(L),0} = 0$ and $\sigma + \rho = \sigma' + \rho'$ so feasibility for these paths is clear. 
Optimality (ii) is also clear. 

Now we show (iii) that the new $(\rho,\sigma)$ satisfies conditions (a) and (b) of the Lemma.  (a) is clear since $\hat{E}_0 = \hat{E}_0'$ and we have not changed any $\sigma(e)$
inside $\hat{E}_0$. For (b) we firstly note that $\sigma(e) = |k(x) - k(y)|$ for all $e = \{x , y\} \in E_0$ by construction. Also $\hat{d}_{\sigma}^0(x,y) = \hat{d}_{\sigma'}^0(x,y)$
since $\hat{E}_0' = \hat{E}_0$ and $\sigma' = \sigma$ on these edges. We are done with (b).

(c) We consider an optimal $(\rho,\sigma)$ satisfying the properties in (a) and (b). We work by contradiction. That is we assume there is at least one edge $e \in E'_0\backslash\hat{E}'_0$ and prove a contradiction. Consider any one of these edges and call it $e_\star$. Construct a new solution:
\begin{equation}
( \rho'(e) , \sigma'(e) ) = \begin{cases} \big(0,   \sigma(e)\big) \,, & e = e_\star \\ (\rho(e),\sigma(e) ) \,, & {\rm otherwise} \end{cases}
\end{equation}
This clearly has a smaller objective, so if we can show that $(\rho',\sigma')$ is feasible we prove a contradiction since $(\rho,\sigma)$ is optimal. 
For paths $L \in \mathcal{P}_{AB,C}$ that intersect $e_\star$ -- without loss of generality we may consider paths that intersect $e_\star$ only once.
We then use Lemma~\ref{minilem} which becomes $\rho'(L) \geq n$ as required. 

Now consider paths $L \in \mathcal{P}_{A,B}$ with $e_\star \in L$. If there are no such paths, then there is nothing more to prove. If there is at least one such path we know from (b) that
$\rho(e_\star) \in 2 \mathbb{Z}_+$ and this implies that the even condition in \Eqref{eveness} is preserved for $(\rho',\sigma')$. Thus we now consider only the inequality implied in \Eqref{eveness}.
We set $e_\star =\{x, y\}$ and pick these vertices so that $L : A \rightarrow x \rightarrow y \rightarrow B$. 
Define the subpaths as $L_{A,x}$ and $L_{y,B}$ so that $L = L_{A,x} \cup e_\star \cup L_{y,B}$ .
We must consider four cases depending on whether the distances $\hat{d}_\sigma^0(x,B)$ or $\hat{d}_\sigma^0(y,A)$ are finite:

\begin{figure}[ht]
    \centering
    \includegraphics[scale=.35]{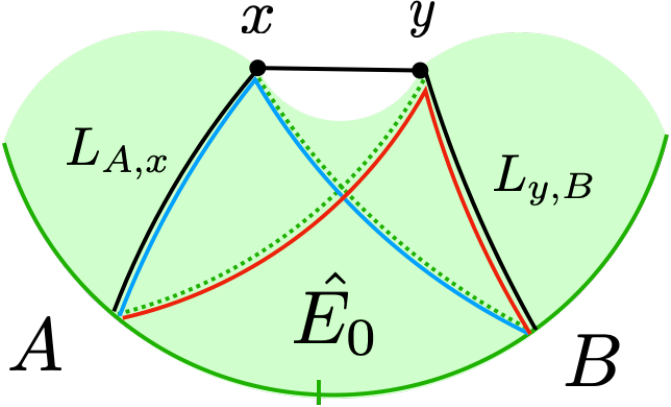}
    \caption{Starting from the path $L=L_{A,x}\cup e_{\star} \cup L_{y,B} \in \mathcal{P}_{A,B}$ (black solid) we construct two new paths in $\mathcal{P}_{A,B}$ (red and blue solid) using minimal paths  (green dashed) that are contained entirely in $\hat{E}_0$.}
    \label{fig:estardeform}
\end{figure}

\begin{enumerate}
\item If $\hat{d}_\sigma^0(x,B) < \infty$ and $\hat{d}_\sigma^0(y,A) < \infty$ then construct two new paths that by-pass $e_\star$ as follows. Consider the minimal paths for these finite distances: from $A \rightarrow y$ and $x \rightarrow B$. Join these, respectively, to $L_{y,B}$ and $L_{A,x}$ (see \figref{fig:estardeform}). These minimal paths are contained in $\hat{E}_0$ (so $\rho$ vanishes on them). We apply feasibility for $(\rho,\sigma)$ to these two new paths:
\begin{align}
\hat{d}_\sigma^0(A,y) + (\rho + \sigma)(L_{y,B})  & \geq 2(n- \delta_{\rho(L_{y,B}),0}) \\
 (\rho + \sigma)(L_{A,x})  + \hat{d}_\sigma^0(x,B) & \geq 2(n- \delta_{\rho(L_{A,x}),0})
\end{align}
adding these two bounds together and using the distance function $k(x)$ we have:
\begin{equation}
(\rho' + \sigma)(L) - \sigma(e_\star) + k(y) - k(x) \geq 2(n + 1- \delta_{\rho(L_{y,B}),0} - \delta_{\rho(L_{A,x}),0})
\end{equation}
then $\sigma(e_\star) = | k(y) - k(x) | \geq k(y) - k(x)$ (by condition (b)) and:
\begin{equation}
1-\delta_{\rho(L_{y,B}),0} - \delta_{\rho(L_{A,x}),0}
\geq - \delta_{\rho(L_{y,B}),0}  \delta_{\rho(L_{A,x}),0} = - \delta_{\rho'(L),0}
\end{equation}
combines to give:
\begin{equation}
\label{newfeas}
(\rho' + \sigma)(L) \geq 2(n - \delta_{\rho'(L),0})
\end{equation}
the required feasibility statement. 

\item If $\hat{d}_\sigma^0(x,B) = \infty$ and $\hat{d}_\sigma^0(y,A) = \infty$ then the form of the function $k(x)$ from (b) requires that
$k(x) = 0$ and $k(y) = 2(n-1)$. This implies $\sigma(e_\star) = 2(n-1)$.
Then we note the bound:
\begin{align}
(\rho + \sigma)(L) &\geq (\rho + \sigma)(e_\star)  + 2 (1- \delta_{\rho(L_{y,B}),0}  \delta_{\rho(L_{A,x}),0}) \\
&= \rho(e_\star) + 2 (n- \delta_{\rho'(L),0})
\end{align}
which follows by dropping all contributions from the path $L_{y,B}$ and $L_{A,x}$ except for a crude estimate counting a minimal contribution if either $\rho(L_{y,B})$ 
or $\rho(L_{A,x})$ is non-zero. The bound that we get from this minimal contribution must be even, since $(\rho + \sigma)(e_\star)$ is even, and $(\rho + \sigma)(L) $ is even by feasibility. So any gap between them must be even. So again we get \Eqref{newfeas}. 

\item If $\hat{d}_\sigma^0(x,B) < \infty$ and $\hat{d}_\sigma^0(y,A) = \infty$ (the reverse case follows a similar argument) then we must have $k(y) = 2(n-1)$.
We also must have  $\hat{d}_\sigma^0(x,A) = \infty$ since otherwise we could construct a path $A \rightarrow x \rightarrow B \rightarrow y$ inside $\hat{E}_0$ and this
would violate the condition $\hat{d}_\sigma^0(y,A) = \infty$. So $k(x) = 2(n-1)$ and hence $\hat{d}_\sigma^0(x,B)  = 0$ and $\sigma(e_\star) = 0$. We estimate:
\begin{equation}
\label{onebound}
(\rho + \sigma)(L) \geq (\rho + \sigma)(L_{A,x}) + \rho(e_\star) + 2 ( 1- \delta_{\rho(L_{y,B}),0}) 
\end{equation}
where we again crudely dropped all contributions from $L_{y,B}$ except if $\rho(L_{y,B})$ is non-zero. Evenness also demands the gap in the bound is $2$. 
We also consider a combined path that bypasses $e_\star$ using $L_{A,x}$ and the minimal $\hat{d}_\sigma^0$-distance path from $y$ to $B$ inside $E_0$
and apply feasibility:
\begin{equation}
\label{twobound}
(\rho + \sigma)( L_{A,x}) + \hat{d}_\sigma^0(x,B) \geq 2(n - \delta_{\rho(L_{A,x}),0})
\end{equation}
Combinding \Eqref{onebound} and \Eqref{twobound} gives:
\begin{equation}
(\rho' + \sigma)(L) \geq 2(n + 1- \delta_{\rho(L_{y,B}),0} - \delta_{\rho(L_{A,x}),0}) \geq  2 (n- \delta_{\rho'(L),0})
\end{equation}
as required.
\end{enumerate} 
We have now established feasibility for all possible cases and so we find the desired contradiction. We conclude that there are no such edges and
$E_0 = \hat{E_0}$. 
\end{proof}

We now move on to prove Lemma~\ref{lem:subG}.
We use the optimal $(\rho,\sigma)$ constructed in Lemma~\ref{lem:prop0}.
We set $V' = (AB)_0$ and $E' = E_0$. The above result establishes that $\rho$ vanishes on $E'$ where only $\sigma$ is non-zero.
We now use this  as an input to the following half integer program that lives on the complementary reduced graph defined as:
\begin{equation}
G^c = (V^c, E^c) \, \qquad V^c = (V \backslash V') \cup (AB)' \qquad E^c = E \backslash E'
\end{equation}
with $(AB)' = V \cap \mu_G(V')$.  
\begin{lemma}
\label{lem:plugin}
Given an optimal $(\rho,\sigma)$ satisfying the properties in Lemma~\ref{lem:prop0}, then
\begin{equation}
\label{plugin}
\varrho(e) = (\rho(e)+ \sigma(e))|_{E^c} - \mathbf{1}_{\mu(V')}(e)
\end{equation}
 is feasible for the following half integer program on the graph $G^c = (V^c, E^c)$:
\begin{align}
\label{tohere}
M \equiv  &\min_{\varrho}   \sum_{e \in  E^c}  w(e) \varrho(e) \\ 
&{\rm subject \,\, to}\, \quad \forall L \in \mathcal{P}_{\Gamma_k ,\Gamma_{k'}}\,\, : \varrho(L) \in  |k-k'|  + \mathbb{Z}_{\geq 0} \\ 
&{\rm and \,\, }\, \qquad\forall L \in \mathcal{P}_{\Gamma_k,C} \, \, : \varrho(L) \in (n-1) + \mathbb{Z}_{\geq 0} \\
&{\rm for\,\,all}\quad k,k'=0,\ldots,2(n-1)
\end{align}
where $\varrho(e) \in \mathbb{Z}_{\geq 0}/2$ and $\Gamma_k = \{ x \in (AB)': k(x) = k \} $ and $C$, are boundary vertices on the reduced graph. Thus:
\begin{equation}
\label{bdbd}
I \geq M + w(\mu(V'))  + \sum_{e =\{x,y\} \in E'} w(e) |k(x) - k(y)|
\end{equation}
\end{lemma}
\begin{proof}
We firstly check feasibility for paths $\mathcal{P}_{\Gamma_k, \Gamma_{k'}}$. We start by picking a subset of
such paths (possibly empty) $\widetilde{\mathcal{P}}_{\Gamma_k, \Gamma_{k'}}$ for all $k,k'$, with the extra condition that
the path only intersects the boundary edges $\mu( V')$ twice (this need not always be the case for all paths, even ones that are edge disjoint). 

Let $k' < k$ (the case $k' = k$ is trivial) and consider $\widetilde{L} \in \widetilde{\mathcal{P}}_{\Gamma_k, \Gamma_{k'}}$ .  We can construct a path in $\mathcal{P}_{A,B}$ by attaching minimal curves for the distance $d_\sigma^0$ through $E'$.
We consider such curves from $A \rightarrow \Gamma_{k'}$ and $\Gamma_k \rightarrow B$. 
We apply feasability to the combination:
\begin{equation}
(\sigma + \rho)(\widetilde{L}) + (k' -k) + 2(n-1) \geq 2 n
\end{equation}
Using $\mathbf{1}_{\mu(V')}(\widetilde{L}) = 2$ for such curves gives
\begin{equation}
\label{widetilde}
\varrho(\widetilde{L}) \geq |k- k'|
\end{equation}
Now any curve $L \in \mathcal{P}_{\Gamma_k, \Gamma_{k'}}$ can be constructed as a sequence of these restricted $\widetilde{L}$ curves from $\Gamma_k \rightarrow \Gamma_{k_1} 
\rightarrow    \ldots \Gamma_{k_N} \rightarrow \Gamma_{k'}$ for arbitrary $k_i : i =1 \ldots N$. Thus:
\begin{equation}
\varrho(L) \geq |k -k_1| + | k_1 - k_2 | + \ldots | k_N - k'| \geq | k - k'|
\end{equation}
by the triangle inequality and \Eqref{widetilde}. That concludes feasability for $\mathcal{P}_{\Gamma_k, \Gamma_{k'}}$.

For a paths $\mathcal{P}_{ \Gamma_k, C}$ we again have to deal with possible multiple intersections with $\mu(V')$. 
We again restrict to $\widetilde{\mathcal{P}}_{ \Gamma_k, C}$ such that these paths intersect $\mu(V')$ only once. 
Consider $\widetilde{L} \in \widetilde{\mathcal{P}}_{ \Gamma_k, C}$ and attach a minimal curve to $AB$ in the $E'$ graph.
Then:
\begin{equation}
\varrho(\widetilde{L}) = \rho(\widetilde{L}) + \sigma( \widetilde{L}) -1 \geq \rho(\widetilde{L}) -1 \geq n -1
\end{equation}
where in the first inequality we simply dropped the $\sigma$ contribution and in the second we used feasibility for the combined curve $ \in \mathcal{P}_{AB,C}$ (the minimal curve part sits in the region with $\rho = 0$.) Again any curve in the more general set $L \in \mathcal{P}_{ \Gamma_k, C}$ can always be written as a combination of curves  $\Gamma_k \rightarrow \Gamma_{k_1} \rightarrow    \ldots \Gamma_{k_N} \rightarrow C$. Thus we have:
\begin{equation}
\varrho(L) \geq |k-k_1| + \ldots + | k_{N-1} - k_N| + \,n-1 \geq n-1
\end{equation}
as required. We conclude that $\varrho$ is feasible for this intersecting cut problem. The estimate \Eqref{bdbd} now follows by plugging in \Eqref{plugin} to \Eqref{tohere} and finally using the form of $\sigma$ implied by Lemma~\ref{lem:prop0} on the rest of the edges $E'$.
\end{proof}
\begin{remark}
The half integer relaxation for this program is convenient in the next section. The above feasible $\varrho$ is integer valued, but the optimal solution might be half integral. 
However this does not bother us, since at this stage we only strive for an inequality. Also once the chain of Theorem~\ref{thm:RPBI} collapses the optimal solution will be integral. 
\end{remark}

We have now proven Lemma~\ref{lem:subG}, as can be seen by picking and choosing results from Lemma~\ref{lem:prop0} and Lemma~\ref{lem:plugin}. 

\subsection{Lemma~\ref{lem:optk}: An intersecting cut problem}
 \label{sec:intcut}
 \label{app:optk}
 
 We now study the $\ell$-intersecting cut problem defined in Definition~\ref{def:kint}.
 \begin{figure}[h]
 \centering
 \includegraphics[scale=.4]{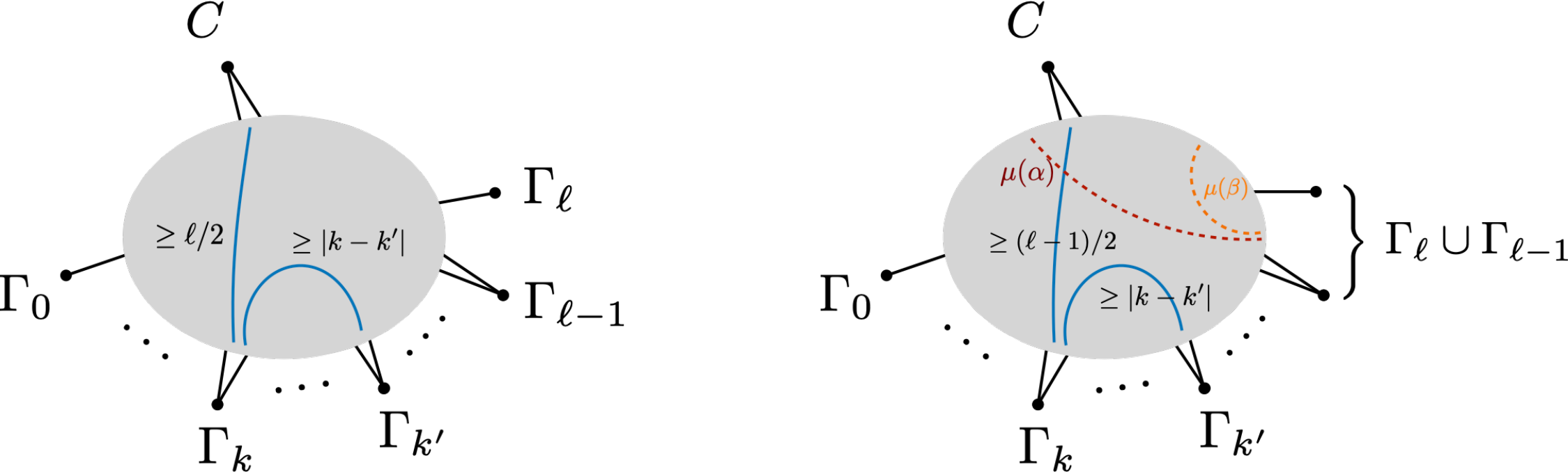} 
 \caption{(left) A sketch of a $\ell$-intersection problem defined in Definition~\ref{def:kint}. 
 (right) The reduction of the $\ell$-intersection cut problem to the $\ell-1$-intersection cut problem.}
 \label{fig:int-programs}
 \end{figure}
 We will solve this problem by a recursive reduction of the $\ell$-intersection cut problem to the $\ell-1$-intersection cut problem, as described below in Lemma~\ref{kkm1}.
 We will use a feasible solution to the $\ell$-intersection cut problem to construct a feasible solution to the $\ell-1$-intersection cut problem. See \figref{fig:int-programs}.
\begin{lemma}
\label{kkm1}
Given an feasible $\varrho$ for the $\ell$-intersecting cut problem with $\{ \Gamma_0, \Gamma_1, \ldots, \Gamma_{\ell} \}$ then
there exists a cut, $\alpha$, for $(\Gamma_0 \cup \Gamma_1 \cup \ldots \Gamma_{\ell-1}) : (\Gamma_\ell \cup C )$
and a cut, $\beta$, for
$\Gamma_\ell : (\Gamma_0 \cup \Gamma_1 \cup \ldots \Gamma_{\ell-1}\cup C)$ where $\beta$ is disjoint to $\alpha$ such that:
\begin{itemize}
\item  for $\ell > 1$ there is a feasible $\varrho'$ for the $\ell-1$-intersecting
cut problem for $\{ \Gamma_0, \Gamma_1, \ldots, \Gamma_{\ell-1} \cup \Gamma_\ell \}$  with:
\begin{equation}\label{dM}
M(\varrho) = M(\varrho') + \frac{1}{2} \left( w(  \mu(\alpha) ) +    w(  \mu(\beta) ) \right) 
\end{equation}
\item for $\ell=1$ we simply have the bound:
\begin{equation}
\label{Mbd}
M(\varrho) \geq \frac{1}{2} \left( w(  \mu(\alpha) ) +    w(  \mu(\beta) ) \right) 
\end{equation}
\end{itemize}

\end{lemma}
Before we present a proof, we use the above result to prove Lemma~\ref{lem:optk}.
\begin{proof}[Proof of Lemma~\ref{lem:optk}]
Starting from an optimal solution $\varrho$ for the $\ell$-intersecting cut problem $M(\varrho)$, we apply Lemma~\ref{kkm1} repeatedly to arrive at
\begin{equation}
    M(\varrho) \ge \frac{1}{2}\sum^{\ell-1}_{k=0} \left( w(\mu(\alpha_k))+w(\mu(\beta_k)) \right)
\end{equation}
where $\alpha_k$ is a cut for $(\Gamma_0\cup\ldots\Gamma_k):(\Gamma_{k+1}\cup\ldots\Gamma_\ell\cup C)$ and $\beta_k$ is a cut for $(\Gamma_{k+1}\cap\ldots\Gamma_\ell):(\Gamma_0\cup\ldots\Gamma_k\cup C)$.
Minimizing over the all such cuts $\alpha_k$ and $\beta_k$ then gives
\begin{equation}
    M(\varrho) \ge \frac{1}{2}\sum^{\ell-1}_{k=0} \left( w(\mu(\alpha'_k))+w(\mu(\beta'_k)) \right)
\end{equation}
where $\alpha'_k$ and $\beta'_k$ are the minimal cuts.
We now use them to construct a new $\rho'$, defined by
\begin{equation}
    \varrho'(e) = \frac{1}{2}\sum_{k=0}^{\ell-1}(\mathbf{1}_{\mu(\alpha'_k)}+\mathbf{1}_{\mu(\beta'_k)})(e)
\end{equation}
It is clear that $\varrho'$ is feasible from the topology of the cuts (see \figref{fig:rhofeasible}) so $M(\varrho)\ge M(\varrho')$. 
\begin{figure}[h]
    \centering
    \includegraphics[scale=.35]{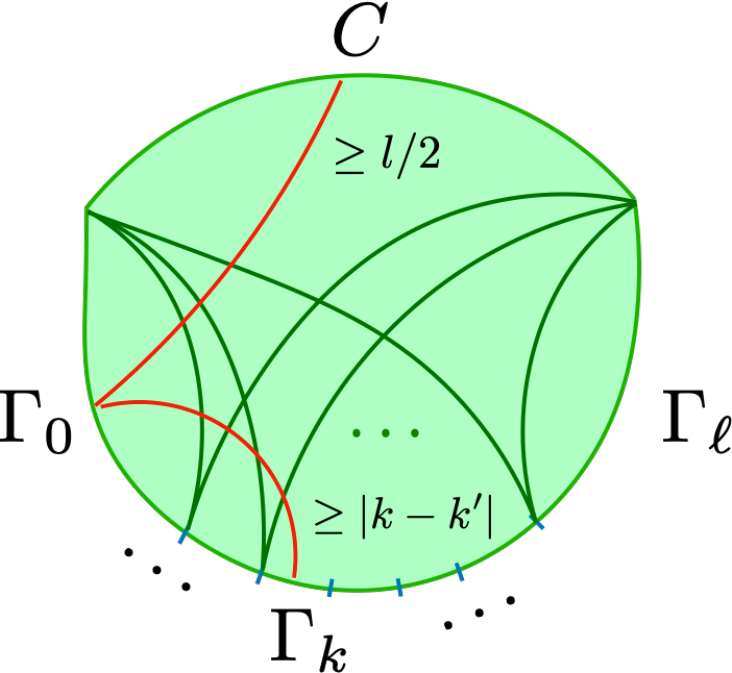}
    \caption{An example configuration to the solution of $M(\varrho')$. The green solid lines represent minimal cuts $\alpha_k$ and $\beta_k$. Consider a path in $L \in \mathcal{P}_{\Gamma_k,C}$ (depicted as red line). It is clear that it must cross at least $\ell$ minimal surfaces in order to reach $C$ so $w(L)\ge \ell/2$. Similarly, for a path in $L'\in\mathcal{P}_{\Gamma_k,\Gamma_{k'}}$ it must cross at least $2|k-k'|$ minimal surfaces so $w(L')\ge |k-k'|$.}
    \label{fig:rhofeasible}
\end{figure}
Also, 
\begin{equation}
 M(\varrho')=\frac{1}{2}\sum_{e\in E}w(e)\sum_{k=0}^{\ell-1}(\mathbf{1}_{\mu(\alpha'_k)}+\mathbf{1}_{\mu(\beta'_k)})(e) = \frac{1}{2}\sum^{\ell-1}_{k=0} \left( w(\mu(\alpha'_k))+w(\mu(\beta'_k)) \right)
\end{equation}
so $M(\varrho')\ge M(\varrho)$ since $\alpha'_k$ and $\beta'_k$ are minimal cuts.
Thus we have inequalities in both ways and it must be that
\begin{equation}
    M=M(\varrho) = M(\varrho') = \frac{1}{2}\sum^{\ell-1}_{k=0} \left( w(\mu(\alpha'_k))+w(\mu(\beta'_k)) \right)
\end{equation}
\end{proof}

\begin{proof}
We first consider the region $\beta$. Define this as:
\begin{align}
 \beta &= \{ x \in V : d_\varrho(x, \Gamma_\ell) = 0 \}
\end{align}
It is clear that this satisfies the cut properties stated in the Lemma. 
It is also clear that $\tilde{\varrho}(e) := \varrho(e) - (1/2) \mathbf{1}_{\mu(\beta)}(e) \geq 0$. We show that $\tilde{\varrho}$ is feasible for the following $1/2$ integer
program:
\begin{alignat}{3}
\label{tildeM}
\tilde{M} &\equiv  \min_{\tilde{\varrho}} \tilde{M}(\tilde{\varrho}), \qquad \tilde{M}(\tilde{\varrho}) =   \sum_{e \in  E}  w(e) \tilde{\varrho}(e) \span\span\span\span  \\
&{\rm subject \,\, to} \qquad &&\forall L \in \mathcal{P}_{\Gamma_k,\Gamma_{k'}}: && \tilde{\varrho}(L) \in  |k-k'| + \mathbb{Z}_{\geq 0} \\ 
&{\rm and \,\, } \qquad && \forall L \in \mathcal{P}_{\Gamma_k,C}: && \tilde{\varrho}(L) \in \ell/2 + \mathbb{Z}_{\geq 0} \\
&{\rm subject \,\, to} \qquad &&\forall L \in \mathcal{P}_{\Gamma_k ,\Gamma_{\ell}}: && \tilde{\varrho}(L) \in  (\ell-k-1/2) + \mathbb{Z}_{\geq 0} \\ 
\label{seegap}
&{\rm and \,\, } \qquad &&\forall L \in \mathcal{P}_{\Gamma_\ell,C}: && \tilde{\varrho}(L) \in (\ell-1)/2 + \mathbb{Z}_{\geq 0}\\
&\text{for all} \quad k,k'=0,\cdots,\ell-1.\span\span\span\span \nonumber
\end{alignat}
This program is defined for $\ell\geq 1$. The last constraint is trivial if $\ell=1$. We sketch this program in \figref{fig:k-half-program}.
 \begin{figure}[h]
 \centering
 \includegraphics[scale=.4]{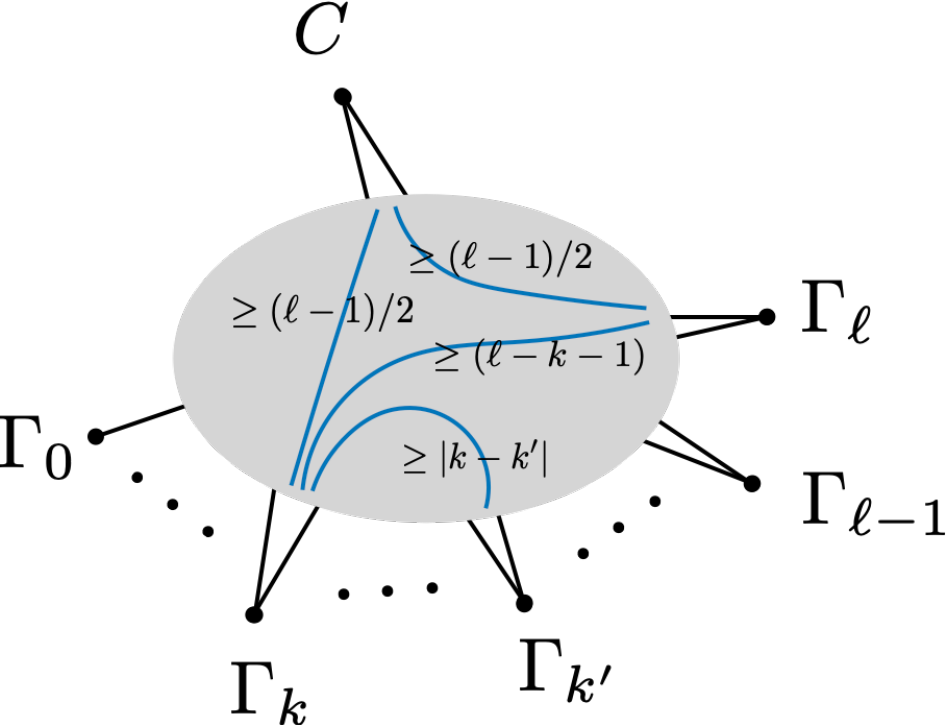}
 \caption{The intermediate half integer program described in \Eqref{tildeM}.}
 \label{fig:k-half-program}
 \end{figure}

Feasability is clear for paths that intersect $\mu(\beta)$ the minimal number of times, that is the subset of paths defined below:
\begin{alignat}{4}
\text{(I)}& \qquad &&\widetilde{\mathcal{P}}_{\Gamma_k, \Gamma_{k'}} &&= \{ L \in \mathcal{P}_{\Gamma_k, \Gamma_{k'}} : \mathbf{1}_{\mu(\beta)}(L)= 0\} &&\geq |k-k'| \\
\text{(II)}& \qquad &&\widetilde{\mathcal{P}}_{C, \Gamma_{k}}  &&= \{ L \in \mathcal{P}_{C, \Gamma_{k}} :\mathbf{1}_{\mu(\beta)}(L) = 0\} &&\geq \ell/2\\
\label{bdIII}
\text{(III)}& \qquad &&\widetilde{\mathcal{P}}_{\Gamma_k, \Gamma_{\ell}}  &&= \{ L\in \mathcal{P}_{\Gamma_k, \Gamma_{\ell}} : \mathbf{1}_{\mu(\beta)}(L) = 1\}  &&\geq (\ell-k-1/2) \\
\label{bdIV}
\text{(IV)}& \qquad &&\widetilde{\mathcal{P}}_{C, \Gamma_{\ell}} &&= \{ L \in \mathcal{P}_{C, \Gamma_{\ell}} : \mathbf{1}_{\mu(\beta)}(L)  = 1 \}  &&\geq (\ell-1)/2,
\end{alignat}
for all $k,k'=0,\cdots,\ell-1$. We have listed the constraints for the $\varrho'$ problem on the right and on the left we have given labels to the various cases of paths.
It is also clear we maintain the integer condition for $\tilde{\varrho}$ in \Eqref{tildeM} due to the topology of the paths, so we need only consider the inequalities below.
Any other path not in this class can be decomposed using these paths. There are four different cases to consider (see \figref{fig:midtypes}):

\begin{figure}
    \centering
    \includegraphics[scale=.35]{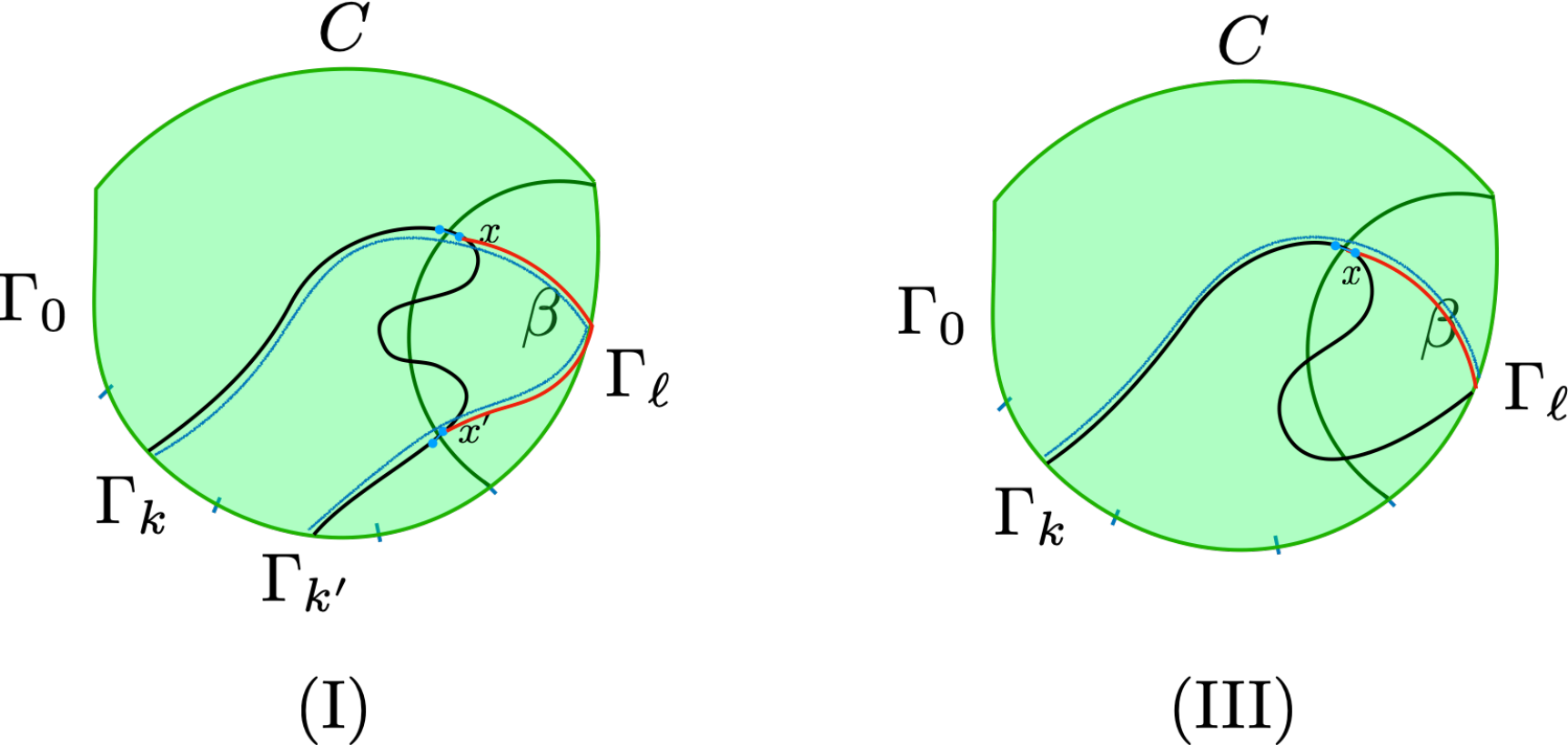}
    \caption{We perform surgery on a path $L$ (black solid) at point $x$ and $x'$ using minimal paths (red solid) and reduce $L$ to the union of paths (blue solid) and use it to bound $\tilde{\varrho}(L)$. The procedure prescribed for type-I and type-III paths are depicted here. The procedure for type-II and type-IV paths are similar.}
    \label{fig:midtypes}
\end{figure}

\begin{itemize}
    \item[(I)]
    If $L \in \mathcal{P}_{\Gamma_k, \Gamma_{k'}}$ with $ 0 \leq k < k' \leq \ell-1$ we can bound these paths via:
    \begin{align}
    \label{ineqPP}
    \tilde{\varrho}(L) + \tilde{\varrho}(L_{x, \Gamma_\ell}) + \tilde{\varrho}(L_{x', \Gamma_\ell}) &\geq \tilde{\varrho}(\tilde{L}_{\Gamma_\ell, \Gamma_k}) 
    +  \tilde{\varrho}(\tilde{L}_{\Gamma_\ell, \Gamma_{k'}}) \\ &\geq  (\ell-k-1/2) +  (\ell-k'-1/2)  \geq (k' -k)+1
    \nonumber
    \end{align}
    where $x$ and $x'$ are the first and last points inside $\beta$ where the path $L$ enters and leaves. We have sewn on paths $L_{x, \Gamma_\ell}$ and $L_{x', \Gamma_\ell}$
    to these points, where we can pick these paths as the ones minimizing the distance $d_{\varrho}(x,\Gamma_\ell) = 0$ and $d_{\varrho}(x',\Gamma_\ell) =0$. 
    In particular for these paths $\tilde{\varrho}(P_{x, \Gamma_\ell}) = \varrho(P_{x, \Gamma_\ell}) = 0$ and similarly for $x'$.  The first inequality in \Eqref{ineqPP} drops the mid portion
    of the curve and applies the bound \Eqref{bdIII} for a curve that intersects $\mu(\beta)$ once, that is $\tilde{L}_{\Gamma_\ell, \Gamma_{k}}$ and $\tilde{L}_{\Gamma_\ell, \Gamma_{k'}}$.

    \item[(II)]
    If $L \in \mathcal{P}_{C, \Gamma_{k}}$ for $0 \leq k \leq \ell-1$ we can, in a similar manner as above, split this into two and show:
    \begin{equation}
    \tilde{\varrho}(L) 
    \geq  (\ell-k-1/2) +  (\ell-1)/2  \geq \ell/2
    \end{equation}
    where we applied \Eqref{bdIII} and \Eqref{bdIV}.
    
    \item[(III)]
    If $L \in \mathcal{P}_{\Gamma_k, \Gamma_{\ell}}$ for $0 \leq k \leq \ell-1$, then we simply drop the portion of the path after the first intersection with $x \in \beta$ along the path $\Gamma_k \rightarrow x \rightarrow \Gamma_\ell$. Adding the curve $\tilde{\varrho}(L_{x, \Gamma_\ell}) =0$ gives the estimate:
    \begin{equation}
    \tilde{\varrho}(L) 
    \geq  (\ell-k-1/2)
    \end{equation}
    where we applied \Eqref{bdIII}.
    
    \item[(IV)]
    If $L \in \widetilde{\mathcal{P}}_{C, \Gamma_{\ell}}$ we do the same and drop the portion of the path after the first intersection to give:
    \begin{equation}
    \tilde{\varrho}(L) 
    \geq  (\ell-1)/2
    \end{equation}
    where we again applied \Eqref{bdIV}. 
\end{itemize}
This completes the proof that $\tilde{\varrho}$ is feasible for \Eqref{tildeM}.  

We now introduce the region:
\begin{equation}
\label{threesets}
\alpha^c = \{ x : d_{\tilde{\varrho}} (x, \Gamma_\ell)  +  d_{\tilde{\varrho}} (x,C)  = (\ell-1)/2\} \cup \{ x :  d_{\tilde{\varrho}}(x,\Gamma_\ell) = 0 \} \cup\{ x: d_{\tilde{\varrho}}(x,C) = 0 \}
\end{equation}
We check that it satisfies the cut properties. It is clear that $C\cup \Gamma_\ell \subset \alpha^c$ by definition. 
Assume that $\Gamma_k \in \alpha^c$ for some $ 0 \leq k \leq \ell-1$.  Thus either:
\begin{equation}
 (\ell-1)/2 =  d_{\tilde{\varrho}} (\Gamma_k, \Gamma_\ell)  +  d_{\tilde{\varrho}} (\Gamma_k,C)   \geq (\ell-k-1/2) + \ell/2 \geq (\ell+1)/2
\end{equation}
which is not possible. Or $ d_{\tilde{\varrho}} (\Gamma_k, \Gamma_\ell) \geq (\ell-k-1/2)$ which is not possible or $ d_{\tilde{\varrho}} (\Gamma_k, C) \geq \ell/2$ which is also not possible. 
Thus we have a contradiction and $\Gamma_k \in \alpha$. This establishes the cut properties stated. 

We now define:
\begin{equation}
\varrho'(e) =  \tilde{\varrho}(e) - \frac{1}{2} \mathbf{1}_{\mu(\alpha)}(e)
\end{equation}
We aim to show that $\varrho'(e) \geq 0$. Consider an edge $ e=\{x,y\}\in\mu(\alpha)$ with $x \in \alpha^c$ and $ y \in \alpha$. The triangle inequality to $C$ states that:
\begin{equation}
\tilde{\varrho}(e) = d_{ \tilde{\varrho}}(x,y) \geq (d_{ \tilde{\varrho}}(y,C) - d_{ \tilde{\varrho}}(x,C))
\end{equation}
So if $d_{ \tilde{\varrho}}(x,C)  = 0$ then $\tilde{\varrho}(e) \geq d_{ \tilde{\varrho}}(y,C) > 0$ since $y \in \alpha$ and so must have this strictly greater than $0$.
Similarly for $\Gamma_\ell$:
\begin{equation}
\tilde{\varrho}(e) \geq d_{ \tilde{\varrho}}(x,y) = (d_{ \tilde{\varrho}}(y,\Gamma_\ell) - d_{ \tilde{\varrho}}(x,\Gamma_\ell))
\end{equation}
So $d_{ \tilde{\varrho}}(x,\Gamma_\ell)  = 0$ then $\tilde{\varrho}(e) > 0$. Finally if $ d_{\tilde{\varrho}} (x, \Gamma_\ell)  +  d_{\tilde{\varrho}} (x,C)  = (\ell-1)/2$ we add the two inequality above to show that:
\begin{equation}
2\tilde{\varrho}(e) \geq d_{ \tilde{\varrho}}(y,C) + d_{ \tilde{\varrho}}(y,\Gamma_\ell) > 0 
\end{equation}
Thus in all case we have $\tilde{\varrho}(e) \geq 1/2$, by the integrality gap.  Indeed $\varrho'(e) \geq 0$. 

We thus have: 
\begin{equation}
M(\varrho) = M(\varrho') + \frac{1}{2} \left( w(  \mu(\alpha) ) +    w(  \mu(\beta) ) \right) 
\end{equation} 
as required. For $\ell =1$ we simply bound $M(\varrho') \geq 0$ and we are done. For $\ell >1$ we need to check feasibility of $\varrho'$ for the $\ell-1$-intersecting cut program.
Paths that cross $\mu(\alpha)$ a minimal number of times are clearly feasible:
\begin{alignat}{4}
\label{adI}
\text{(I)}& \qquad &&\widetilde{\mathcal{P}}_{C, \Gamma_{\ell}} &&= \{ L \in \mathcal{P}_{C, \Gamma_{\ell}} : \mathbf{1}_{\mu(\alpha)}(L)  = 0 \} && \geq (\ell-1)/2 \\
\text{(II)}& \qquad &&\widetilde{\mathcal{P}}_{C, \Gamma_{k}} &&= \{ L \in \mathcal{P}_{C, \Gamma_{k}} :\mathbf{1}_{\mu(\alpha)}(L) = 1\}  && \geq (\ell-1)/2\\
\label{adIII}
\text{(III)}& \qquad &&\widetilde{\mathcal{P}}_{\Gamma_k, \Gamma_{\ell}} &&= \{ L \in \mathcal{P}_{\Gamma_k, \Gamma_{\ell}} : \mathbf{1}_{\mu(\alpha)}(L) = 1\} && \geq (\ell-k-1) \\
\text{(IV)}& \qquad &&\widetilde{\mathcal{P}}_{\Gamma_k, \Gamma_{k'}} &&= \{ L \in \mathcal{P}_{\Gamma_k, \Gamma_{k'}} : \mathbf{1}_{\mu(\alpha)}(L)= 0\} && \geq |k-k'|
\end{alignat}
for all $k,k'=0,\cdots,\ell-1$, where we have listed the constraints for the $\varrho'$ problem on the right and on the left we have given labels to the various cases of paths.

We now prove a basic result that will seed the rest of our discussion. We consider a path $L \in \mathcal{P}_{C,\Gamma_\ell}$ but now with $\mathbf{1}_{\mu(\alpha)}(L) \geq 2$, see \figref{fig:twointersections}.
\begin{figure}[h]
\centering
\includegraphics[scale=.35]{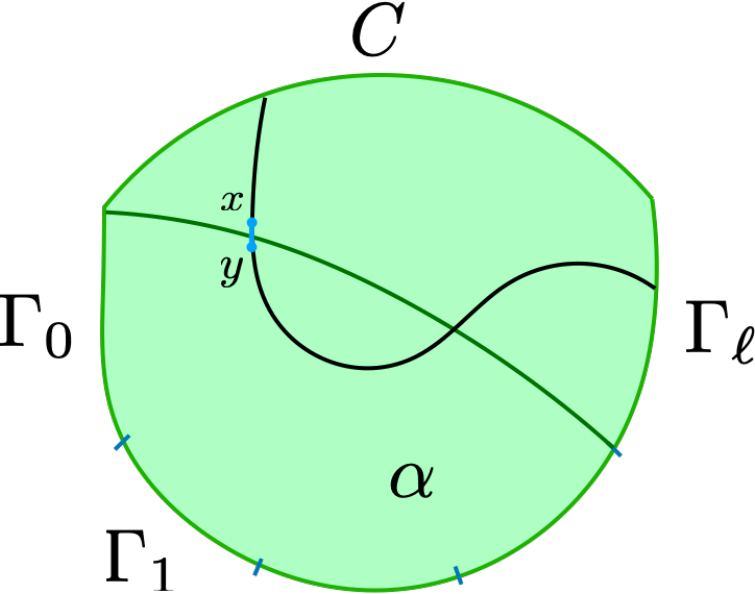}
\caption{We consider a path from $C$ to $\Gamma_\ell$ that intersect the cut surface $\mu(\alpha)$ at least twice and we denote $e=\{x,y\}$ to be the first edge in the path that crosses $\mu(\alpha)$ (shown as the green solid line).}
\label{fig:twointersections}
\end{figure}
Consider the first edge $e$ in $\mu(\alpha)$ the path crosses (starting at $C$.)  Let $e =\{x,y\}$ with $x \in \alpha^c$ and $y \in \alpha$. 
We know that:
\begin{equation}
d_{\tilde{\varrho}}(y,C) + d_{\tilde{\varrho}}(y,\Gamma_\ell) > (\ell-1)/2
\end{equation}
Using the minimality of these paths and comparing these to the two segments of $L$ split at $y$ we find:
\begin{equation}
\label{torepeat}
\tilde{\varrho}(L) > (\ell-1)/2 \quad \implies \quad \tilde{\varrho}(L)  \geq (\ell-1)/2 + 1
\end{equation}
where we used the fact that the gap for such paths is $1$ (see \Eqref{seegap}.) If this curve had only two intersections with $\alpha$ ($\mathbf{1}_{\mu(\alpha)}(L) = 2$) we would be done since then we have shown that $\varrho'(L) \geq (\ell-1)/2$. 

Another obvious bound applies to any path $L$:
\begin{equation}
\label{obvious}
\tilde{\varrho}(L)  \geq \frac{1}{2} \mathbf{1}_{\mu(\alpha)}(L) 
\end{equation}
since we know $\tilde{\varrho}$ on these edges. This bound is too crude to be used on its own, but we will still make use of it below when we start performing surgery on the paths and we find paths that start and end on the same boundary region. In this later case the bound we just derived can be tight. 

We now address the four types of paths:
\begin{itemize}
\item[(I)] Consider a path $L \in \mathcal{P}_{C,\Gamma_\ell}$ intersecting $N$ times with $\mu(\alpha)$, where $N \geq 2$ is even. We aim to show that $\tilde{\varrho}(L) \geq (\ell-1)/2 + N/2$. We do this by induction. We have proved the case $N=2$ in \Eqref{torepeat}. We assume it is true for $N-2$ and prove if for $N$. Consider the second edge $e = \{x,y\}$ that intersects $\mu(\alpha)$ along the path $C \rightarrow y \rightarrow x \rightarrow \Gamma_\ell$, where $x \in \alpha^c$. There are now three cases (a,b,c) to consider depending on which set $x$ belongs to in \Eqref{threesets}.

(a)  If $d_{\tilde{\varrho}}(x,C) + d_{\tilde{\varrho}}(x,\Gamma_\ell) = (\ell-1)/2$ then consider the minimal paths $L_{x,C}, L_{x,\Gamma_\ell}$ defining these two distances. 
We show that both of these curves $L_{x,C}, L_{x,\Gamma_\ell}$ lie entirely inside $\alpha^c$. If not there would be some $y \in \alpha$ (the first vertex where
either $L_{x,C}, L_{x,\Gamma_\ell}$ leaves $\alpha^c$) with 
$d_{\tilde{\varrho}}(y,C) + d_{\tilde{\varrho}}(y,\Gamma_\ell) = (\ell-1)/2$ and this is a contradiction.
We use $L_{x,C}, L_{x,\Gamma_\ell}$  to perform surgery on $L$ as in the first figure shown in \figref{fig:typeI}.
\begin{figure}[h]
\centering
\includegraphics[scale=.35]{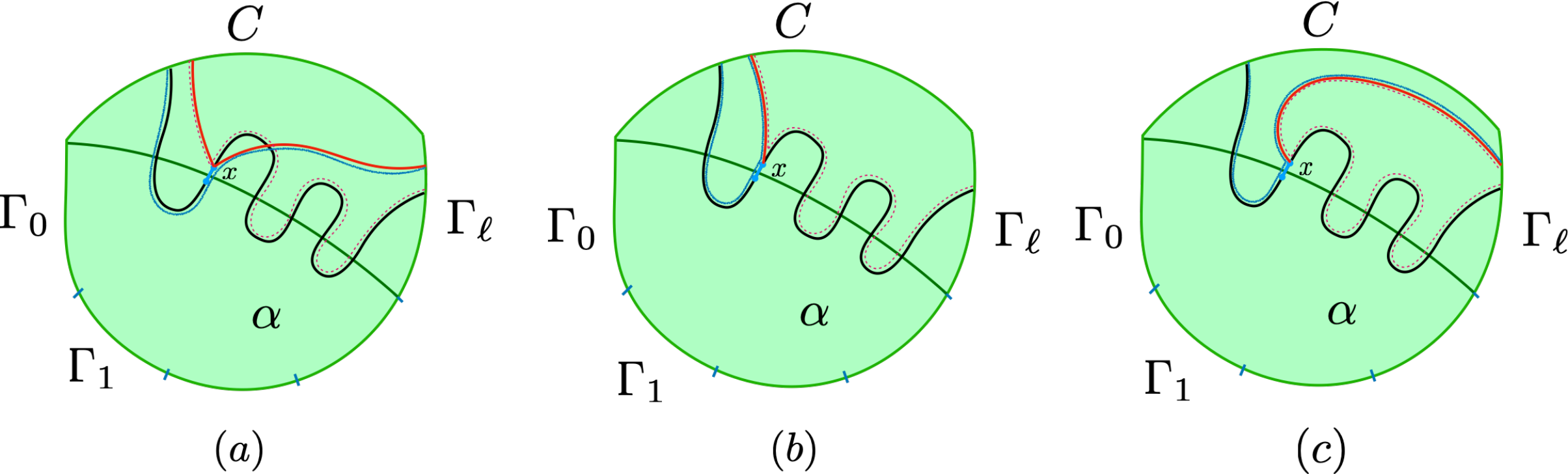}
\caption{We perform surgeries on the $N$-intersecting path $L\in \mathcal{P}_{C,\Gamma_\ell}$ (black solid lines) at the point $x$ using various minimal paths (red solid lines) and reduce $L$ to two paths -- the blue solid lines, which intersect $\mu(\alpha)$ two times, and red dashed lines, which intersect $\mu(\alpha)$ $N-2$ times.}
\label{fig:typeI}
\end{figure}
That is we start with $L$ and $L_{x,C} \cup L_{x,\Gamma_\ell}$ and end up with two paths in $\mathcal{P}_{C,\Gamma_\ell}$. These later curves intersect $\mu(\alpha)$ two times and $N-2$ times respectively. Thus:
\begin{align}
 \tilde{\varrho}(L)  + (\ell-1)/2 &= \tilde{\varrho}(L) + \tilde{\varrho}(L_{x,C} \cup L_{x,\Gamma_\ell})  \\&\geq  \big((\ell-1)/2 + 1\big) + \big((\ell-1)/2 + (N-2)/2 \big) 
\end{align}
Thus $ \tilde{\varrho}(P)  \geq (\ell-1)/2 + N/2 $ as required. 

(b) If $d_{\tilde{\varrho}}(x,C) = 0$, we instead perform surgery with two copies of this minimal path. It is clear this path remains inside $\alpha^c$. These paths do not cost anything. 
See the second figure in \figref{fig:typeI} for the pattern.
In particular we find, after surgery, a curve that starts and ends in $C$
and that intersects $\mu(\alpha)$ twice, and
a curve in $\mathcal{P}_{C,\Gamma_\ell}$ intersecting $N-2$ times. For the former curve we use the estimate \Eqref{obvious} and find:
\begin{equation}
\tilde{\varrho}(P) \geq \big( 1 \big) + \big( (\ell-1)/2 + (N-2)/2 \big) = (\ell-1)/2  + N/2
\end{equation} 
as required.

(c)  If $d_{\tilde{\varrho}}(x,\Gamma_\ell) =0$ then we again use two copies of this minimal path to perform surgery:
see the third figure in \figref{fig:typeI} for the pattern.

The right hand side is now a path $\mathcal{P}_{C,\Gamma_\ell}$ with two intersections and a path from  $\mathcal{P}_{\Gamma_\ell,\Gamma_\ell}$ with $N-2$ intersections.
Thus:
\begin{equation}
\tilde{\varrho}(L) \geq \big( (\ell-1)/2+1 \big) + \big( (N-2)/2 \big) = (\ell-1)/2  + N/2
\end{equation} 
where we used \Eqref{torepeat} and \Eqref{obvious}. We have completed the induction step. 

\item[(II)] Consider a path $L \in \mathcal{P}_{C,\Gamma_k}$ for some $ 0 \leq k \leq \ell-1$ and intersecting $N$ times with $\mu(\alpha)$, where $N \geq 3$ is odd. We aim to show that $\tilde{\varrho}(L) \geq (\ell-1)/2 + N/2$.  Consider the \emph{last} edge $e = \{x,y\}$ that intersects $\mu(\alpha)$ along the path $C \rightarrow y \rightarrow x \rightarrow \Gamma_k$, where $x \in \alpha^c$. There are again three cases (a,b,c) to consider:

\begin{figure}[h]
\centering
\includegraphics[scale=.35]{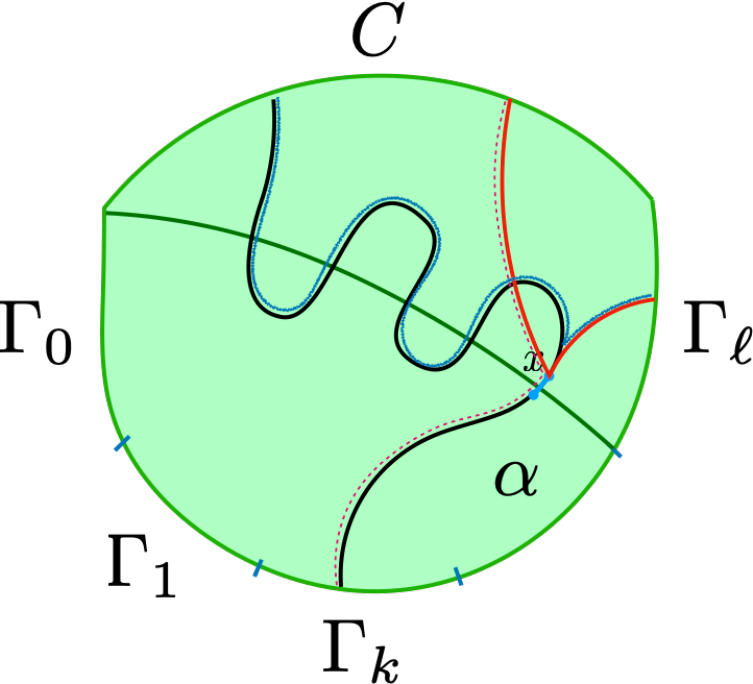}
\caption{We perform a surgery that turns a path $L\in\mathcal{P}_{C,\Gamma_k}$ of type-II (a) (black solid curve) to a type-I path in $\mathcal{P}_{C,\Gamma_\ell}$ (blue solid) and a path in $\widetilde{\mathcal{P}}_{C,\Gamma_k}$ (red dashed) that only intersects $\mu(\alpha)$ one time.}
\label{fig:typeII}
\end{figure}
(a)  If $d_{\tilde{\varrho}}(x,C) + d_{\tilde{\varrho}}(x,\Gamma_\ell) = (\ell-1)/2$, we again perform surgery as shown in \figref{fig:typeII}.
from which we arrive at a path in $\mathcal{P}_{C,\Gamma_\ell}$ intersecting $N-1$ times and
a path  in $\mathcal{P}_{C,\Gamma_k}$ with one intersection. We dealt with the later path at the start and the former we have already bounded in (I) above. Thus:
\begin{align}
 \tilde{\varrho}(L)  + (\ell-1)/2 \geq  \big( (\ell-1)/2 + (N-1)/2 \big) + \big( \ell/2 \big)
\end{align}
implying  $\tilde{\varrho}(L) \geq (\ell-1)/2 + N/2$. 

(b)  If $d_{\tilde{\varrho}}(x,C) = 0$, we take these minimal paths and join in the obvious way to find a path in $\mathcal{P}_{C,C}$ with $N-1$ intersections and one in 
$\mathcal{P}_{C,\Gamma_k}$ with one intersection. Thus:
\begin{equation}
\tilde{\varrho}(L) \geq \big( (N-1)/2 \big) + \big( \ell/2 \big) = (\ell-1)/2  + N/2
\end{equation} 

(c)  If $d_{\tilde{\varrho}}(x,\Gamma_\ell) =0$  we use these minimal paths to construct a path in $\mathcal{P}_{\Gamma_k, \Gamma_\ell}$ with one intersection
and one in  $\mathcal{P}_{C, \Gamma_\ell}$ with $N-1$ intersections. Hence:
\begin{equation}
\tilde{\varrho}(L) \geq \big( \ell-k-1/2 \big) + \big( (\ell-1)/2 + (N-1)/2 \big) \geq  (\ell-1)/2  + N/2
\end{equation} 
where we used $ \ell -k -1 \geq 0$. And we claim victory for these paths.

\item[(III)] Consider a path $L \in \mathcal{P}_{\Gamma_k,\Gamma_\ell}$ for some $ 0 \leq k \leq \ell-1$ and intersecting $N$ times with $\mu(\alpha)$, where $N \geq 3$ is odd. 
We now consider the second edge $e = \{x,y\}$ along the path from $L : \Gamma_k \rightarrow x \rightarrow y \rightarrow \Gamma_\ell$
with $ x \in \alpha^c$. As usual, there are three cases:

\begin{figure}[h]
\centering
\includegraphics[scale=.35]{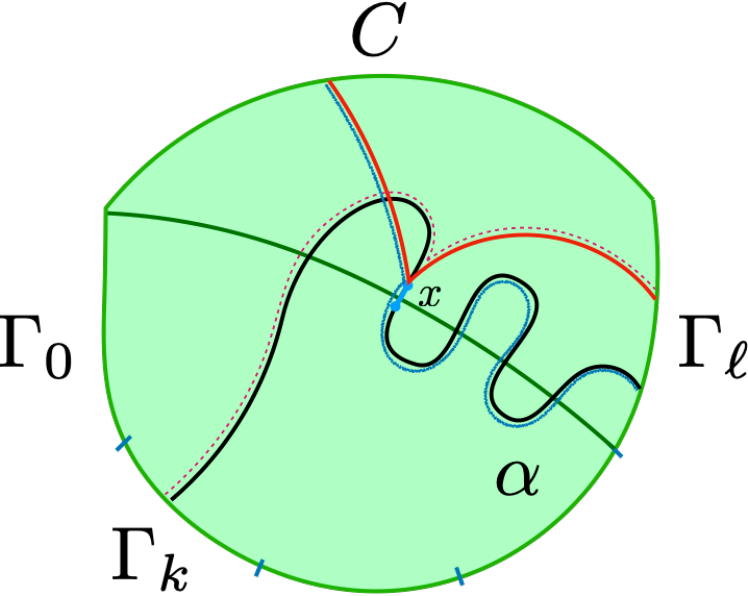}
\caption{We perform a surgery that turns a path $L\in\mathcal{P}_{k,\Gamma_\ell}$ of type-III (a) (black solid curve) to a type-II path in $\mathcal{P}_{C,\Gamma_\ell}$ (blue solid) and a path in $\widetilde{\mathcal{P}}_{\Gamma_k,\Gamma_\ell}$ (red dashed) that only intersects $\mu(\alpha)$ one time.}
\label{fig:typeIII}
\end{figure}
(a)  If $d_{\tilde{\varrho}}(x,C) + d_{\tilde{\varrho}}(x,\Gamma_\ell) = (\ell-1)/2$, we again perform surgery as shown in \figref{fig:typeIII}.
The resulting paths are in $\mathcal{P}_{\Gamma_k,\Gamma_\ell}$ with one intersections
and and in $\mathcal{P}_{C, \Gamma_\ell}$ with $(N-1)$ intersections.  Thus:
\begin{align}
 \tilde{\varrho}(L)  + (\ell-1)/2 \geq  \big( \ell-k-1/2 \big) + \big( (\ell-1)/2 + (N-1)/2  \big)
\end{align}
Or $ \tilde{\varrho}(L) \geq (\ell-k-1) + N/2$ as required. 

(b) Now if $d_{\tilde{\varrho}}(x,C) = 0$, our surgery results in a path in $\mathcal{P}_{C,\Gamma_k}$ with one intersection
and a path in $\mathcal{P}_{C,\Gamma_\ell}$ with $N-1$ intersections. Hence:
\begin{equation}
\tilde{\varrho}(L) \geq \big(  \ell/2 \big) + \big( (\ell-1)/2 + (N-1)/2 \big) = \ell-1  + N/2 \geq  (\ell-k-1) + N/2
\end{equation} 

(c) Now if $d_{\tilde{\varrho}}(x,\Gamma_\ell) = 0$, our surgery results in a path in $\mathcal{P}_{\Gamma_k,\Gamma_\ell}$ with one intersection
and a path  in  $\mathcal{P}_{\Gamma_\ell,\Gamma_\ell}$ with $N-1$ intersections. Hence:
\begin{equation}
\tilde{\varrho}(L) \geq \big( \ell-k-1/2 \big) + \big(  (N-1)/2 \big) =  (\ell-k-1) + N/2
\end{equation} 

\item[(IV)] Consider a path $L \in \mathcal{P}_{\Gamma_k,\Gamma_{k'}}$ for some $ 0 \leq k < k' \leq \ell-1$ and intersecting $N$ times with $\mu(\alpha)$, where $N \geq 2$ is even. 
We now consider the second edge $e = \{x,y\}$ along the path from $L : \Gamma_k \rightarrow x \rightarrow y \rightarrow \Gamma_{k'}$
with $ x \in \alpha^c$. As usual there are three cases:

\begin{figure}[h]
\centering
\includegraphics[scale=.35]{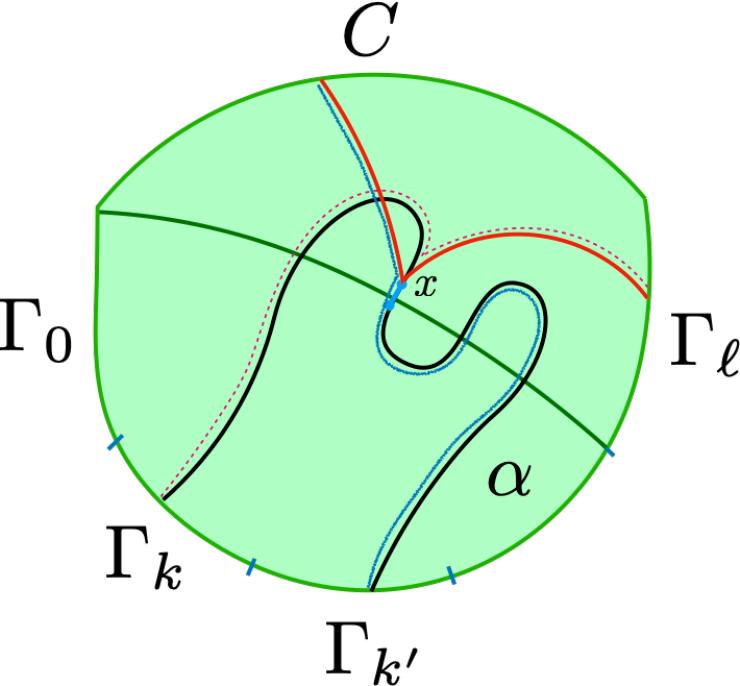}
\caption{We perform a surgery that turns a path $L\in\mathcal{P}_{k,\Gamma_k'}$ of type-IV (a) (black solid curve) to a type-II path in $\mathcal{P}_{C,\Gamma_k'}$ (blue solid) and a path in $\widetilde{\mathcal{P}}_{\Gamma_k,\Gamma_k'}$ (red dashed) that only intersects $\mu(\alpha)$ one time.}
\label{fig:typeIV}
\end{figure}
(a)  If $d_{\tilde{\varrho}}(x,C) + d_{\tilde{\varrho}}(x,\Gamma_\ell) = (\ell-1)/2$,
surgery results in a path in $\mathcal{P}_{\Gamma_k,\Gamma_\ell}$ with one intersection and another path in $\mathcal{P}_{C,\Gamma_{k'}}$ with $N-1$ intersections, see \figref{fig:typeIV}. Thus, using (II) we have:
\begin{align}
 \tilde{\varrho}(L)  + (\ell-1)/2 \geq  \big( \ell-k-1/2 \big) + \big( (\ell-1)/2 + (N-1)/2  \big)
\end{align}
implying $ \tilde{\varrho}(L) \geq (\ell-k-1)+ N/2 \geq (k' -k) +N/2$ where we used $\ell-1 \geq k'$. This is the required bound. 

(b) Now if $d_{\tilde{\varrho}}(x,C) = 0$, our surgery results in a path in $\mathcal{P}_{\Gamma_k,C}$ with one intersection and a path in $\mathcal{P}_{\Gamma_{k'},C}$
with $(N-1)$ intersections. Thus:
\begin{equation}
\tilde{\varrho}(L) \geq \big(  \ell/2 \big) + \big( (\ell-1)/2 + (N-1)/2 \big) = \ell-1  + N/2 \geq (k'-k) + N/2
\end{equation} 
where we again used (II). 

(c) Now if $d_{\tilde{\varrho}}(x,\Gamma_\ell) = 0$, our surgery results in a path in $\mathcal{P}_{\Gamma_k,\Gamma_\ell}$ with one intersection
and a path in $\mathcal{P}_{\Gamma_\ell,\Gamma_{k'}}$ with $N-1$ intersections. Hence:
\begin{equation}
\tilde{\varrho}(L) \geq \big( \ell-k-1/2 \big) + \big( \ell-k'-1 +(N-1)/2 \big) \geq (k'-k) +N/2
\end{equation} 
where we used (III) and $\ell-k'-1 \geq -(\ell-k'-1)$. And we are done.
\end{itemize}

The above bound establish feasibility of $\varrho'$ for the $\ell-1$ intersecting cut problem.
(The integer gap conditions are again all automatic because the even/oddness of the number of
intersections is fixed by the topology of the path.)
\end{proof}

\subsection{Lemma~\ref{conv:prob}: 
\label{app:conv}
Probabilistic convergence for $S_R$}
Consider the renormalized operator:
\begin{equation}
 \widehat{\mathcal{O}} = \left( \varrho \otimes \varrho \right) \chi^{  2 n  \mathcal{A}(AB:C)  } 
\end{equation}
and define the renormalized measure:
\begin{equation}
d \hat{\mu}_\Psi(\hat{\lambda}) = \chi^{ 2 (n-1) \mathcal{A}(A:B:C) -2 n \mathcal{A}(AB:C) }  \sum_{i} |\left< \Psi \right| \left. v_i \right>|^2 \delta( \hat{\lambda} - \hat{\lambda}_i) d \hat{\lambda}
\end{equation}
where $\hat{\lambda} = \lambda  \chi^{  2 n  \mathcal{A}(AB:C)  } $.
We know that new measure satisfies
\begin{equation}
\label{eq:mu_hat_norm}
\lim_{\chi \rightarrow \infty} \overline{ \int_0^\infty d \hat{\mu}_\Psi(\hat{\lambda}) \hat{\lambda}^{m/2} } = 1
\end{equation}
for $m/2 \in \mathbb{Z}_{\geq 1}$. Note that we do not know the zeroth moment of $\hat{\lambda}$.

Pick some cut-off $\Lambda > 1$ and define
\begin{equation}
\label{eq:Delta_k}
\Delta_{m/2} = \overline{ \int d\hat{\mu}_\Psi(\hat{\lambda})  \hat{\lambda}^{m/2} \theta(\hat{\lambda} -\Lambda)}  \leq \Lambda^{-m/2}  \overline{  \int d\hat{\mu}_\Psi(\hat{\lambda})  \hat{\lambda}^{m} \theta(\hat{\lambda} -\Lambda)} \leq \Lambda^{-m/2}   \overline{ \int d\hat{\mu}_\Psi(\hat{\lambda}) \hat{\lambda}^{m} } 
\end{equation}
From now on we will set $\Lambda =2$. Then we have:
\begin{equation}
\label{eq:Delta_m}
\lim_{\chi \rightarrow \infty} \Delta_{m/2} \leq  \lim_{\chi \rightarrow \infty}  \Delta_{m'/2} \leq 2^{-m'/2}  
\end{equation}
for all $m'\in \mathbb{Z}$ and $m' \geq m \ge 1$. Taking $m' \rightarrow \infty$ proves that $\lim_{\chi \rightarrow \infty} \Delta_{m/2}=0$ for $m\in \mathbb{Z}_{\ge 1}$.
In particular, for any polynomial function $f(\hat{\lambda})$ we have that
\begin{equation}
\label{eq:vanishing}
    \overline{\int_2^\infty d\hat{\mu}_{\Psi}(\hat{\lambda}) f(\hat{\lambda})} 
    = \sum_{k\in \mathbb{N}} \frac{f^{(k)}(0)}{k!} \Delta_k 
    \underset{\chi\to\infty}{\to} 0
\end{equation}
In other words, the measure $\hat{\mu}_\Psi(\hat{\lambda})$ is highly concentrated in the interval $[0,2]$ and it suffices to only consider test functions with compact support on the interval.

Let us approximate the square root function on $[0,2]$ as:
\begin{equation}
p_M(x) \equiv \sum_{\mu=1}^M \sqrt{\frac{2\mu}{M}} b_{\mu,M}(x/2)
\end{equation}
where $b_{\mu,M}(x)$ are the Bernstein polynomials of degree $M$.
Since only $b_{0,M}(x)$ has the constant monomial in it, we only need the higher moments. For all $\delta > 0$ there exists some integer $M$ such that
\begin{equation}
    \|\sqrt{x} - p_M(x)\| _{L^\infty[0,2]} < \delta
\end{equation}
Now we write
\begin{align}
\label{eq:C123}
\begin{split}
\left| \left( \int_0^\infty d\hat{\mu}_\Psi(\hat{\lambda}) \hat{\lambda}^{1/2} \right)  - 1 \right| & 
\leq  \underbrace{\int_2^\infty d\hat{\mu}_\Psi(\hat{\lambda})\hat{\lambda}^{1/2}}_{C_1}  
+ \underbrace{\left| \left( \int_0^2 d\hat{\mu}_\Psi(\hat{\lambda}) p_M(\hat{\lambda}) \right)  - 1 \right|}_{C_2}  \\
 &\quad + \underbrace{\left| \int_0^2 d\hat{\mu}_\Psi(\hat{\lambda}) \left(\hat{\lambda}^{1/2}  - p_M(\hat{\lambda}) \right) \right|}_{C_3}
\end{split}
\end{align}
Using Markov inequality along with \Eqref{eq:Delta_m} we can show that for any integer $m$,
\begin{equation}\label{eq:Markov}
{\rm Pr}\left(C_1 \geq \epsilon \right) 
\leq \frac{\Delta_{m} }{\epsilon}
\underset{\chi\to\infty}{\to} 0
\end{equation}
For $C_2$ we consider:
\begin{align}
C_2 \leq \underbrace{\left| \left( \int_0^2 d\hat{\mu}_\Psi(\hat{\lambda}) p_M(\hat{\lambda}) \right)  - \overline{  \left( \int_0^2 d\hat{\mu}_\Psi(\hat{\lambda}) p_M(\hat{\lambda}) \right)} \right|}_{C^{\rm first}_2}
 +  \underbrace{\left|  \overline{  \left( \int_0^2 d\hat{\mu}_\Psi(\hat{\lambda}) p_M(\hat{\lambda}) \right)} -1 \right|}_{C_2^{\rm second}}
\end{align}
We can bound the first term using Chebyshev's inequality:
\begin{align}
    \text{Pr}\left(C_2^{\rm first}\ge \epsilon\right) \le \frac{\sigma^2}{\epsilon^2}
\end{align}
where 
\begin{align}
\label{eq:c2_var}
    \sigma^2 = \text{Var}(C_2^{\rm first})
    \approxeq \text{Var}\left( \int^\infty_0 d\hat{\mu}_\Psi(\hat{\lambda}) p_M(\hat{\lambda}) \right) 
    = \sum_{m=1}^{M} p_m \text{Var}\left( \int^\infty_0 d\hat{\mu}_\Psi(\hat{\lambda}) \hat{\lambda}^m \right)      
\end{align}
where $p_m$ are Taylor coefficients of $p_M(\lambda)$.
Note that we have extended the integration limit in \Eqref{eq:c2_var}. The error from doing so can be shown to vanish in the limit $\chi\to\infty$ by application of \Eqref{eq:vanishing}. The expectation values of the double moments are related to 
 \begin{equation}
      \overline{ \left( \int^\infty_0 d\hat{\mu}_\Psi(\hat{\lambda}) \hat{\lambda}^m \right)^2 }
      = \chi^{(4n-1)\mathcal{A}(A:B:C)-4n\mathcal{A}(AB:C)}
      \overline{\bra{\Psi}^{\otimes 2}\widehat{O}^{m}\otimes \widehat{O}^{m} \ket{\Psi}^{\otimes 2}}
 \end{equation}
 which can be computed by a different symmetry group optimization problem defined on  $G$. We now minimize over $g\in S_{2mn}$ in this graph such that the boundary conditions are $\tilde{g}_A\equiv g^{(1)}_A g^{(2)}_A, ~ \tilde{g}_B\equiv g^{(1)}_B g^{(2)}_B$ and $\tilde{g}_C =\id$, where $g^{(1)}_{A,B}$ permutes the first $mn$ replicas and leaving the second $mn$ copies invariant; whereas $g^{(2)}_{A,B}$ permutes the second $mn$ replicas and leaving the first invariant. Our analysis in \secref{sec:main} largely carries over. 
 The main difference is that we must now coarse-grain using the new element $\tilde{X} \equiv \tilde{g}_A\wedge \tilde{g}_B = X^{(1)}X^{(2)}$ with $X^{(i)}$ defined similarly as above.
 We need the following generalization of Lemma~\ref{lem:tAtB}:
 \begin{lemma}
 \label{lem:P4n_trig}
     For any $q\in P_{4n}$ we have
     \begin{equation}
         d(\tilde{q}_A,q)+d(\tilde{q}_B,q) \ge d(\tilde{q}_A,\tilde{q_B}) + 2(1-\delta_{\#^{(1)}_1(q),0}) + 2(1-\delta_{\#^{(2)}_1(q),0}) +
         2\delta_{q\vee \tau, \mathbb{Z}_{4n}}
     \end{equation}
     where $\tilde{q}_{A,B}=q_{\tilde{X}}(\tilde{g}_{A,B})$, $\#^{(1,2)}_1(\cdot)$ counts the number of singlets in the first (second) sets of $2n$ elements, and $\tau=\{\mathbb{Z}_{2n},\mathbb{Z}_{2n}\}$ is the maximal element in $P_{4n}$ that is disconnected between the two sets of $2n$ elements.
 \end{lemma}
 \begin{proof}
     If $q\in P_{2n}\times P_{2n}$ then $\delta_{q\vee\tau,\mathbb{Z}_{4n}}=0$ and we can simply break down the problem into two smaller problems on disconnected copies of $2n$ elements. Applying Lemma \ref{lem:tAtB} on each copy proves the result.

     Now suppose that $q\in P_{4n}\backslash P_{2n}\times P_{2n}$. Then there must exists some $u_{ij}\in P_{4n}\backslash P_{2n}\times P_{2n}$ such that $p\vee u_{ij} = p$, where $u_{ij}$ is the unique partition with a doublet connecting element $i$ from the first copy to element $j$ in the second copy and singlets at every other position. We write
     \begin{align}
     \begin{split}
     d(\tilde{q}_A,p)&=d(\tilde{q}_A,p\vee u_{ij})
     = \#(\tilde{q}_A)+\#(p\vee u_{ij})-2\#(\tilde{q}_A\vee p \vee u_{ij}) \\
     &= 1+\#(\tilde{q}_A\vee u_{ij})+\#(p\vee u_{ij})-2\#(\tilde{q}_A\vee p \vee u_{ij}) \\
     &=d(\tilde{q}_A\vee u_{ij},p) + 1
     \end{split}
     \end{align}
     and similarly for $d(\tilde{q}_B,p)$. Thus
     \begin{align}
     \begin{split}
         d(\tilde{q}_A,p)+d(p,\tilde{q}_B) &= d(\tilde{q}_A\vee u_{ij},p) + d(p,\tilde{q}_B\vee u_{ij}) + 2 \\
         &\ge d(\tilde{q}_A\vee u_{ij},\tilde{q}_B\vee u_{ij}) + 2
         \end{split}
     \end{align}
     by triangle inequality.
     This bound can be strengthen in a similar fashion as the proof in Lemma~\ref{lem:tAtB}. 
     The biparpite graph of $\tilde{q}_A\vee u_{ij}$ and $\tilde{q}_B \vee u_{ij}$ is now connected with two cycles, each corresponding to a $2n$-element copy.
     A singlet in the first copy will break the first cycle and leads to a enhancement of the bound by $2$, and likewise for a singlet in the second copy. Since $d(\tilde{q}_A\vee u_{ij},\tilde{q}_B\vee u_{ij})=d(\tilde{q}_{A},\tilde{q}_B)$ we obtain
     \begin{align}
         d(\tilde{q}_A,p)+d(p,\tilde{q}_B) \ge d(\tilde{q}_A,\tilde{q}_B) + 2(1-\delta_{\#^{(1)}_1(q),0}) + 2(1-\delta_{\#^{(2)}_1(q),0}) + 2
     \end{align}
     And this completes the proof.
 \end{proof}
 Using Lemma~\ref{lem:P4n_trig} we see that we may restrict to the disconnected elements $q\in P_{2n}\times P_{2n}$ (and hence $g\in S_{mn}\times S_{mn}$ as the coarse-graining retains this information) in the optimization problems, since for any path $L\in\mathcal{P}_{A:B}$ and vertex $v$ in the path, any $q(v)\in P_{4n}\backslash P_{2n}\times P_{2n}$ will lead to a stricter bound in the integer program.
 Thus, we find that the optimal value of our problem is simply twice of that in the original problem:
\begin{equation}
    \overline{\left( \int d\mu_\Psi(\lambda) \lambda^{m/2} \right)^2} 
    = \chi^{-4(n-1)\mathcal{A}(A:B:C)+2n(m-2)\mathcal{A}(AB:C)}(1+O(1/\chi))
\end{equation}
And we have that $\sigma^2 = O(1/\chi)$ and $C^{\rm first}_2\overset{\rm Pr}{\rightarrow} 0$ as $\chi\to\infty$.
For the second term we have
\begin{align}
\label{eq:C2_second}
\begin{split}
    C_2^{\rm second} &\le 
    \left| \overline{\left( \int_0^\infty d\hat{\mu}_\Psi(\hat{\lambda}) p_M(\hat{\lambda}) \right)} -1 \right| + \left| \overline{\left( \int_2^\infty d\hat{\mu}_\Psi(\hat{\lambda}) p_M(\hat{\lambda}) \right)} \right| \\
    &\underset{\chi\to\infty}{\longrightarrow}\left| p_M(1)-1 \right|  \le \delta
\end{split}
\end{align}
where we have used \Eqref{eq:mu_hat_norm} to rewrite the moment integrals in the first line, and the bound $| p_M(1) - 1| \leq \| p_M(x) - \sqrt{x} \|_{L^{\infty}[0,2]}=\delta$. The second term in the first line vanishes by \Eqref{eq:vanishing} in the $\chi\to\infty$ limit.

For $C_3$ we write
\begin{align}
\label{eq:C_3}
\begin{split}
C_3 &  \leq
 \int_0^2 d\hat{\mu}_\Psi(\hat{\lambda}) \left| \hat{\lambda}^{1/2}  - p_M(\hat{\lambda}) \right|  \\
 & \leq  \int_0^{2} d\hat{\mu}_\Psi(\hat{\lambda}) \hat{\lambda}^{1/2}  f(\ln\hat{\lambda})
 + \int_{0}^2 d\hat{\mu}_\Psi(\hat{\lambda})(1- f(\ln\hat{\lambda}) )\left| \hat{\lambda}^{1/2}  - p_M(\hat{\lambda}) \right| \\ 
 & \leq \underbrace{\int_0^{\infty} d\hat{\mu}_\Psi(\hat{\lambda}) \hat{\lambda}^{1/2}  f(\ln \hat{\lambda})}_{C_3^{\rm first}} 
 + ~ \underbrace{\delta 
 \int_{0}^2 d\hat{\mu}_\Psi(\hat{\lambda}) (1- f(\ln \hat{\lambda}) )}_{C^{\rm second}_3}
 \end{split}
\end{align}
where passing from the first line to the second we used the fact that $p_M(x)$ approaches $\sqrt{x}$ from below, which follows from the positivity of $p_M(x)$.
We have also introduced a semi-positive function $f$ with $f(x) = 0$ for $x\geq 0$ and whose other properties we will enumerate below. For the second term, we suppose that:
 \begin{equation} 
 \| (1-f(\ln \hat{\lambda})) \hat{\lambda}^{-1} \|_{L^\infty_{[0,2]}}  
  \equiv f_1 < \infty
 \end{equation}
So that
\begin{align}
    C_3^{\rm second} \le f_1 \delta \int^2_0 d\hat{\mu}_{\Psi}(\hat{\lambda}) \hat{\lambda}
    \le f_1 \delta \int^\infty_0 d\hat{\mu}_{\Psi}(\hat{\lambda}) \hat{\lambda}
\end{align}

To proceed we need to following Lemma:
\begin{lemma}
\label{lem:fourierbound}
     Consider a density matrix $\rho$ on a finite dimensional Hilbert space $\mathcal{H}$ and supported on $\pi_\rho$ with ${\rm Tr}\, \pi_{\rho} \equiv \Lambda$ then
for a real positive semidefinite function $ f \in C^\infty(\mathbb{R})$ such that $f(x) = 0$ for all $x \geq 0$
and such that $f'$ is a rapidly decaying Schwartz function $f' \in \mathcal{S}(\mathbb{R})$, then:
\begin{equation}
\label{eq:wmc1}
 \left|\left< \eta_1 \right| f( \ln \rho  + \ln \Lambda) \left| \eta_2 \right>   \right|\leq \| \mathfrak{F}(f')\|_{L^1}  \| \rho - \pi_\rho /\Lambda \|_1 
\end{equation}
with normalized $\eta_{1,2} \in \mathcal{H}$ and where $\mathfrak{F}$ denotes the Fourier transform.
\end{lemma}
\begin{proof}
Since $\rho$ will have a minimum non-zero eigenvalue  $\lambda_{\rm min}$ we can cut off $f_{a}(x) = f(x)  w_a(x)$
where $w_a \in C^{\infty}$ is a smooth
cutoff function:
\begin{equation}
w_a(x) = w( x / a)
\end{equation}
with $w(x) = 0$ for $x < -2$ and $w(x) = 1$ for $x > -1$ and generally $0 \leq w \leq 1$ and also $w' \in \mathcal{S}$. Choosing $a > - \ln \lambda_{\rm min} - \ln \Lambda \geq 0$ we can replace
$f$ with $f_a$. The result $f_a$ is a smooth function of compact support so this has a Fourier transform. Thus:
\begin{equation}
f_a( \ln \rho  + \ln \Lambda)  = 
\int d s  \mathfrak{F}(f_a)(s) \Lambda^{is} \rho^{i s} 
\end{equation}
If $\rho$ and $\rho'$ commute then we can simultaneously diagonalize these and
\begin{align}
\left|\left< \eta_1 \right| (\rho^{is} - (\rho')^{is}) \left| \eta_2 \right> \right| 
& = \left| \sum_{i} (\exp( - is E_i) - \exp( - is E_i')) \left< \eta_1 \right. \left| i \right> \left< i \right. \left| \eta_2 \right> \right|  \\
&\leq \sum_{i } \left| 1 - \exp( - is (E_i' - E_i)) \right| \leq |s| \| \rho - \rho' \|_{1} 
\end{align}
where $E_i\equiv -\ln \lambda_i$, and the inequality follows from the bound $|\sin(x)|  < |x|$.
Thus:
\begin{equation}
\left| \left< \eta_1 \right| f( \ln \rho  + \ln \Lambda) \left| \eta_2 \right>   - \left< \eta_1 \right| f( \ln \rho'  + \ln \Lambda) \left| \eta_2 \right> \right|
\leq \|  s  \mathfrak{F}(f_a)(s) \|_{L^1}
 \| \rho - \rho' \|_{1}
\end{equation}
We will show that we can remove the cutoff function $w_a$, i.e.:
\begin{equation}
\label{claimedlimit}
\lim_{a \rightarrow \infty} \|  s  \mathfrak{F}(f_a)(s) \|_{L^1} = \|   \mathfrak{F}(f')(s) \|_{L^1}
\end{equation}
for functions $f$ and $w$ that were specified in the statement and above. 
Following \Eqref{claimedlimit}, if we set $\rho' = \pi_{\rho}/\Lambda$ and use the properties of the function in the statement we find $f( \ln \rho'  + \ln \Lambda) =0$ away from the subspace of support of $\rho$, and thus we have proved our claim.

To prove \Eqref{claimedlimit} we write
\begin{equation}
i s  \mathfrak{F}(f_a)(s) = \mathfrak{F}(f_a')(s) = \mathfrak{F}( f' w_a) (s)  +  \mathfrak{F}(f w_a') (s)
\end{equation}
We have two remainder terms to analyze: 
\begin{equation}
\left| \| s  \mathfrak{F}(f_a)(s) \|_{L^1} - \| \mathfrak{F}( f')  \|_{L^1} \right|  \leq \| \mathfrak{F}( f' w_a) \|_{L^1} +  \|  \mathfrak{F}( f' (1-w_a)) \|_{L^1}
\end{equation}

For the second term:
\begin{align}
\label{eq:fourierbound}
\begin{split}
 \|  \mathfrak{F}( f' (1-w_a)) \|_{L^1} & \leq \| (1+s^2)^{-1} \|_{L^1}  \| (1+s^2)  \mathfrak{F}( f' (1-w_a))(s) \|_{L^{\infty}} \\
& \leq \pi( \|  f' (1-w_a) \|_{L^1} + \|  (f' (1-w_a))'' \|_{L^1} ) \\
& \leq \pi \|\tfrac{1-w_a}{1+x^2} \|_{L^1} \left( \|f'(1+x^2)\|_{L^\infty} +  \|f'''(1+x^2)\|_{L^\infty} \right) \\  
& \qquad \qquad + \pi \left(2 \| f'' \|_{L^1}  \|w_a' \|_{L^\infty}   +  \| f' \|_{L^1}  \|w_a'' \|_{L^\infty} \right) \\
& \leq \pi^2 \big( \sup_{x\le -a}|f'(x)(1+x^2)| +  \sup_{x\le -a}|f'''(x)(1+x^2)| \big) \\  
& \qquad \qquad + \pi \left(2 \| f'' \|_{L^1}  \|w_a' \|_{L^\infty}   +  \| f' \|_{L^1}  \|w_a'' \|_{L^\infty} \right)
\end{split}
\end{align}
where passing to the second line we used the Hausdorff-Young inequality and we also used H\"older's inequality throughout.
But $ \|w_a' \|_{L^\infty}  = (1/a)  \|w' \|_{L^\infty} \rightarrow 0$ and similarly $\|w_a'' \|_{L^\infty}  = (1/a^2)  \|w'' \|_{L^\infty} \rightarrow 0$
as $a \rightarrow \infty$ and also:
\begin{equation}
 \sup_{x\leq - a} | (1+x^2) f'(x) | < a^{-2} \sup_{x\leq - a} | x^2 (1+x^2) f'(x) | < a^{-2} \|  x^2 (1+x^2) f'(x) \|_{L^\infty} \rightarrow 0
\end{equation}
similarly for the $f'''(x)$ term.

For the first term use:
\begin{align}
\begin{split}
\mathfrak{F}(f w_a')(s) &= \mathfrak{F}( (f-c_a) w_a')(s)  + c_a \int_{-\infty}^{\infty} d x e^{i x s} w_a'(x)  \\
& = \mathfrak{F}( (f-c_a) w_a')(s)  +  c_a  \int_{-\infty}^{\infty} d x e^{i x s a} w'(x)  \\
& = \mathfrak{F}( (f-c_a) w_a')(s)  +  c_a \mathfrak{F}( w')(s a)
\end{split}
\end{align}
for some constant $c_a$. 
We now pick $c_a =   \inf_{x < -a} f(x)$. Using the same strategy as above we write 
\begin{align}
\begin{split}
\| \mathfrak{F}( (f -c_a) w_a')\|_{L^1} & \leq \pi ( \|  (f-c_a) w_a' \|_{L^1} + \|  ((f-c_a) w_a')'' \|_{L^1} ) \\ &\leq \pi ( \|f- c_a\|_{L^\infty_{(-\infty,-a]}}  \| w_a' \|_{L^1} +  \|(f-c_a) \|_{L^\infty_{(-\infty,-a]}}  \| w_a''' \|_{L^1})
\\ & \qquad \qquad +  \pi\| f'' \|_{L^1}  \|w_a' \|_{L^\infty}   +  2\pi \| f' \|_{L^1}  \|w_a'' \|_{L^\infty} 
\end{split}
\end{align}
The last two terms are dealt with as above. Note that:
\begin{align}
\begin{split}
\|f - C \|_{L^\infty} & \leq \|f \|_{L^\infty} + C   =C + \sup_x \left| f(b) + \int_b^x dy f'(y) \right|
 \\ & \leq C+ |f(b)| + \sup_{x}  \int_b^x dy |f'(y)| \leq C+ |f(b)| + \| f' \|_{L^1}
 \end{split}
\end{align}
which is clearly finite. This analysis implies that $c_a$ is finite and $\|f - c_a \|_{L^\infty}$ is finite. Thus we need to compute:
\begin{equation}
 \| w_a' \|_{L^1}  = \| w' \|_{L^1} \, \qquad \| w_a''' \|_{L^1}  =\frac{1}{a^2} \| w''' \|_{L^1}
\end{equation}
the latter of which vanishes. Since the first term does not vanish we instead note that:
\begin{align}
\begin{split}
\sup_{x < -a} |f(x)- c_a| &= \sup_{x< -a} f(x) - \inf_{x < -a} f(x)
\approx f(x_s) - f(x_i) \\ &= \int_{x_i}^{x_s} d x f'(x)   \leq  \int_{x_i}^{x_s} d x | f'(x)| \leq \int_{-\infty}^{a} d x | f'(x)| \\
& \leq a^{-2} \int_{-\infty}^{a} d x x^2 | f'(x)| \leq a^{-2} \| x^2 f'(x) \|_{L^1} \rightarrow 0
\end{split}
\end{align}
where $x_i$ and $x_s$ approximate the location of the inf and sup respectively. This approximation is what we mean by $\approx$ and this can be removed after taking limits. 

All that is left to do is to compute:
\begin{equation}
\| \mathfrak{F}( w')( a \cdot ) \|_{L^1} = \int_{-\infty}^{\infty} ds | \mathfrak{F}( w')( a s) | = a^{-1} \| \mathfrak{F}( w') \|_{L^1}  \rightarrow 0
\end{equation}
Thus we have establishing the limit \Eqref{claimedlimit}.
\end{proof}

We now write the first term of \Eqref{eq:C_3} out:
\begin{align}
\begin{split}
&C_3^{\rm first} \chi^{ -2 (n-1) \mathcal{A}(A:B:C) +n \mathcal{A}(AB:C) }   \nonumber\\ &=  \left< 1_{AB}^{\otimes n} \right| \Sigma_A^\dagger (\varrho^{1/2} \otimes  1 ) f( \ln (\varrho \otimes \varrho) + 2 n  \mathcal{A}(AB:C) \ln \chi ) (1 \otimes \varrho^{1/2}) \Sigma_A \left|  1_{AB}^{\otimes n} \right> \\
& \leq \| \mathfrak{F}(f')\|_{L^1}  \| \varrho \otimes \varrho - \pi /\chi^{ 2 n  \mathcal{A}(AB:C)} \|_1  \| (1 \otimes \varrho^{1/2}) \Sigma_A \left|  1_{AB}^{\otimes n} \right>  \|   \| ( \varrho^{1/2} \otimes 1) \Sigma_A \left|  1_{AB}^{\otimes n} \right>  \| 
\end{split}
\end{align} 
where we have applied \Eqref{eq:wmc1} in the second line.
Note that $ \| (1 \otimes \varrho^{1/2}) \Sigma_A \left|  1_{AB}^{\otimes n} \right>  \|=\|( \varrho^{1/2} \otimes 1) \Sigma_A \left|  1_{AB}^{\otimes n} \right>\|=1$.
Setting $f_2 \equiv  \| \mathfrak{F}(f')\|_{L^1} $ we have:
\begin{align} 
\begin{split}
C_3 &\leq f_1 \delta\int^\infty_0 d\hat{\mu}_\Psi (\hat{\lambda}) +  f_2 \| \varrho \otimes \varrho - \pi /\chi^{ 2 n  \mathcal{A}(AB:C)} \|_1  \int_0^\infty d\hat{\mu}_\Psi(\hat{\lambda}) \hat{\lambda}^{1/2}  \\
&\leq  f_1  \delta\int^\infty_0 d\hat{\mu}_\Psi (\hat{\lambda}) +  f_2 \| \varrho \otimes \varrho - \pi /\chi^{ 2 n  \mathcal{A}(AB:C)} \|_1 \left( 1 + \left| \left( \int_0^\infty d\hat{\mu}_\Psi(\hat{\lambda}) \hat{\lambda}^{1/2} \right)  - 1 \right|  \right)
\end{split}
\end{align}
Since the same quantity appears on the left hand side of \Eqref{eq:C123} we should subtract and write:
\begin{equation}
\left| \left( \int_0^\infty d\hat{\mu}_\Psi(\hat{\lambda}) \hat{\lambda}^{1/2} \right)  - 1 \right|  \left ( 1 -   f_2 \| \varrho \otimes \varrho - \pi /\chi^{ 2 n  \mathcal{A}(AB:C)}   \|_1 \right)
\leq C_1 + C_2 + C_3'
\end{equation} 
with
\begin{equation}
\label{eq:C3'}
C_3' \leq f_1 \delta\int^\infty_0 d\hat{\mu}_\Psi (\hat{\lambda}) +  f_2 \| \varrho \otimes \varrho - \pi /\chi^{ 2 n  \mathcal{A}(AB:C)} \|_1 \underset{\chi\to\infty}{\to} f_1\delta
\end{equation}
where we have used \Eqref{eq:wmc2} in the limit.
Now suppose that  $0 < f_2 < 1/2$ then using the fact that the trace distance is bounded by $2$:
\begin{equation}
\left| \left( \int_0^\infty d\hat{\mu}_\Psi(\hat{\lambda}) \hat{\lambda}^{1/2} \right)  - 1 \right|
\leq \frac{C_1 + C_2 + C_3'}{1 - 2 f_2}
\end{equation}
Since $\delta$ can be made arbitrarily small, we have $C_1\overset{Pr}{\to}0$ by \Eqref{eq:Markov}, $C_2\overset{Pr}{\to}0$ by Chebyshev's inequality and the vanishing of the variance of moments and \Eqref{eq:C2_second}, and $C'_3\overset{Pr}{\to}0$ by \Eqref{eq:C3'}.
The rest of the proof is fairly standard and we find:
\begin{equation}
\chi^{ 2 (n-1) \mathcal{A}(A:B:C) -n \mathcal{A}(AB:C)}  {\rm Tr}( \rho_{AA^\star}^{(1/2)})^n = \left( \int_0^\infty d\hat{\mu}_\Psi(\hat{\lambda}) \hat{\lambda}^{1/2} \right)  \mathop{\rightarrow}^{Pr} 1
\end{equation}
The map $x \rightarrow -\frac{1}{n-1} \ln x$ is continuous for $x > 0$, which is where the random variable is defined so, so by the continuous
mapping theorem we prove \Eqref{cont:prob}.

We have imposed various properties on $f$ in the above proof. To finish we must show there exists a function with these properties.
We desire:
\begin{enumerate}
\item $f\in C^\infty(\mathbb{R})$ with $f\ge 0$ and $f(x) = 0$ for $x \geq 0$.
Also $f' \in \mathcal{S}(\mathbb{R})$.
\item 
$f_1 \equiv \|(1-f(\ln x))x^{-1}\|_{L^{\infty}_{[0,2]}}=\sup_{x\leq  0} | (1- f(x) ) e^{-x} |< \infty.$
\item $f_2 = \| \mathfrak{F}(f') \|_{L^1} < 1/2.$
\end{enumerate}

We will show existence by explicit construction. Consider a smooth bump function $b\ge 0$ compactly supported on $(-1,0)$ and set
\begin{equation}
    f(x) = \frac{1}{L\|b\|_{L^1}} \int^0_x dx\, b(x/L)
\end{equation}
for some constant $L>0$ to be determined. 
We now check: 1. is trivially satisfied. 2. We have $f(x)=1$ for $x<-L$ and
\begin{align}
    f_1 = \frac{1}{L\|b\|_{L^1}}\sup_{-L < x \le 0} \left| e^{-x}\int^x_{-L} dx \,\sigma(x) \right|  < \infty
\end{align}
3. We calculate
\begin{align}
    f_2 = \frac{\|\mathfrak{F}(b(x/L))(s)\|_{L^1}}{L\|b\|_{L^1}}
    = \frac{\int ds |\int dx e^{-isx}b(x/L)|}{L\|b\|_{L^1}} 
    = \frac{\|\mathfrak{F}(b)\|_{L^1}}{L\|b\|_{L^1}}
\end{align}
Thus to satisfy 3. we simply pick $L>2\|\mathfrak{F}(b)\|_{L^1}/\|b\|_{L^1}$. Thus, we complete the proof.

\bibliographystyle{jhep}
\bibliography{mybibliography}

\providecommand{\href}[2]{#2}\begingroup\raggedright\begin{thebibliography}{10}

\bibitem{Ryu:2006bv}
S.~Ryu and T.~Takayanagi, {\it {Holographic derivation of entanglement entropy
  from AdS/CFT}},  {\em Phys. Rev. Lett.} {\bf 96} (2006) 181602,
  [\href{http://arxiv.org/abs/hep-th/0603001}{{\tt hep-th/0603001}}].

\bibitem{Ryu:2006ef}
S.~Ryu and T.~Takayanagi, {\it {Aspects of Holographic Entanglement Entropy}},
  {\em JHEP} {\bf 08} (2006) 045,
  [\href{http://arxiv.org/abs/hep-th/0605073}{{\tt hep-th/0605073}}].

\bibitem{Hubeny:2007xt}
V.~E. Hubeny, M.~Rangamani, and T.~Takayanagi, {\it {A Covariant holographic
  entanglement entropy proposal}},  {\em JHEP} {\bf 07} (2007) 062,
  [\href{http://arxiv.org/abs/0705.0016}{{\tt arXiv:0705.0016}}].

\bibitem{Freedman:2016zud}
M.~Freedman and M.~Headrick, {\it {Bit threads and holographic entanglement}},
  {\em Commun. Math. Phys.} {\bf 352} (2017), no.~1 407--438,
  [\href{http://arxiv.org/abs/1604.00354}{{\tt arXiv:1604.00354}}].

\bibitem{Headrick:2022nbe}
M.~Headrick and V.~E. Hubeny, {\it {Covariant bit threads}},
  \href{http://arxiv.org/abs/2208.10507}{{\tt arXiv:2208.10507}}.

\bibitem{Collins:2010fsu}
B.~Collins, I.~Nechita, and K.~Zyczkowski, {\it {Random graph states, maximal
  flow and Fuss-Catalan distributions}},
  \href{http://arxiv.org/abs/1003.3075}{{\tt arXiv:1003.3075}}.

\bibitem{Hayden:2016cfa}
P.~Hayden, S.~Nezami, X.-L. Qi, N.~Thomas, M.~Walter, and Z.~Yang, {\it
  {Holographic duality from random tensor networks}},  {\em JHEP} {\bf 11}
  (2016) 009, [\href{http://arxiv.org/abs/1601.01694}{{\tt arXiv:1601.01694}}].

\bibitem{Cui:2018dyq}
S.~X. Cui, P.~Hayden, T.~He, M.~Headrick, B.~Stoica, and M.~Walter, {\it {Bit
  Threads and Holographic Monogamy}},  {\em Commun. Math. Phys.} {\bf 376}
  (2019), no.~1 609--648, [\href{http://arxiv.org/abs/1808.05234}{{\tt
  arXiv:1808.05234}}].

\bibitem{Akers:2019gcv}
C.~Akers and P.~Rath, {\it {Entanglement Wedge Cross Sections Require
  Tripartite Entanglement}},  \href{http://arxiv.org/abs/1911.07852}{{\tt
  arXiv:1911.07852}}.

\bibitem{Susskind:2014rva}
L.~Susskind, {\it {Computational Complexity and Black Hole Horizons}},  {\em
  Fortsch. Phys.} {\bf 64} (2016) 24--43,
  [\href{http://arxiv.org/abs/1403.5695}{{\tt arXiv:1403.5695}}]. [Addendum:
  Fortsch.Phys. 64, 44--48 (2016)].

\bibitem{Brown:2015bva}
A.~R. Brown, D.~A. Roberts, L.~Susskind, B.~Swingle, and Y.~Zhao, {\it
  {Holographic Complexity Equals Bulk Action?}},  {\em Phys. Rev. Lett.} {\bf
  116} (2016), no.~19 191301, [\href{http://arxiv.org/abs/1509.07876}{{\tt
  arXiv:1509.07876}}].

\bibitem{Hayden:2021gno}
P.~Hayden, O.~Parrikar, and J.~Sorce, {\it {The Markov gap for geometric
  reflected entropy}},  \href{http://arxiv.org/abs/2107.00009}{{\tt
  arXiv:2107.00009}}.

\bibitem{doi:10.1137/S0097539792225297}
E.~Dahlhaus, D.~S. Johnson, C.~H. Papadimitriou, P.~D. Seymour, and
  M.~Yannakakis, {\it The complexity of multiterminal cuts},  {\em SIAM Journal
  on Computing} {\bf 23} (1994), no.~4 864--894.

\bibitem{Dutta:2019gen}
S.~Dutta and T.~Faulkner, {\it {A canonical purification for the entanglement
  wedge cross-section}},  \href{http://arxiv.org/abs/1905.00577}{{\tt
  arXiv:1905.00577}}.

\bibitem{Zou:2020bly}
Y.~Zou, K.~Siva, T.~Soejima, R.~S.~K. Mong, and M.~P. Zaletel, {\it {Universal
  tripartite entanglement in one-dimensional many-body systems}},  {\em Phys.
  Rev. Lett.} {\bf 126} (2021), no.~12 120501,
  [\href{http://arxiv.org/abs/2011.11864}{{\tt arXiv:2011.11864}}].

\bibitem{Takayanagi:2017knl}
T.~Takayanagi and K.~Umemoto, {\it {Entanglement of purification through
  holographic duality}},  {\em Nature Phys.} {\bf 14} (2018), no.~6 573--577,
  [\href{http://arxiv.org/abs/1708.09393}{{\tt arXiv:1708.09393}}].

\bibitem{Nguyen:2017yqw}
P.~Nguyen, T.~Devakul, M.~G. Halbasch, M.~P. Zaletel, and B.~Swingle, {\it
  {Entanglement of purification: from spin chains to holography}},  {\em JHEP}
  {\bf 01} (2018) 098, [\href{http://arxiv.org/abs/1709.07424}{{\tt
  arXiv:1709.07424}}].

\bibitem{Akers:2023obn}
C.~Akers, T.~Faulkner, S.~Lin, and P.~Rath, {\it {Entanglement of purification
  in random tensor networks}},  {\em Phys. Rev. D} {\bf 109} (2024), no.~10
  L101902, [\href{http://arxiv.org/abs/2306.06163}{{\tt arXiv:2306.06163}}].

\bibitem{doi.org/10.1002/net.3230210106}
S.~Chopra and M.~R. Rao, {\it On the multiway cut polyhedron},  {\em Networks}
  {\bf 21} (1991), no.~1 51--89.

\bibitem{Schrijver:LP}
A.~Schrijver, {\em Theory of Linear and Integer Programming}.
\newblock John Wiley \& Sons, Inc., USA, 1986.

\bibitem{Kusuki:2019evw}
Y.~Kusuki and K.~Tamaoka, {\it {Entanglement Wedge Cross Section from CFT:
  Dynamics of Local Operator Quench}},
  \href{http://arxiv.org/abs/1909.06790}{{\tt arXiv:1909.06790}}.

\bibitem{Akers:2021pvd}
C.~Akers, T.~Faulkner, S.~Lin, and P.~Rath, {\it {Reflected entropy in random
  tensor networks}},  {\em JHEP} {\bf 05} (2022) 162,
  [\href{http://arxiv.org/abs/2112.09122}{{\tt arXiv:2112.09122}}].

\bibitem{Gadde:2022cqi}
A.~Gadde, V.~Krishna, and T.~Sharma, {\it {New multipartite entanglement
  measure and its holographic dual}},  {\em Phys. Rev. D} {\bf 106} (2022),
  no.~12 126001, [\href{http://arxiv.org/abs/2206.09723}{{\tt
  arXiv:2206.09723}}].

\bibitem{Penington:2022dhr}
G.~Penington, M.~Walter, and F.~Witteveen, {\it {Fun with replicas:
  tripartitions in tensor networks and gravity}},
  \href{http://arxiv.org/abs/2211.16045}{{\tt arXiv:2211.16045}}.

\bibitem{Gadde:2023zzj}
A.~Gadde, V.~Krishna, and T.~Sharma, {\it {Towards classification of
  holographic multi-partite entanglement measures}},
  \href{http://arxiv.org/abs/2304.06082}{{\tt arXiv:2304.06082}}.

\bibitem{Gadde:2023zni}
A.~Gadde, S.~Jain, V.~Krishna, H.~Kulkarni, and T.~Sharma, {\it {Monotonicity
  conjecture for multi-party entanglement. Part I}},  {\em JHEP} {\bf 02}
  (2024) 025, [\href{http://arxiv.org/abs/2308.16247}{{\tt arXiv:2308.16247}}].

\bibitem{Akers:2022zxr}
C.~Akers, T.~Faulkner, S.~Lin, and P.~Rath, {\it {Reflected entropy in random
  tensor networks. Part II. A topological index from canonical purification}},
  {\em JHEP} {\bf 01} (2023) 067, [\href{http://arxiv.org/abs/2210.15006}{{\tt
  arXiv:2210.15006}}].

\bibitem{Cao:2024tog}
X.~Cao and T.~Faulkner, {\it {Ramp from Replica Trick}},
  \href{http://arxiv.org/abs/2405.15873}{{\tt arXiv:2405.15873}}.

\bibitem{MONJARDET1981173}
B.~Monjardet, {\it Metrics on partially ordered sets—a survey},  {\em
  Discrete Mathematics} {\bf 35} (1981), no.~1 173--184. Special Volume on
  Ordered Sets.

\bibitem{birkhoff1967lattice}
G.~Birkhoff, {\em Lattice Theory}.
\newblock American Mathematical Society, Providence, 3rd~ed., 1967.

\bibitem{10.1093/acprof:oso/9780199535255.001.0001}
S.~Boucheron, G.~Lugosi, and P.~Massart, {\em {Concentration Inequalities: A
  Nonasymptotic Theory of Independence}}.
\newblock Oxford University Press, 02, 2013.

\bibitem{Harper:2021uuq}
J.~Harper, {\it {Hyperthreads in holographic spacetimes}},  {\em JHEP} {\bf 09}
  (2021) 118, [\href{http://arxiv.org/abs/2107.10276}{{\tt arXiv:2107.10276}}].

\bibitem{Harper:2019lff}
J.~Harper and M.~Headrick, {\it {Bit threads and holographic entanglement of
  purification}},  {\em JHEP} {\bf 08} (2019) 101,
  [\href{http://arxiv.org/abs/1906.05970}{{\tt arXiv:1906.05970}}].

\bibitem{Bao:2019zqc}
N.~Bao and N.~Cheng, {\it {Multipartite Reflected Entropy}},  {\em JHEP} {\bf
  10} (2019) 102, [\href{http://arxiv.org/abs/1909.03154}{{\tt
  arXiv:1909.03154}}].

\bibitem{Akers:2018fow}
C.~Akers and P.~Rath, {\it {Holographic Renyi Entropy from Quantum Error
  Correction}},  {\em JHEP} {\bf 05} (2019) 052,
  [\href{http://arxiv.org/abs/1811.05171}{{\tt arXiv:1811.05171}}].

\bibitem{Dong:2018seb}
X.~Dong, D.~Harlow, and D.~Marolf, {\it {Flat entanglement spectra in
  fixed-area states of quantum gravity}},  {\em JHEP} {\bf 10} (2019) 240,
  [\href{http://arxiv.org/abs/1811.05382}{{\tt arXiv:1811.05382}}].

\bibitem{Dong:2019piw}
X.~Dong and D.~Marolf, {\it {One-loop universality of holographic codes}},
  {\em JHEP} {\bf 03} (2020) 191, [\href{http://arxiv.org/abs/1910.06329}{{\tt
  arXiv:1910.06329}}].

\bibitem{Bao:2018pvs}
N.~Bao, G.~Penington, J.~Sorce, and A.~C. Wall, {\it {Beyond Toy Models:
  Distilling Tensor Networks in Full AdS/CFT}},  {\em JHEP} {\bf 11} (2019)
  069, [\href{http://arxiv.org/abs/1812.01171}{{\tt arXiv:1812.01171}}].

\bibitem{klappenecker2005mutually}
A.~Klappenecker and M.~Rotteler, {\it Mutually unbiased bases are complex
  projective 2-designs},  in {\em Proceedings. International Symposium on
  Information Theory, 2005. ISIT 2005.}, pp.~1740--1744, IEEE, 2005.

\bibitem{gross2007evenly}
D.~Gross, K.~Audenaert, and J.~Eisert, {\it Evenly distributed unitaries: On
  the structure of unitary designs},  {\em Journal of mathematical physics}
  {\bf 48} (2007), no.~5 052104.

\bibitem{webb2015clifford}
Z.~Webb, {\it The clifford group forms a unitary 3-design},  {\em arXiv
  preprint arXiv:1510.02769} (2015).

\bibitem{zhu2017multiqubit}
H.~Zhu, {\it Multiqubit clifford groups are unitary 3-designs},  {\em Physical
  Review A} {\bf 96} (2017), no.~6 062336.

\bibitem{Nezami:2016zni}
S.~Nezami and M.~Walter, {\it {Multipartite Entanglement in Stabilizer Tensor
  Networks}},  {\em Phys. Rev. Lett.} {\bf 125} (2020) 241602,
  [\href{http://arxiv.org/abs/1608.02595}{{\tt arXiv:1608.02595}}].

\bibitem{Cheng:2022ori}
N.~Cheng, C.~Lancien, G.~Penington, M.~Walter, and F.~Witteveen, {\it {Random
  tensor networks with nontrivial links}},
  \href{http://arxiv.org/abs/2206.10482}{{\tt arXiv:2206.10482}}.

\bibitem{Akers:2024wab}
C.~Akers, R.~M. Soni, and A.~Y. Wei, {\it {Multipartite edge modes and tensor
  networks}},  \href{http://arxiv.org/abs/2404.03651}{{\tt arXiv:2404.03651}}.

\bibitem{Akers:2024ixq}
C.~Akers and A.~Y. Wei, {\it {Background independent tensor networks}},
  \href{http://arxiv.org/abs/2402.05910}{{\tt arXiv:2402.05910}}.

\bibitem{Dong:2023kyr}
X.~Dong, S.~McBride, and W.~W. Weng, {\it {Holographic Tensor Networks with
  Bulk Gauge Symmetries}},  \href{http://arxiv.org/abs/2309.06436}{{\tt
  arXiv:2309.06436}}.

\bibitem{Qi:2022lbd}
X.-L. Qi, {\it {Emergent bulk gauge field in random tensor networks}},
  \href{http://arxiv.org/abs/2209.02940}{{\tt arXiv:2209.02940}}.

\bibitem{Donnelly:2016qqt}
W.~Donnelly, B.~Michel, D.~Marolf, and J.~Wien, {\it {Living on the Edge: A Toy
  Model for Holographic Reconstruction of Algebras with Centers}},  {\em JHEP}
  {\bf 04} (2017) 093, [\href{http://arxiv.org/abs/1611.05841}{{\tt
  arXiv:1611.05841}}].

\bibitem{Colafranceschi:2022dig}
E.~Colafranceschi, S.~Langenscheidt, and D.~Oriti, {\it {Holographic properties
  of superposed quantum geometries}},
  \href{http://arxiv.org/abs/2207.07625}{{\tt arXiv:2207.07625}}.

\bibitem{Bao:2020zgx}
N.~Bao, N.~Cheng, S.~Hern\'andez-Cuenca, and V.~P. Su, {\it {The Quantum
  Entropy Cone of Hypergraphs}},  {\em SciPost Phys.} {\bf 9} (2020), no.~5 5,
  [\href{http://arxiv.org/abs/2002.05317}{{\tt arXiv:2002.05317}}].

\bibitem{Walter:2020zvt}
M.~Walter and F.~Witteveen, {\it {Hypergraph min-cuts from quantum entropies}},
   {\em J. Math. Phys.} {\bf 62} (2021), no.~9 092203,
  [\href{http://arxiv.org/abs/2002.12397}{{\tt arXiv:2002.12397}}].

\bibitem{Dong:2024gud}
X.~Dong, J.~Kudler-Flam, and P.~Rath, {\it {Entanglement Negativity and Replica
  Symmetry Breaking in General Holographic States}},
  \href{http://arxiv.org/abs/2409.13009}{{\tt arXiv:2409.13009}}.

\bibitem{Milekhin:2022zsy}
A.~Milekhin, P.~Rath, and W.~Weng, {\it {Computable Cross Norm in Tensor
  Networks and Holography}},  \href{http://arxiv.org/abs/2212.11978}{{\tt
  arXiv:2212.11978}}.

\bibitem{Yin:2022toc}
C.~Yin and Z.~Liu, {\it {Universal Entanglement and Correlation Measure in
  Two-Dimensional Conformal Field Theories}},  {\em Phys. Rev. Lett.} {\bf 130}
  (2023), no.~13 131601, [\href{http://arxiv.org/abs/2211.11952}{{\tt
  arXiv:2211.11952}}].

\end{thebibliography}\endgroup
 
\end{document}